\documentclass[reqno]{amsart}
\usepackage{extarrows}
\usepackage{bbm}
\usepackage{xcolor}
\usepackage[all]{xy}
\usepackage{amsmath,amssymb}
\usepackage{mathrsfs}
\usepackage{latexsym}
\usepackage{wasysym}
\usepackage{graphicx,epsfig,subfigure}
\usepackage{float}
\usepackage{enumerate}
\usepackage[dvipdfmx,pdfborder={0 0 0},
colorlinks=true, linkcolor=blue, citecolor=red, anchorcolor=green]{hyperref}
\allowdisplaybreaks
\newtheoremstyle{mythm}{}{0pt}{\rm}{}{\bf}{}{8pt}{}
\theoremstyle{mythm}
\newtheorem{Theorem}{Theorem}[section]
\newtheorem{lemma}[Theorem]{Lemma}
\newtheorem{Proposition}[Theorem]{Proposition}
\newtheorem{Corollary}[Theorem]{Corollary}
\newtheorem{remark}[Theorem]{Remark}

\newtheorem{Definition}[Theorem]{Definition}
\newtheorem{assm}{Assumption}[section]

\numberwithin{equation}{section}

\newcommand{\la}{\langle}
\newcommand{\ra}{\rangle}

\newcommand{\beq}{\begin{equation}}
\newcommand{\eeq}{\end{equation}}
\newcommand{\bes}{\begin{equation*}}
\newcommand{\ees}{\end{equation*}}
\newcommand{\rmi}{{\rm i}}
\renewcommand{\Im}{{\rm Im}}
\newcommand{\id}{{\rm I}}
\newcommand{\Id}{{\rm Id}}
\newcommand{\ot}{\omega t}
\newcommand{\M}{\cal M}
\newcommand{\B}{\cal B}
\renewcommand{\L}{\cal L}
\renewcommand{\l}{\ell}
\newcommand{\besi}{\ensuremath{{\beta,\sigma}}}
\newcommand{\phom}{\phi,\omega}
\newcommand{\phomp}{\phi,\omega'}
\newcommand{\omp}{\omega,\omega'}
\newcommand{\dom}{\omega-\omega'}
\newcommand{\nom}{\omega\ne\omega'}
\newcommand{\bom}{\big(B(\omega)\big)}
\newcommand{\bomp}{\big(B(\omega')\big)}
\newcommand{\besn}[1]{\ensuremath{{\beta,s-#1\sigma}}}
\newcommand{\gapl}{\gamma^+}
\newcommand{\gami}{\gamma^-}
\newcommand{\copl}{C_\omega^+}
\newcommand{\comi}{C_\omega^-}
\newcommand{\cmpl}{C_\mu^+}
\newcommand{\cmmi}{C_\mu^-}
\newcommand{\clpl}{C_\lambda^+}
\newcommand{\clmi}{C_\lambda^-}

\newcommand{\wt}{\widetilde}
\newcommand{\lra}[2]{\la#1,#2\ra}

\newcommand{\expc}{\exp\Big(\frac{C}{\sigma^{a_3}}\Big)}
\newcommand{\exptc}{\exp\Big(\frac{2C}{\sigma^{a_3}}\Big)}

\newcommand{\expcl}{\exp\Big(\frac{C}{\sigma_{l+1}^{a_3}}\Big)}
\newcommand{\exptcl}{\exp\Big(\frac{2C}{\sigma_{l+1}^{a_3}}\Big)}

\newcommand{\pmdp}{P^--\text{diag}(P^-)}

\newcommand{\s}{s}

\newcommand{\mbhs}{\M_\beta(\Pi_*,\frac s2)}

\newcommand{\fij}{{\txs \frac ij}}
\newcommand{\fji}{{\txs \frac ji}}

\newcommand{\lst}{\B(\l_s^2)}
\newcommand{\lms}{\B(\l_{-s}^2)}
\newcommand{\lzt}{\B(\l_0^2)}
\newcommand{\ldt}{\B(\l_{-2\delta}^2)}
\newcommand{\lzd}{\B(\l_0^2,\l_{-2\delta}^2)}
\newcommand{\lzi}{\B(\l_0^2,\l_{-2\iota}^2)}

\newcommand{\px}{\partial_x}

\def\e{\epsilon}
\def\ve{\epsilon}

\def\Z{{\Bbb Z}}
\def\R{{\Bbb R}}
\def\T{{\Bbb T}}
\def\C{{\Bbb C}}
\def\O{{\cal O}}

\let\cal=\mathcal
\def\H{{\cal H}}

\hfuzz=\maxdimen
\newcommand{\belh}{H_{\frac{1}{3}}^{(1)}}

\newcommand{\sbelhz}[1]{\sqrt{\frac{\pi z}{2}}H_{\frac{1}{3}}^{(1)}(#1)}

\newcommand{\Bgs}[1]{\ensuremath{\Big|#1\Big|}}

\newcommand{\dss}{\displaystyle}
\newcommand{\txs}{\textstyle}
\newcommand{\scs}{\scriptstyle}

\begin{document}
\newtheorem{Sup}{\textbf{Assumption}}

\title[Reducibility of 1-d Schr\"odinger equation]{Reducibility of 1-d Schr\"odinger equation with unbounded oscillation perturbations}
\author{Z. Liang and Z. Wang}

\address {School of Mathematical Sciences and
Key Lab of Mathematics for Nonlinear Science, Fudan University,
Shanghai 200433, China} \email{zgliang@fudan.edu.cn, 19110180010@fudan.edu.cn}

\date{}
%\tableofcontents
\begin{abstract} We build a new estimate for the normalized eigenfunctions of the operator $-\partial_{xx}+\cal V(x)$ based on the oscillatory integrals and  Langer's turning point method, where $\cal V(x) \sim |x|^{2\l}$ at infinity  with $\l>1$. From it and an improved reducibility theorem
we show that the equation
\[\txs
{\rm i}\partial_t \psi =-\partial_x^2 \psi+\cal V(x) \psi+\e \la x\ra^{\mu} W(\nu x,\omega t)\psi,~\psi=\psi(t,x),~x\in\R,~ \mu<\min\big\{\l-\frac23,{\scs\frac{\sqrt{4\l^2-2\l+1}-1}2}\big\},
\]
can be reduced in $L^2(\R)$  to an autonomous system for most values of the frequency vector $\omega$ and $\nu$,
where  $W(\varphi, \phi)$  is a smooth map from $ \T^d\times \T^n$ to $\R$ and odd in $\varphi$.
\end{abstract}
\maketitle
\section{Introduction of the Main Results}\label{introduction}
\subsection{Main results}
\noindent In this paper we study the problem of reducibility of the time dependent Schr\"odinger equation
\begin{equation}\label{introduction01}
\begin{aligned}
&\H(t)\psi(x,t)=\rmi\partial_t\psi(x,t),\quad x\in\R;\\
&\H(t):=-\frac{d^2}{dx}+\cal V(x)+\e\la x\ra^\mu W(\nu x,\ot),\quad\e\in\R,
\end{aligned}
\end{equation}
 where $\cal V(x)\sim |x|^{2\l}$ at infinity with $\l>1$ and $W$ is a smooth function on $\T^d\times \T^{n}$.
 In order to state the results we need to introduce  some notations and spaces.
 We define the weight $\lambda(x,\xi)=(1+\xi^2+|x|^{2\l})^{\frac1{2\l}}$.
For $x,y\in\R$, define $\la x\ra:=\sqrt{1+x^2}$ and $x\vee y:=\max\{x,y\},~x\wedge y:=\min\{x,y\}$ and $x\sim y$ means that there exist some positive constants $C,~\widetilde C$ such that $\widetilde Cy\le x\le C y$. $d, n\in \Z_{+}$.
As \cite{BamII} we define the following. \\
{\bf Symbol.} The space $ S^{m_1,m_2}$ is the space of the symbols $g\in\cal C^\infty(\R^2)$ such that $\forall~k_1,k_2\ge0$, there exists $C_{k_1,k_2}$ with the property that
\begin{equation}\label{symnorm}
|\partial_\xi^{k_1}\partial_x^{k_2}g(x,\xi)|\le C_{k_1,k_2}\big(\lambda(x,\xi)\big)^{m_1-k_1\l}\la x\ra^{m_2-k_2}.
\end{equation}
The best constants $C_{k_1,k_2}$ such that \eqref{symnorm} holds form a family of semi-norms for that space $ S^{m_1,m_2}$.
{\bf Quantization.} To a symbol $g\in S^{ m_1,m_2}$, we associate its Weyl quantization, namely the operator $g^{w}(x,-\rmi\partial_x)$, defined by
\[
g^{w}(x,-\rmi\partial_x)\psi(x):=\frac1{2\pi}\int_{\R^2}e^{(x-y)\cdot\xi}g\big(\frac{x+y}2,\xi\big)\psi(y)dyd\xi.
\]
\indent We use the symbol $\lambda(x,\xi)$ to define, for $s\geq 0$ the spaces
$\cal{H}^s=D([\lambda^{w}(x,-{\rm i}\partial_x)]^{s(\l+1)})$(domain of the $(s(\l+1))$th- power of the operator operator $\lambda^{w}(x,-{\rm i}\partial_x)$ endowed by the graph norm.
For negative $s$, the space $\cal{H}^s$ is the dual of $\cal{H}^{-s}$. We will denote by $\cal{B}(\cal H_1,\cal H_2)$ the space of bounded linear operators from $\cal H_1$ to $\cal H_2$, where $\cal H_1,\cal H_2$ are Banach spaces. In particular, $\cal B(\cal H_1,\cal H_1)$ is usually abbreviated as $\cal B(\cal H_1)$. As \cite{BamI} in what follows we will identify $L^2$ with $\ell_0^2$ by introducing the basis denoted by $\{h_j(x)\}_{j\geq 1}$ of the eigenvector of $H_0 : = -\partial_{xx}+\cal V(x)$.
Similarly we will identify $\cal{H}^s$ with the space $\ell_{2s}^2$ of the sequences $\psi_j$ such that
$\sum\limits_{j\geq 1}j^{2s}|\psi_j|^2<\infty$. \\
\indent  Now we can state our main results below. Consider the time dependent Schr\"odinger equation (\ref{introduction01})
under the following conditions:\\
A1: We assume that the potential $\cal V$ belongs to $S^{0,2\l}$ to be symmetric and non - negative, namely,
$\cal V(x)= \cal V(-x)\geq 0$ and furthermore admits an asymptotic expansion of the form
$\cal V(x) : = |x|^{2\l}(c_0+ w(x)) $, where $w(x) =\sum\limits_{j\geq 1}\frac{c_j}{|x|^{2j}}$ is convergent on $(\sqrt{A_*}, +\infty)$
and $\overline{\lim\limits_{n\to \infty}} {|c_j|}^{\frac{1}{n}}=A_*\geq 0$. \\
A2: $\nu\in [A, B]^d$ satisfies Diophantine conditions, namely there exist $\bar{\gamma}>0$ and $\tau_1>d-1$ s.t.
$$|\la k,\nu\ra|\ge\frac{\bar\gamma}{|k|^{\tau_1}}, \qquad k\neq 0,\  \bar{\gamma}>0. $$
A3: $W(\varphi,\phi)$ is defined on $\T^{ d}\times\T^n$ and  for $\forall~(\varphi,\phi)\in\T^{ d}\times\T^n$, $W(-\varphi,\phi)=-W(\varphi,\phi)$.
For any $\varphi\in\T^ d$ and all $\alpha=(\alpha_1,\cdots,\alpha_ d),~\partial_\varphi^\alpha W(\varphi,\phi)$ is analytic on $\T_\rho^n$ and continuous on $\T^ d\times\overline{\T_\rho^n}$, where $0\le|\alpha|=\alpha_1+\cdots+\alpha_ d\le d([1\vee\tau_1]+ d+2)$.\\
\begin{remark}\label{pudeguji}
We denote by $\lambda_j$  the sequence of the eigenvalues of $H_0$ labeled in increasing order.
As \cite{BamII, HRa} one can show that  $\lambda_j\sim cj^{\frac{2\l}{\l+1}}$ as $j\to \infty$ and $\lambda_1>0$.
\end{remark}

Our purpose is to prove the following.
\begin{Theorem}\label{mainthm1}
Assume A1-A3 and $\mu<\txs(\l-\frac23)\bigwedge\scs\frac{\sqrt{4\l^2-2\l+1}-1}2$. Fix a $\gamma>0$ small, there exists $\e_*>0$ such that for all $0\leq \e<\e_*$ there exists a closed set $\Pi_*\subset\Pi:=[0,1]^n$ and $\forall~\omega\in\Pi_*$,
the linear Schr\"odinger equation \eqref{introduction01} reduces to a linear equation with constant coefficients in $  L^2$.\\
\indent More precisely, for a $\gamma>0$, there exists $\e_*>0$ such that for all $0\leq \e<\e_*$ there exists a closed set $\Pi_*\subset\Pi$ satisfying $\text{meas}(\Pi\backslash\Pi_*)\le C\gamma$, and for $\omega\in\Pi_*$, there exists a unitary (in $L^2$) time quasiperiodic operator $\Psi_{\omega,\e}(\phi)$ such that $t\mapsto\psi(t,\cdot)\in L^2$ satisfies \eqref{introduction01} if and only if $t\mapsto u(t,\cdot)=\Psi_{\omega,\e}^{-1} \psi(t,\cdot)$ satisfies the equation
$$\rmi\dot u=\H_\infty u$$ with $\H_\infty=\text{diag}\{\lambda_j^\infty\}$ and $|\lambda_j^\infty-\lambda_j|\le C\e j^{\big(\frac{\mu}{\l+1}-\frac1{\l+1}(\frac13\wedge\frac{\mu+1}{2\mu+2\l+1})\big)\bigvee0}$.\\
Furthermore one has:\\
1. $\lim\limits_{\gamma\to0}\text{meas}(\Pi\backslash\Pi_*)=0$;\\
2. $\Psi_{\omega,\e}(\phi)$ is analytic in the norm $\|\cdot\|_{\cal{B}( L^2)}$ on $|\Im\phi|<\frac s2$;\\
3. $\|\Psi_{\omega,\e}(\ot)-\Id\|_{\cal{B}(  L^2)}\le C\e^\frac23$.
\end{Theorem}
\begin{remark}
\[\txs
(\l-\frac23)\bigwedge\frac{\sqrt{4\l^2-2\l+1}-1}2=
\begin{cases}
\l-\frac23,&1<\l<\frac43;\\
\frac{\sqrt{4\l^2-2\l+1}-1}2=\l-\frac34+\frac1{16}\l^{-1}+\O(\l^{-2}),&\l\ge\frac43.
\end{cases}
\]
From a straightforward computation we have $\frac{\sqrt{4\l^2-2\l+1}-1}2\geq  \l-\frac34$ when $\l\geq \frac43$.
\end{remark}
\begin{remark} $s$ satisfies $0<s<\rho$.
\end{remark}

\begin{remark}
In \cite{BG} Bambusi and Graffi first proved the reducibility of 1d Schr\"odinger equation with an unbounded time quasiperiodic perturbation. They assumed a similar potential as (\ref{introduction01}) and the perturbation operator is $\epsilon W(x,\omega t )$ with $|W(x,\phi)|\sim |x|^{\beta}$ as $|x|\rightarrow \infty$, where $\beta<\l-1$.  The reducibility in the limiting case $\beta=\l-1$ was obtained by Liu and Yuan in \cite{LY10}. Comparing with \cite{BG} and \cite{LY10}, we improve  the boundedness for $\beta$ from $\l-1$ to $\txs(\l-\frac23)\bigwedge\frac{\sqrt{4\l^2-2\l+1}-1}2$ for  $\l>1$ when  the perturbation terms have the oscillatory forms as (\ref{introduction01}).
\end{remark}

\begin{remark}
In \cite{BamII} Bambusi studied the reducibility of 1d Schr\"odinger  equations
\begin{equation*}%\label{scheqn1}
\begin{aligned}
&\H_0(t)\psi(x,t)=\rmi\partial_t\psi(x,t),\quad x\in\R;\\
&\H_0(t):=-\frac{d^2}{dx}+\cal V(x)+\e W(\omega t),\quad\e\in\R,
\end{aligned}
\end{equation*}
where  $\cal V(x)\sim |x|^{2\l}$  at infinity and  satisfies a similar condition as ${\rm A1}$. The perturbation operator $W(\omega t)$ belongs to a class of unbounded symbols,
in which the oscillatory perturbation terms $\e \la x\ra^{\mu} W(\nu x,\omega t)$ in (\ref{introduction01}) are clearly excluded by  \cite{BamII}, \cite{BamI} and \cite{BamIII}.  See
Remark 2.7 in \cite{BamII} for further explanations.
\end{remark}

\begin{remark}
See  \cite{LiangLuo2019} for 1d quantum harmonic oscillators with similar perturbation terms as (\ref{introduction01}).  We remark that the reducibility results in \cite{GP16}, \cite{LiangLuo2019}  and \cite{LiangW19} were proved only in $\mathcal{H}^1$.   The reason partly lies in  (iii) of Lemma 2.1 in \cite{GP16}.   In our notations, when the perturbation operator $P$ belongs to $\mathcal M_{\beta}$, one can only deduce that $P\in\cal{B}(\l_t^2;\l_{-t}^2)$ for  $t>2\beta+1$  from
Lemma  \ref{algebra}, where $0\leq 2\beta<\iota-1$.
\end{remark}

Similarly, we consider the Schr\"odinger equation
\begin{equation}\label{scheqnsin}
\begin{aligned}
&\H_1(t)\psi(x,t)=\rmi\partial_t\psi(x,t),\quad x\in\R;\\
&\H_1(t):=-\frac{d^2}{dx^2}+\cal V(x)+\e\la x\ra^\mu X(x,\ot),
\end{aligned}
\end{equation}
where $X(x,\phi)=\sum\limits_{k\in\Lambda}a_k(\phi)\sin kx+b_k(\phi)\cos kx$ with $k\in\Lambda\subset\R\backslash\{0\}$ with $|\Lambda|<\infty$, $a_k(\phi)$ and $b_k(\phi)$ are analytic on $\T_\rho^n$ and continuous on $\overline{\T_\rho^n}$.
\begin{Theorem}\label{mainthm2}
Assume A1 and $\mu<\txs(\l-\frac23)\bigwedge\frac{\sqrt{4\l^2-2\l+1}-1}2$. For a $\gamma>0$ small, there exists $\e_*>0$ such that for all $0\leq \e<\e_*$ there exists a closed set $\Pi_*\subset\Pi:=[0,1]^n$ and $\forall~\omega\in\Pi_*$, the linear Schr\"odinger equation  \eqref{scheqnsin} reduces to a linear equation with constant coefficients in $ L^2$.
\end{Theorem}

We consider the Schr\"odinger equation
\begin{equation}\label{scheqnbam}
\begin{aligned}
&\H_2(t)\psi(x,t)=\rmi\partial_t\psi(x,t),\quad x\in\R;\\
&\H_2(t):=-\frac{d^2}{dx^2}+\cal V(x)+\e\la x\ra^\mu g(x,\omega t),
\end{aligned}
\end{equation}
where $g(x,\phi)$ is continuous on $x\in\R$ and analytic on $\T_\rho^n$ and there exists a  positive constant $C$ such that  for any $(x,\phi)\in\R\times\T_\rho^n$,
$
|g(x,\phi)|\le C.
$
\begin{Corollary}\label{maincorBam}
Assume A1 and $\mu<\l-1$. For a $\gamma>0$ small, there exists $\e_*>0$ such that for all $0\leq \e<\e_*$ there exists a closed set $\Pi_*\subset\Pi:=[0,1]^n$ and $\forall~\omega \in\Pi_*$, the Schr\"odinger equation \eqref{scheqnbam} reduces to a linear equation with constant coefficients in $ L^2$.
\end{Corollary}
\begin{remark}
The set $\Pi_{*}$ in Theorem \ref{mainthm2} and Corollary \ref{maincorBam} is similar as that in Theorem \ref{mainthm1}.
\end{remark}
\begin{remark}
Corollary \ref{maincorBam} was firstly proved in \cite{BG}.
\end{remark}
A consequence of the above theorems and corollary is that in the considered range of parameters all the Sobolev norms, i.e. the $\H^s$ norms of the solutions are bounded forever and the spectrum of the Floquet operator is pure point.\\
 \indent In the end we recall some relevant results. See \cite{GT11},  \cite{Wang08} and \cite{WLiang17} for the reducibility results for 1d harmonic oscillators with bounded perturbations. We remark that
the pseudodifferential calculus is used for checking the assumption B3 in this paper.   More applications of pseudodifferential calculus can be found in the following  papers (e.g. \cite{BBM14, BBHM17, BM16, FP15, Giu17, IPT05,PT01}).  We mention that some higher dimensional results  have been recently obtained \cite{BGMR18,FGMP18, GP16, LiangW19,  Mon17b}. \\
\indent As we mentioned before, the reducibility implies the boundedness of the solutions in some Sobolev norms for all the time.  There are many literatures relative with the upper boundedness of the solution in some Sobolev space( e.g. \cite{BGMR2019}, \cite{Bou99}, \cite{MR2017}, \cite{Mon18}). There are not too much papers to study the lower boundedness of the PDEs. See the interesting examples given by Bourgain for a Klein-Gordon and Schr\"odinger equation on $\T$(\cite{Bou99}), by Delort for the harmonic oscillator on $\R$ (\cite{Del14}). Combining  the ideas in \cite{BGMR18} and \cite{Eli1992},  Z.  Zhao, Q.  Zhou and the first author \cite{LZZ20} build some lower boundedness estimates  for 1d harmonic oscillators with quadratic time-dependent perturbations. We remark that the result in \cite{Del14} was  reproved in \cite{Mas2018} by exploiting the idea in \cite{GY00}.

\subsection{A new oscillatory integral estimation}
The following oscillatory integral estimations are critical for us to establish Theorem \ref{mainthm1} and \ref{mainthm2}.
\begin{assm}\label{poteass}
The potential $V(x)$ is real-valued and of $C^3$-class. There exists a positive constant $R_0$ such that the following conditions are satisfied for $V(x)$ when $|x|\ge R_0$:
\begin{flalign}
(\rm i).\phantom{\rm ii}&~V''(x)\ge0.&\label{convex}\\
(\rm ii).\phantom{\rm i}&~\text{For}~j=1,2,3,~|xV^{(j)}(x)|\le C_1|V^{(j-1)}(x)|,~\text{where}~C_1\ge1.&\label{decay}\\
(\rm iii).&~\text{For}~\l>1,~D_1|x|^{2\l}\le V(x)\le D_2|x|^{2\l},~\text{where}~0<D_1\le1\le D_2<\infty.&\notag%\label{growth}
\end{flalign}
\end{assm}
\begin{assm}\label{pertass}
The function $f(x)$ is real-valued and of $C^1$-class. There exist positive constant $R_0$ and $C_2>0$ such that for $|x|\ge R_0$,  $|f(x)|\le C_2|x|^\mu$ and
$|f'(x)|\le C_2|x|^{\mu-1}$
where $\mu\ge0$.
\end{assm}
\begin{lemma}\label{mainlem} Assume $V(x)$ satisfies Assumption \ref{poteass}.
Let $h_n(x)$ be the normalized eigenfunction of $H=-\partial_{xx}+V(x)$ with the eigenvalue $\lambda_n$
and $f(x)$ satisfies Assumption \ref{pertass}, then for any $k\neq 0$, one has
\begin{equation}\label{mainesto}
\Big|\int_{\R}f(x)e^{{\rm i} kx}h_m(x) \overline{{h}_n(x)}dx\Big| \le C (|k|\vee  |k|^{-1})(\lambda_m\lambda_n)^{\frac\mu{4\l}-\frac1{4\l}(\frac13\bigwedge\frac{\mu+1}{2\mu+2\l+1})},
\end{equation}
where $C$ depends on $(\mu,\l)$ and
$\mu\ge0$.
\end{lemma}
\begin{remark}
See \cite{YajimaZhang} for $L^{p}$ estimate of $h_n(x)$ based on Langer's turning point method when $n$ is large enough.  For a complete introduction of  Langer's turning point method refer to  the contents in Chapter 22.27 of \cite{T2}. Refer to  \cite{WLiang17} for a weighted $L^2$ estimate of the eigenfunctions of $-\partial_{x}^2+x^2$ on $\R$, which is another application of this method.
\end{remark}
\begin{remark}
In \cite{LiangLuo2019},  Luo and the first author proved the following:
 for any $k\neq 0$ and  for any $m,n\geq 1$,
\begin{eqnarray}\label{liangluo}
\Big|\int_{\R}\la x\ra^{\mu} e^{{\rm i} kx}f_m(x) \overline{{f}_n(x)}dx\Big| \leq C (|k|\vee  |k|^{-1})m^{-\frac{1}{12}+\frac{\mu}{4}} n^{-\frac{1}{12}+\frac{\mu}{4}},
\end{eqnarray}
where $C$ is an absolute constant and $\mu\ge0$ and $(-\frac{d^2}{dx^2}+x^2)f_j=(2j-1)f_j$ with $\|f_j\|_{L^2(\R)}=1$ with $ j\geq 1$.
\end{remark}

\begin{lemma}\label{mainlem2} Assume $V(x)$ satisfies Assumption \ref{poteass} and $h_n(x)$ is the same as in Lemma \ref{mainlem}. If $f(x)$ is continuous and  satisfies $|f(x)|\le C_2|x|^\mu$ for $|x|\ge  R_0$ with $0\le\mu<\l-1$, then
\begin{equation}\label{mainest}
\Big|\int_{\R}f(x)h_m(x) \overline{{h}_n(x)}dx\Big| \le C(\lambda_m\lambda_n)^{\frac\mu{4\l}}.
\end{equation}
\end{lemma}

\noindent {Acknowledgements.}
The first author is very grateful to D. Bambusi for several invaluable discussions on the concerned perturbations in this paper.  Both authors were partially supported by National Natural Science Foundation of China(Grants No.  11371097; 11571249) and
Natural Science Foundation of Shanghai(Grants No. 19ZR1402400).

\section{A Reducibility Theorem}\label{section2}
Before give the proof of main theorem we present a reducibility theorem in a more abstract setting, which is a  part of KAM theory.
 We remark that KAM theory is almost well-developed for nonlinear Hamiltonian PDEs in 1-d context. See \cite{BBP2, GY1, KLiang, KaPo,  Kuk93, Ku1, Ku2, KP, LZ, LY1, LiuYuan, Pos, WLiang17, W90, ZGY} for 1-d KAM results.  Comparing with 1-d case, the KAM results for multidimensional PDEs are relatively few. Refer to  \cite{EGK, EK, GXY, GY2, GP16,  LiangW19, PX} for n-d results. We refer to \cite{Berti16} and \cite{Berti2019} for an almost complete picture of recent KAM theory.
\subsection{Setting} Following \cite{GP16}, we will introduce some spaces and norms and discuss some algebraic properties.
\\ \noindent
{\bf Linear Space.} Let $s\in \R$, we define the complex weighted-$\l^2$-space
\[
\l_s^2=\big\{\xi=(\xi_j\in\C,~j\in\Z_+)~\big|~\|\xi\|_s<\infty\big\},~\text{where }~\dss\|\xi\|_s^2=\sum_{j\in\Z_+}j^s|\xi_j|^2.
\]
{\bf Infinite Matrices.} We denote by $\M_\beta$ the set of infinite matrix $A:\Z_+\times\Z_+\mapsto\C$ that satisfy
$
|A|_\beta:=\sup_{i,j\ge1}|A_i^j|(ij)^{-\beta}<\infty.
$
We will also need the space $\M_\beta^+$ the following subspace of $\M_\beta$: an infinite matrix $A$ is in $\M_\beta^+$ if
$
|A|_\beta^+=\sup_{i,j\ge1}|A_i^j|(ij)^{-\beta}(1+|i-j|)(i^{\iota-1}+j^{\iota-1})<\infty,
$
where $0\le2\beta<\iota-1$ and $\iota>1$.\\
\begin{lemma}\label{expd}
For $\iota>1,~|k^\iota-j^\iota|\ge\frac12|k-j|(k^{\iota-1}+j^{\iota-1})$ (see \cite{KLiang}).
\end{lemma}
\begin{lemma}\label{algebra}
If $0\le2\beta<\iota-1$, there exists a constant $C>0$ such that\\
(i). Let $A\in\M_\beta$ and $B\in\M_\beta^+$, then $AB$ and $BA$ belong to $\M_\beta$ and
\(
|AB|_\beta,~|BA|_\beta\le C|A|_\beta|B|_\beta^+.
\)
The proof is by the definition.\\
(ii). Let $A,B\in\M_\beta^+$, then $AB$ belongs to $\M_\beta^+$ and
\(
|AB|_\beta^+\le C|A|_\beta^+|B|_\beta^+.
\)
\begin{proof}
Since $A,~B\in\M_\beta^+$, then
\begin{align*}
&|(AB)_i^j|(ij)^{-\beta}(1+|i-j|)(i^{\iota-1}+j^{\iota-1})\\
\le~&|A|_\beta^+|B|_\beta^+
\sum_{k\ge1}\frac{k^{2\beta}(1+|i-j|)(i^{\iota-1}+j^{\iota-1})}{(1+|k-j|)(1+|i-k|)(k^{\iota-1}+j^{\iota-1})(k^{\iota-1}+i^{\iota-1})}.
\end{align*}
We discuss two cases. When $i=j$, the proof is simple.
When $i\ne j$, note $|i-k|+|k-j|\ge|i-j|$, we have $|k-j|\ge\frac12|i-j|$ or $|i-k|\ge\frac12|i-j|$.
The proof is by a straightforward computation.
\end{proof}\noindent
(iii). Let $A\in\M_\beta$, then for any $t>2\beta+1,~A\in\cal{B}(\l_t^2;\l_{-t}^2)$ and
\(
\|A\xi\|_{-t}\le C|A|_\beta\|\xi\|_t.
\)
\begin{proof}
 Let $A\xi=\eta,~\eta_i=\sum_{k\ge1}A_i^k\xi_k$, then
$$
\|\eta\|_{-t}^2=\sum_{i\ge1}i^{-t}|\eta_i|^2=\sum_{i\ge1}i^{-t}\big|\sum_{k\ge1}A_i^k\xi_k\big|^2
\le|A|_\beta^2\sum_{i\ge1}\frac1{i^{t-2\beta}}\big(\sum_{k\ge1}k^{2\beta-t}\big)\big(\sum_{k\ge1}k^t|\xi_k|^2\big).
$$
Note $t>1+2\beta$, then we obtain $\|A\xi\|_{-t}\le C|A|_\beta\|\xi\|_t$.
\end{proof}\noindent
(iv). Let $A\in\M_\beta^+$, then for any $s\in[0,2\iota-2\beta-1),~A\in\lst$ and satisfies
\[
\|A\|_{\lst}\le C(\beta,\iota,s)|A|_\beta^+,
\]
where $0\le2\beta<\iota-1$.

\noindent
(v). Let $A\in\M_\beta^+$, then for any $s\in[ 0,2\iota-2\beta-1),~A\in\lms$ and satisfies
\[
\|A\|_{\lms}\le C(\beta,\iota,s)|A|_\beta^+,
\]
where $0\le2\beta<\iota-1$.\\
The proof for (iv) and (v) is a little long and we will delay it in section \ref{section4}.
\end{lemma}

\begin{lemma}\label{complete}
$(\M_\beta,|\cdot|_{\beta})$ and $(\M_\beta^+,|\cdot|_{\beta}^+)$ are Banach spaces.
\end{lemma}\noindent
{\bf Parameter.} In the paper $\omega$ will play the role of a parameter belonging to $\Pi=[0,1]^n$. All the constructed operators or matrices will depend on $\omega$ in Lipschitz sense which will be clear in the following.\\ \indent
Let $D\subset\Pi$  and $\sigma>0$. We denote by $\M_\beta(D,\sigma)$ the set of mappings at $\T_\sigma^n\times D\ni(\phom)\mapsto Q(\phom)\in\M_\beta$ which is real analytic on $\phi\in\T_\sigma^n=\big\{\phi\in\C^n~\big|~|{\rm Im}\phi|<\sigma\big\}$ and Lipschitz continuous on $\omega\in D$. This space is equipped with the norm
\(
\|Q\|_\besi^{D,\L}=\|Q\|_\besi^D+\|Q\|_\besi^{D,lip},
\)
where we define
$
\|Q(\cdot,\omega)\|_\besi=\sup\limits_{|{\rm Im}\phi|<\sigma}|Q(\phom)|_\beta$,
$\|Q\|_\besi^D=\sup\limits_{\omega\in D}\|Q(\phom)\|_\besi$ and
$$\|Q\|_\besi^{D,lip}=
\begin{array}{c}
\sup\\ \substack{\omp\in D\\ \nom}
\end{array}
\frac{\|Q(\phom)-Q(\phomp)\|_\besi}{|\dom|}.$$
Similarly, we can define the subspace of $\M_\beta(D,\sigma)$, named by $\M_\beta^+(D,\sigma)$, the set of mappings at $\T_\sigma^n\times D\ni(\phom)\mapsto R(\phom)\in\M_\beta^+$ which is real analytic on $\phi\in\T_\sigma^n$ and Lipschitz continuous on $\omega\in D$. This space is equipped with the norm
\(
\|R\|_\besi^{D,\L,+}=\|R\|_\besi^{D,+}+\|R\|_\besi^{D,lip,+},
\)
where
we define $\|R\|_\besi^{D,+}$ and $\|R\|_\besi^{D,lip,+}$ similarly.
Generally, for a Banach space $(\cal B,\|\cdot\|)$, we define the parameterized space $\cal B(D,\sigma)$, the set of mappings at $\T_\sigma^n\times D\ni(\phom)\mapsto F(\phom)\in\cal B$ which is real analytic on $\phi\in\T_\sigma^n$ and Lipschitz continuous on $\omega\in D$. This space is equipped with the norm
$\|F\|_{\cal B, \sigma}^{D,\L}:=\|F(\cdot,\omega)\|_{\cal B, \sigma}^{D}+\|F(\cdot,\omega)\|_{\cal B, \sigma}^{D,lip}$,
where we define
$
\|F(\cdot,\omega)\|_{\cal{B}, \sigma}=\sup\limits_{|\Im\phi|<\sigma}\|F(\phom)\|_{\cal{B}}$,
$\|F\|_{\cal{B}, \sigma}^D=\sup\limits_{\omega\in D}\|F(\phom)\|_{\cal{B}, \sigma}$ and
$$\|F\|_{\cal{B}, \sigma}^{D,lip}=
\begin{array}{c}
\sup\\ \substack{\omp\in D\\ \nom}
\end{array}
\frac{\|F(\phom)-F(\phomp)\|_{\cal{B}, \sigma}}{|\dom|}.$$
In addition, for convenience we abbreviate $\B(D,0)$ as $\B(D)$.
\begin{lemma}\label{algebrapara}
If $0\le2\beta<\iota-1$ and $\iota>1$, there exists a constant $C>0$ such that\\
(i). If $A\in\M_\beta(D,\sigma)$ and $B\in\M_\beta^+(D,\sigma)$, then $AB$ and $BA$ belong to $\M_\beta(D,\sigma)$ and
\[
\|AB\|_\besi^{D,\L},~\|BA\|_\besi^{D,\L}\le  C\|A\|_\besi^{D,\L}\|B\|_\besi^{D,\L,+}.
\]
Moreover, if $A,~B\in\M_\beta^+(D,\sigma)$, then $AB$ belongs to $\M_\beta^+(D,\sigma)$ and
\[
\|AB\|_\besi^{D,\L,+}\le C\|A\|_\besi^{D,\L,+}\|B\|_\besi^{D,\L,+}.
\]
\noindent (ii). If $B\in\M_\beta^+(D,\sigma)$ then $e^B-\id$ and $e^B-\id-B$ belong to $\M_\beta^+(D,\sigma)$ and
\begin{align}
&\|e^B-\id\|_\besi^{D,\L,+}\le e^{C\|B\|_\besi^{D,+}}\|B\|_\besi^{D,\L,+}.\notag
\\
&\|e^B-\id-B\|_\besi^{D,\L,+}\le Ce^{C\|B\|_\besi^{D,+}}\big(\|B\|_\besi^{D,\L,+}\big)^2.\label{bexpbetaplusb}
\end{align}
Moreover, if $A\in\M_\beta(D,\sigma)$, then $Ae^B,~e^BA\in\M_\beta(D,\sigma)$ and
\begin{equation}\label{abexpbetaplus}
\|Ae^B\|_\besi^{D,\L},~\|e^BA\|_\besi^{D,\L}\le\|A\|_\besi^{D,\L}e^{3C\|B\|_\besi^{D,\L,+}}.
\end{equation}
\begin{proof}
If $A\in\M_\beta(D,\sigma),~B\in\M_\beta^+(D,\sigma)$, note $\dss e^B-\id=\sum_{n=1}^\infty\frac{B^n}{n!}$,thus for any $\omega\in D$ and $\phi\in\T_\sigma^n$,
\begin{align*}
|e^B-\id|_\beta^+\le\sum_{n=1}^{\infty}\frac{|B^n|_\beta^+}{n!}\le \frac1C\big(e^{(C|B|_\beta^+)}-1\big)\le \|B\|_\besi^{D,+}e^{C\|B\|_\besi^{D,+}}.
\end{align*}
It follows
\begin{equation}\label{bexpdplus}
\|e^B-\id\|_\besi^{D,+}\le \|B\|_\besi^{D,+}e^{C\|B\|_\besi^{D,+} }.
\end{equation}
In the following we estimate $\|e^B-\id\|_\besi^{D,lip,+}$. In fact,
$
e^{B(\omega)}-e^{B(\omega')}=\sum_{n=1}^\infty\frac{\big(B(\omega)\big)^n-\big(B(\omega')\big)^n}{n!}.
$
By induction for $n\geq 1$, we have
$
\big\|\bom^n-\bomp^n\big\|_\besi^+\le n\big(C\|B\|_\besi^{D,+}\big)^{n-1}\|B(\omega)-B(\omega')\|_\besi^+.
$
It follows for $\nom,~\omp\in D$,
\[
\|e^B(\omega)-e^{B}(\omega')\|_\besi^+\le\sum_{n=1}^\infty\frac{n\big(C\|B\|_\besi^{D,+}\big)^{n-1}\|B(\omega)-B(\omega')\|_\besi^+}{n!}.
\]
Thus,
$
\frac{\|e^B(\omega)-e^B(\omega')\|_\besi^+}{|\dom|}\le \|B\|_\besi^{D,lip}e^{C\|B\|_\besi^{D,+}}.
$
Together with \eqref{bexpdplus} we have
$
\|e^B-\id\|_\besi^{D,\L,+}\le \|B\|_\besi^{D,\L,+}e^{C\|B\|_\besi^{D,+}}.
$
Similarly, we obtain \eqref{bexpbetaplusb}. The proof of \eqref{abexpbetaplus} is similar.
\end{proof}
\noindent (iii) If $B\in\M_\beta^+(D,\sigma)$ and $ P\in\M_\beta(D,\sigma)$, then $e^{-B} P e^B$ belongs to $\M_\beta(D,\sigma)$ and
\begin{align*}
\|e^{-B} P e^B- P\|_\besi^D&\le Ce^{3C\|B\|_\besi^{D,+}}\| P\|_\besi^D\|B\|_\besi^{D,+},\\
\|e^{-B} P e^B- P\|_\besi^{D,\L}&\le Ce^{3C\|B\|_\besi^{D,\L,+}}\| P\|_\besi^{D,\L}\|B\|_\besi^{D,\L,+}.
\end{align*}
\begin{proof}
From \eqref{bexpdplus}, if $B\in\M_\besi^+(D,\sigma)$, then
$
\|e^B-\id\|_\besi^{D,+}\le \|B\|_\besi^{D,+}e^{\big(C\|B\|_\besi^{D,+}\big)}.
$
Similarly,
$
\|e^{-B}-\id\|_\besi^{D,+}\le \|B\|_\besi^{D,+}e^{\big(C\|B\|_\besi^{D,+}\big)}.
$
Since
$
e^{-B} P e^B- P=(e^{-B}-\id) P(e^B-\id)+(e^{-B}-\id) P+ P(e^B-\id),
$
the proof is clear by a straightforward computation.
\end{proof}
\end{lemma}

\subsection{A Reducibility Theorem}
Before present the new reducibility theorem we give a rough introduction of the proof and what is new here.  In fact the equation \eqref{introduction01} can be written as
\[
\rmi\dot x=(A+\e P(\phi))x,
\]
where $A=\text{diag}\{\lambda_j\}_{j\ge1}$ and $\lambda_j\sim j^\iota(j\to\infty)$ and $P=(P_i^j(\phi))$ with $\iota=\frac{2\l}{\l+1}$ and
\[
P_i^j(\phi)=P_j^i(\phi)=\int_\R\la x\ra^\mu W(\nu x,\phi)h_i(x)\overline{h_j(x)}dx.
\]
From Lemma \ref{chengfasuanzi} we can show that the map $\T^n\ni\phi~\mapsto~P(\phi)\in\lzd$ is analytic on $\T_s^n$ if $0\le\mu\le\delta(\l+1)$. If $\mu<\l-1$, then one can choose $\delta=\frac\mu{\l+1}$ and furthermore, $P(\phi)\in\B^\delta$ with $\delta<\frac{\l-1}{\l+1}$(we use the notation from \cite{BG}). From the reducibility Thm. in \cite{BG} and Lemma \ref{chengfasuanzi} we can prove the reducibility result for the equation \eqref{introduction01} when $\mu<\l-1$. \\
\indent The  main improvement is that we can deal with the case when $\l-1\le\mu<(\l-\frac23)\bigwedge\frac{\sqrt{4\l^2-2\l+1}-1}2$ for the equation \eqref{introduction01}  from Theorem \ref{Theorem2.10}.
In the new reducibility theorem we assume that  $0\le\delta<\frac{2\l}{\l+1}-\beta-\frac12$,
$0\leq 2\beta<\iota-1$, $2\beta\leq \delta$ and $0\le\mu\le\delta(\l+1)$.   In fact,  when $\frac13\bigwedge\frac{\sqrt{\l^2+2}-\l}2\le\mu<(\l-\frac23)\bigwedge\frac{\sqrt{4\l^2-2\l+1}-1}2$,  from Lemma \ref{mainlem} one can show that $P(\phi)\in\M_\beta$ with $\beta=\frac\mu{2(\l+1)}-\frac1{2(\l+1)}(\frac13\bigwedge\frac{\mu+1}{2\mu+2\l+1})$. If we choose $\delta=\frac{\l}{\l+1}$, it follows $0\leq 2\beta< \frac{\l-1}{\l+1}$ and all the assumptions in Theorem \ref{Theorem2.10} are easily checked, from which we can prove Theorem \ref{mainthm1}. Comparing with the proof in \cite{BG}, we need to control $\|P^-\|_{\beta,s}^\L+\|P^-\|_{\lzd}^\L$ in every step, while in \cite{BG} only $\|P^-\|_{\lzd}^\L$ was controlled.\\
\indent In the end we explain a little bit about  the equivalence of the two equations
$
\rmi\dot x=(A+P)x
$
and
$
\rmi\dot y=A^\infty(\ot)y,
$
which comes from the reducibility equality
\[
U^\infty(\ot)A^\infty(\ot)=-\rmi\frac d{dt}U^\infty(\ot)+(A^0+P^0(\ot))U^\infty(\ot)\quad\text{in}\quad\lzi(\Pi_*).
\]
 \noindent We remark that from $U^\infty(\ot)-\id\in\M_\beta^+$ and Lemma \ref{algebra}  we have $U^\infty(\ot)-\id\in\B(\l_{-s}^2)$ with $s\in[0,2\iota-2\beta-1)$. But it usually makes no sense for $U^\infty(\ot) A^\infty(\ot)x_0\in\l_{-2\iota}^2$ when $x_0\in\l_0^2$. The proof seriously depends on the homological equation \eqref{homoeqn}. See Lemma \ref{UA1},~\ref{UA2} and \ref{UlAlconverge} for details. Now we present the new reducibility theorem. \\
\indent Consider the non-autonomous, linear differential equation in a separable Hilbert space $\l_0^2$
\begin{equation}\label{redueqn}
{\rm i}\dot x(t)=\big(A+\e P(\omega_1t,\omega_2t,\dots,\omega_nt)\big)x(t),~\e\in\R,
\end{equation}
under the following conditions:\\
B1: $A=\text{diag}\{\lambda_j\}_{j\ge1}$ with $0<\lambda_1<\lambda_2<\cdots$.  There exists a $\iota>1$ such that
$\lambda_j\sim cj^{\iota}$ as $j\to \infty$.
\noindent B2: The map $\T^n\ni\phi~\mapsto~P(\phi) : = \big(P_i^j(\phi)\big)_{i,j\geq 1}\in\M_\beta$ is analytic on $\T_s^n$ and $P_i^j(\phi)=P_j^i(\phi)$ for $\phi\in\T^n$, where $0\le2\beta<\iota-1$.\\
B3: The map $\T^n\ni\phi~\mapsto~P(\phi)\in\lzd$ is analytic on $\T_s^n$, where $0\le\delta<\iota-\beta-\frac12$ and  $2\beta\leq \delta$.
\begin{Theorem}\label{Theorem2.10}
Assume that $B1 - B3$ are satisfied. Then for a $\gamma>0$ small, there exists $\e_*>0$ such that for all $0\le\e<\e_*$ there exists $\Pi_*\subset\Pi:=[0,1]^n$ satisfying $\text{meas}(\Pi\backslash\Pi_*)\le C\gamma$, such that for all $\omega\in\Pi_*$, the equation \eqref{redueqn} reduces to a linear equation
\begin{equation}\label{reduedeqn}
\rmi\dot y(t)=A^\infty(\ot)y(t),\quad A^\infty(\ot):=\text{diag}\{\lambda_1^\infty+\mu_1^\infty(\ot),\lambda_2^\infty+\mu_2^\infty(\ot),\cdots\},
\end{equation}
where $\{\lambda_j^\infty\}_{j\ge1}\in\R$ and the function $\mu_j^\infty(\phi):\T^n~\mapsto~\R$ is analytic on $\T_{\frac s2}^n$ with zero average.\\
\indent More precisely, for $\gamma>0$ small, there exists $\ve_*$ such that for all $0\le\ve<\ve_*$ there exists $\Pi_*\subset\Pi:=[0,1]^n$ satisfying $\text{meas}(\Pi\backslash\Pi_*)\le C\gamma$, and for $\omega\in\Pi_*$, there exists a linear unitary transformation $U^\infty(\phi)$ in $\ell_0^2$ which analytically depends on $\phi\in\T_{\frac s2}^n$ such that $t~\mapsto~y(t)\in\cal C^0(\R,\l_0^2)\cap\cal C^1(\R,\l_{-2\iota}^2)$ satisfies the equation \eqref{reduedeqn} if and only if
\begin{eqnarray}\label{implication}
t~\mapsto~x(t)=U^\infty(\omega t) y(t)\in\cal C^0(\R,\l_0^2)\cap\cal C^1(\R,\l_{-2\iota}^2)
\end{eqnarray} satisfies the equation \eqref{redueqn}, where there exists a positive constant  $C$ such that
\begin{eqnarray}\label{gujiforUinfty}
\|U^\infty(\ot)-\id\|_{\lzt}\le C\ve^\frac23,\quad|\lambda_j^\infty-\lambda_j|\le Cj^{2\beta}\e,\quad|\mu_j^\infty(\ot)|\le Cj^{2\beta}\e.
\end{eqnarray}
\end{Theorem}
\begin{remark}
The assumption $2\beta\leq \delta$ is not necessary but it can simplify the proof.
\end{remark}
\begin{Corollary}\label{tuilun2.12}
Assume that $B1 - B3$ are satisfied.  Then for a $\gamma>0$ small, there exists $\e_*>0$ such that for all $0\le\e<\e_*$ there exists $\Pi_*\subset\Pi:=[0,1]^n$ satisfying $\text{meas}(\Pi\backslash\Pi_*)\le C\gamma$, such that for all $\omega\in\Pi_*$, there is a unitary transformation $U_F(\ot)$ in $\ell_0^2$, quasiperiodic with frequency $\omega$ and such that
\(
\|U_F(\ot)-\id\|_{\lzt}\le C\e^{\frac23},
\)
which transforms \eqref{redueqn} into the equation
\[
\rmi\dot z(t)=A_Fz(t),\quad A_F=\text{diag}\{\lambda_1^\infty,\lambda_2^\infty,\cdots\}.
\]
Moreover, if $x(0)\in\l_0^2$, then $x(t)=U_F(\ot)\text{diag}\{e^{-\rmi\lambda_j^\infty}t\}U_F^{-1}(0)x(0)$ is the solution of \eqref{redueqn} in the sense of (\ref{implication})
\end{Corollary}

\subsection{Squaring the order of the Perturbation}
Let $\T_s^n$ be the complexified torus with $|\Im\phi|<s$
and $\Pi^-$ be a closed nonempty subset of $\Pi$ of positive measure. If the map $f:\T_n^s\times\Pi^-\mapsto\cal B(\l_{s_1}^2,\l_{s_2}^2)$ is analytic on $\phi\in\T_s^n$ and Lipschitz continuous on $\omega\in\Pi^-$, we define
$\|f\|_{\B(\l_{s_1}^2,\l_{s_2}^2),s}^{\Pi^-,\L}=\|f\|_{\B(\l_{s_1}^2,\l_{s_2}^2),s}^{\Pi^-}+\|f\|_{\B(\l_{s_1}^2,\l_{s_2}^2),s}^{\Pi^-,lip}$.
Similarly,  if the map $f:\T_n^s\times\Pi^-\mapsto \cal{M}_{\beta}$ is analytic on $\phi\in\T_s^n$ and Lipschitz continuous on $\omega\in\Pi^-$, we define
$\|f\|_{\beta,s}^{\Pi^-,\L}=\|f\|_{\beta,s}^{\Pi^-}+\|f\|_{\beta,s}^{\Pi^-,lip}$. For convenience we omit the symbol $\Pi^-$ here. \\
\indent Now we consider the equation in $\l_0^2$
\begin{equation}\label{2.4}
{\rm i}\dot{x}(t)=(A^- +P^-(\omega t))x
\end{equation}
under the following conditions\\
\indent H1)
\begin{equation}\label{formalA}
A^-=\text{diag}\{\lambda_1^-(\omega)+\mu_1^-(\omega t,\omega),\lambda_2^-(\omega)+\mu_2^-(\omega t,\omega),\cdots\}.
\end{equation}
Here:\\
\indent H1.a) $\forall~i=1,\cdots,\lambda_i^-(\omega)$ is positive and Lipschitz continuous w.r.t $\omega\in\Pi^-$ and satisfies $C_0^-i^\iota\le\lambda_i^-\le C_1^-i^\iota$ with $C_0^-,~C_1^->0$. We also assume that there is $C_\lambda^->0$ independent of $\omega$ such that
$|\lambda_i^--\lambda_j^-|\ge C_\lambda^-|i^\iota-j^\iota|.$\\
\indent H1.b) There is $C_\omega^->0$ suitably small and $0\leq 2\beta<\iota-1$ such that
\[
\begin{array}{c}
\sup\\ \substack{\omp\in\Pi^-\\ \nom}
\end{array}\frac{|\lambda_i^-(\phom)-\lambda_i^-(\phomp)|}{|\omega-\omega'|}\le C_\omega^-i^{2\beta}.
\]
\indent H1.c) $\forall~i=1,\cdots$, $\mu_i^-(\phi,\omega):\T_s^n\times\Pi^-\mapsto\R$ is analytic w.r.t $\phi$, Lipschitz continuous w.r.t $\omega$, and has zero average, i.e.
$\dss\int_{\T^n}\mu_i^-(\phi,\omega)d\phi=0.$
Moreover it fulfills the estimates
\[
\|\mu_i^-\|_\s\le C_\mu^- i^{2\beta},\quad
\begin{array}{c}
\sup\\ \substack{\omp\in\Pi^-\\ \nom}
\end{array} \frac{\|\mu_i^-(\phom)-\mu_i^-(\phomp)\|_\s}{|\omega-\omega'|}\le C_\omega^- i^{2\beta},
\]
where we denote $\|\mu_i^-(\cdot,\omega)\|_s=\sup\limits_{|\Im\phi|\le s}|\mu_i^-(\phom)|$.\\
\indent H2) The map $P^-:\T_s^n\times\Pi^-\mapsto \cal X$ is analytic w.r.t $\phi\in\T_s^n$ in the norm of $|\cdot |_{\beta}$ or $\|\cdot\|_{\lzd}$ and Lipschitz continuous w.r.t $\omega\in\Pi^-$ uniformly in $\phi\in\T_s^n$,
where $\cal X= \mathcal{M}_\beta$ or $\lzd$. \\
\indent H3) There exist $\gamma^->0$ and $\tau>n+\frac2{\iota-1}$ such that, for any $\omega\in\Pi^-$, one has
$$|\lra k\omega|\ge\frac{\gamma^-}{|k|^\tau},\quad\forall~k\in\Z^n\backslash\{0\},$$
$$|\lambda_i^--\lambda_j^-+\lra k\omega|\ge\frac{\gamma^-|i^\iota-j^\iota|}{1+|k|^\tau},\quad\forall~k\in\Z^n,~i\ne j.$$
Let now
$$B:\T_s^n\ni(\phi_1,\dots,\phi_n)\mapsto B(\phi_1,\dots,\phi_n)\in\mathcal{M}_\beta^+$$
be an analytic map with $B(\phi_1,\dots,\phi_n)$ anti-self-adjoint for each real value of $(\phi_1,\dots,\phi_n)$. Consider the corresponding unitary operator $e^{B(\phi_1,\dots,\phi_n)}$ and, for any $\omega\in\Pi^-$ consider the unitary transformation of basis $x=e^{B(\omega t)}y$. Substitution in equation (\ref{2.4}) yields
\[
\rmi\dot{y}=(A^++P^+(\omega t))y,\quad
A^+:=A^-+\text{diag}(P^-).
\]
In fact, $A^+=\text{diag}\{\lambda_i^++\mu_i^+(\ot)\}$, where $\lambda_i^+=\lambda_i^-+\overline{P_{ii}^-(\phi)}$ (the overline denotes angular average) and $\mu_i^+(\omega t)=\mu_i^-(\omega t,\omega)+P_{ii}^-(\phi)-\overline{P_{ii}^-(\phi)}$. Hence the functions $\mu_i^+(\phi)$ have zero average and $\text{diag}(P^-):=\text{diag}\{P_{11}^-(\omega t),P_{22}^-(\omega t),\dots\}$. The new perturbation $P^+$ is given by
\begin{align*}
P^+&:=\big([A^-,B]-\rmi\dot{B}+(P^--\text{diag}(P^-))\big)\\
&+\big(e^{-B}A^-e^B-A^--[A^-,B]\big)+\big(e^{-B}P^-e^B-P^-\big)-\rmi\big(e^{-B}\frac{d}{dt}e^B-\dot{B}\big).
\end{align*}
The main step of the proof is to construct $B$ such that  the following vanish, i.e. to solve for the unknown $B$ the equation
$[A^-,B]-{\rm i}\dot{B}+(P^--\text{diag}(P^-))=0.$
The construction is based on a lemma by Kuksin and a method from Bambusi \& Graffi \cite{BG}.
We also use the same notation as Bambusi \& Graffi \cite{BG} for reader's convenience.  The proof of Lemma \ref{Bdeguji} is similar as Lemma 3.2  in \cite{BG} and we will concentrate the difference with the proof in \cite{BG}. In the following we introduce Kuksin's lemma for completeness. \\
\indent On the $n$-dimensional torus consider the equation
\begin{equation}\label{unknownB}
-{\rm i}\sum_{k=1}^n \omega_k\frac\partial{\partial\phi_k}\chi(\phi)+E_1\chi(\phi)+E_2h(\phi)\chi(\phi)=b(\phi).
\end{equation}
Here $\chi$ denotes the unknown, while $b,h$ denote given analytic functions on $\T_s^n$. $h$ has zero average; $E_1$, $E_2$ are positive constants and $\|h\|_\s\le1$. Concerning the frequency vector $\omega=(\omega_1,\dots,\omega_n)$ Assumptions are:
$$|\lra k\omega|\ge\frac{\gamma_2}{|k|^\tau},\quad\forall~k\in\Z^n\backslash\{0\},$$
$$|\lra k\omega+E_1|\ge\frac{\gamma_1}{|k|^\tau+1},\quad\forall~k\in\Z^n.$$
The important  hypothesis is an order assumption, namely: given $0<\theta<1$ and $C_0>0$ we assume
\begin{equation}\label{keycd}
E_1^\theta\ge C_0E_2.
\end{equation}
\begin{lemma}[Kuksin]\label{Kuksin}
Under the above assumptions, equation \eqref{unknownB} has a unique analytic solution $\chi$ which for any $0<\sigma<s$ fulfills the
\[
\|\chi\|_{s-\sigma}\le\frac{C_1}{\gamma_1\sigma^{a_4}}\exp{\left(\frac{C_2}{\gamma_2^{a_5}\sigma^{a_3}}\right)}\|b\|_s.
\]
Here $a_3,a_4,a_5,C_1,C_2$ are constants independent of $E_1,E_2,\sigma,s,\gamma_1,\gamma_2,\omega$.
\end{lemma}
By Kuksin's lemma as \cite{BG} we have
\begin{lemma}\label{Bdeguji}
Let
\begin{eqnarray}\label{xianzhitiaojian1}
\gamma^-\ge\gamma_0/2, 0<\sigma\le C(\gamma_0,n,\tau)<1, 0\le C_\omega^-\le1,
\end{eqnarray}
and $\theta=\frac{2\beta}{\iota-1}$ in \eqref{keycd}. For any $0<\sigma<s$, equation
\begin{equation}\label{homoeqn}
[A^-,B]-{\rm i}\dot{B}+(P^--\text{diag}(P^-))=0
\end{equation}
has a unique solution $B\in\mathcal{M}_{\beta,s-\sigma}^+$ analytic on $\T_{s-\sigma}^n$, fulfilling the estimate
\begin{equation}\label{Bestimate}
\|B\|_{\beta,s-\sigma}^{\L,+}\le\expc\big(\|P^-\|_{\beta,s}^\L+\|P^-\|_{\lzd,s}^\L\big),
\end{equation}
where $C=C(\gamma_0,\beta,\iota,n,\tau)$, $a_3=n+\tau+\frac{\theta(n+\tau+2)}{1-\theta}$.
\end{lemma}
\begin{proof}
The equation equals to
\[
-{\rm i}\sum_{k=1}^n \omega_k\frac\partial{\partial\phi_k}B_{ij} +(\lambda_i^--\lambda_j^-)B_{ij}+(\mu_i^-(\phi)-\mu_j^-(\phi))B_{ij}=-P_{ij},\quad i\ne j,\omega\in\Pi^-.
\]
Assume $E_1=(\lambda_i^--\lambda_j^-)\ge0$ and $h_{i,j}(\phi)=\frac{\mu_i^-(\phi)-\mu_j^-(\phi)}{\|\mu_i^-(\phi)-\mu_j^-(\phi)\|_s+1}$, $E_2=\|\mu_i^-(\phi)-\mu_j^-(\phi)\|_s+1$, also denote $\gamma_1=\gamma^-|i^\iota-j^\iota|$ and $\gamma_2=\gamma^-$. We can choose $\theta=\frac{2\beta}{\iota-1}$ and a suitable constant $C_0$ such that \eqref{keycd} holds. In fact, since $|\lambda_i^--\lambda_j^-|\ge C_\lambda^-|i^\iota-j^\iota|\ge\frac{C_\lambda^-}2|i-j|(i^{\iota-1}+j^{\iota-1})$ and $\|\mu_i^-(\phi)-\mu_j^-(\phi)\|_s\le C_\mu^-(i^{2\beta}+j^{2\beta})$.
As \cite{BG}, one has  $E_1^\theta\ge C_0E_2$ with $\theta $ defined above.
Then a direct application of Kuksin's lemma yields
$$\|B_{ij}\|_{s-\sigma}\le\frac{C_1}{\gamma^-|i^\iota-j^\iota|\sigma^{a_4}} \exp{\left(\frac{C_2}{(\gamma^-)^{a_5}\sigma^{a_3}}\right)}\|P_{ij}\|_s, \quad {\rm for\ \ } i\ne j.$$
Note our assumptions for $P^-$,
it results in
$$\|B_{ij}\|_{s-\sigma}(ij)^{-\beta}(1+|i-j|)(i^{\iota-1}+j^{\iota-1})\le\frac{4C_1}{\gamma^-\sigma^{a_4}} \exp{\left(\frac{C_2}{(\gamma^-)^{a_5}\sigma^{a_3}}\right)}\|P^-\|_{\beta,s},\quad \forall i\ne j.$$
When $i=j$, $B_{ij}(\phi)=0$. Therefore, $B(\phi)\in\mathcal{M}_\beta^+$ for any $|\Im\phi|<s-\sigma$ and
$$\|B(\phi)\|_{\beta,s-\sigma}^+\le\frac{4C_1}{\gamma^-\sigma^{a_4}} \exp{\left(\frac{C_2}{(\gamma^-)^{a_5}\sigma^{a_3}}\right)}\|P^-\|_{\beta,s}.$$
A similar computation follows
for $\omega\ne\omega'$, $i\ne j$,
\begin{align*}
&\frac{\|\Delta B_{ij}\|_{s-3\sigma}}{|\omega-\omega'|}(ij)^{-\beta}(1+|i-j|)(i^{\iota-1}+j^{\iota-1})\\
\le&\frac{C}{(\gamma^-)^2\sigma^{2a_4+1}} \exp{\left(\frac{2C_2}{(\gamma^-)^{a_5}\sigma^{a_3}}\right)}\|P^-\|_{\beta,s}^\mathcal{L}.
\end{align*}
Thus,
\[
\|B(\phi)\|_{s-3\sigma}^{\mathcal{L},+}\le\frac{C(\gamma_0,n,\tau)}{(3\sigma)^{2a_4+1}} \exp{\left(\frac{C(\gamma_0,\beta,\iota,n,\tau)}{(3\sigma)^{a_3}}\right)}\|P^-\|_{\beta,s}^\mathcal{L},
\]
where $a_4=n+\tau,~a_5=\frac{1}{1-\theta},~a_3=n+\tau+\frac{\theta(n+\tau+2)}{1-\theta}>1$. So if we denote $m=[2a_4+3]$ and choose $3\sigma\le\min\{\frac{1}{m!C(\gamma_0,n,\tau)},1\}$, then
$\frac{C(\gamma_0,n,\tau)}{(3\sigma)^{2a_4+1}}\le\frac1{(3\sigma)^mm!}\le\exp{\Big(\frac1{3\sigma}\Big)}$.
Redefine $3\sigma$ as $\sigma$ one obtains
\begin{equation}\label{BestiMbeta}
\|B(\phi)\|_{\beta,s-\sigma}^{\mathcal{L},+}\le \expc\|P^-\|_{\beta,s}^\mathcal{L},
\end{equation}
where $C=C(\gamma_0,\beta,\iota,n,\tau)$.
\end{proof}
In fact by Cauchy's estimate and  (ii) of Lemma \ref{algebrapara} and (\ref{BestiMbeta}) we have
\begin{lemma}\label{gujiforPlusIII}
\begin{align*}
&\|\dot B\|_\besn2^{\L,+}\le\frac{C(n)}\sigma\|B\|_\besn{}^{\L,+},\\
&\|\frac{d}{dt}e^B\|_\besn2^{\L,+}\le\frac{C(\beta,n)}\sigma e^{C(\beta)\|B\|_\besn{}^+}\|B\|_\besn{}^{\L,+},\\
&\|\frac{d}{dt}(e^B-\id-B)\|_\besn2^{\L,+}\le\frac{C(\beta,n)}\sigma e^{C(\beta)\|B\|_\besn{}^+}(\|B\|_\besn{}^{\L,+})^2.
\end{align*}
\end{lemma}
\begin{lemma}\label{kaojing1}
For any $|s_1|\le1$, if $\|B\|_\besn{}^{\L,+}\ll 1$, then
$$\|e^{s_1B}\|_{\lzt,s-2\sigma}^\L, \|e^{s_1B}\|_{\ldt,s-2\sigma}^\L \leq 2. $$
\end{lemma}
\begin{proof}
If  $\|B\|_\besn{}^{\L,+}\ll 1$, then
by   (iv) of  Lemma \ref{algebra} and  Lemma \ref{algebrapara}
\begin{align*}
\|e^{s_1B}\|_{\lzt,s-2\sigma}^\L&=\|e^{s_1B}-\id+\id\|_{\lzt,s-2\sigma}^\L
\le1+\|e^{s_1B}-\id\|_{\lzt,s-\sigma}^\L\\
&\le1+C(\beta,\iota)\|e^{s_1B}-\id\|_{\besn{}}^{\L,+}
\le1+C(\beta,\iota)e^{C\|B\|_\besn{}^{+}}\|B\|_\besn{}^{\L,+}\\
&\le 2.
\end{align*}
Similarly, by $0\le\delta<\iota-\beta-\frac12$, (v) of  Lemma \ref{algebra} and Lemma \ref{algebrapara},  if $\|B\|_\besn{}^{\L,+}\ll 1$,
$
\|e^{s_1B}\|_{\ldt,s-2\sigma}^\L\le2.
$
\end{proof}
\begin{lemma}\label{iterule}
Consider the system
\begin{equation}\label{sysminus}
\rmi\dot x=(A^-+P^-(\ot))x
\end{equation}
with the stated assumptions. Assume furthermore that also \eqref{xianzhitiaojian1} holds. Then there exists an anti-self-adjoint operator $B\in\M_\besn{}^+$ analytically depending on $\phi\in\T_\besn{}^n$, and Lipschitz continuous in $\omega\in\Pi^-$ such that\\
(1). $B$ fulfills the estimate  \eqref{Bestimate};\\
(2). For any $\omega\in\Pi^-$ the unitary operator $e^{B(\ot)}$ transforms the system \eqref{sysminus} into the system
\[
\rmi\dot y=(A^++P^+(\ot))y;
\]
(3). The new perturbation $P^+$ fulfills the estimate
\[
\|P^+\|_{\besn 2}^\L+\|P^+\|_{\lzd,s-2\sigma}^\L\le\exptc\big(\|P^-\|_{\beta,s}^\L+\|P^-\|_{\lzd,s}^\L\big)^2,
\]
where we assume that $\|B\|_\besn{}^\L\le\frac1{C(\beta)}$ and $0<\sigma\le\frac1{C(n,\beta)}<1$ and $\tau>1$;\\
(4). For any positive $K$ such that $2(\|P^-\|_{\beta,s}^\L+\|P^-\|_{\lzd,s}^\L)(1+K^\tau)<\gamma^--\gamma^+$, there exists a closed set $\Pi^+\subset\Pi^-$ fulfilling
\[
|\Pi^-\backslash\Pi^+|\le\frac{C\gamma_0}{K^{\tau-n-\frac2{\iota-1}}},
\]
under the condition \eqref{Clambda},\eqref{gammazero},\eqref{Comla} and \eqref{tagaK};\\
(5). If $\omega\in\Pi^+$ assumptions $H1 - H3$ above are fulfilled by $A^+$  and $P^+$ provided that the constants are replaced by the new ones defined by
\begin{equation}\label{iteration}
\begin{aligned}
\gapl&=\gami-2\big(\|P^-\|_{\beta,s}^\L+\|P^-\|_{\lzd,s}^\L\big)(1+K^\tau),\\
\clpl&=\clmi-2\big(\|P^-\|_{\beta,s}^\L+\|P^-\|_{\lzd,s}^\L\big),\\
\copl&=\comi+2\big(\|P^-\|_{\beta,s}^\L+\|P^-\|_{\lzd,s}^\L\big),\\
\cmpl&=\cmmi+2\big(\|P^-\|_{\beta,s}^\L+\|P^-\|_{\lzd,s}^\L\big).
\end{aligned}
\end{equation}
\end{lemma}
\begin{remark}\label{iteremark}
If $\|P^-\|_{\beta,s}^\L+\|P^-\|_{\lzd,s}^\L\le\ve^-,$ we can choose
\[
\gamma^+=\gamma^--2\ve^-(1+K^\tau),~\clpl=\clmi-2\ve^-,~\copl=\comi+2\ve^-,~\cmpl=\cmmi+2\ve^-
\]
in place of \eqref{iteration}.
\end{remark}
\begin{proof}
Similar as \cite{BG}, we can prove that $B$ is anti-self-adjoint operator and $e^{B(\ot)}$ is a unitary operator and (1) and (2) follow easily. For (3), we write  the new perturbation $P^+:=(I)+(II)+(III)$, where
$(I)=e^{-B}A^-e^B-A^--[A^-,B]$,
$(II)=e^{-B}P^-e^B-P^-$ and $(III)=-\rmi(e^{-B}\frac{d}{dt}e^B-\dot B)$.
We first estimate (I). By \eqref{homoeqn} one has
\[
[[A^-,B],B]=[\rmi\dot B-(P^--\text{diag}(P^-)),B]=\rmi[\dot B,B]-[\pmdp,B].
\]
For any $|s_2|\le 1$,  from  (iii) of Lemma \ref{algebrapara} we have
\begin{align*}
&~\|e^{-s_2B}[[A^-,B],B]e^{s_2B}\|_\besn2^\L\\
\le&~C(\beta)\|[\dot B,B]\|_\besn2^\L
+C(\beta)\|[\pmdp,B]\|_\besn2^\L\\
\le & C(\beta,n)\expc\big(\|P^-\|_{\beta,s}^{\L}\big)^2.
\end{align*}
Since
\[
(I)=e^{-B}A^-e^B-A^--[A^-,B]=\int_0^1ds_1\int_0^{s_1}e^{-s_2B}[[A^-,B],B]e^{s_2B}ds_2,
\]
it follows
\[
\|(I)\|_\besn2^\L\le\int_0^1ds_1\int_0^{s_1}\|e^{-s_2B}[[A^-,B],B]e^{s_2B}\|_{\besn2}^{\L}ds_2
\le C(\beta,n)\expc\big(\|P^-\|_\besn2^\L\big)^2.
\]
If $\|B\|_\besn{}^{\L,+}\le\frac1{C(\beta)} $,  by   (iii) of Lemma \ref{algebrapara} one has
$
\|(II)\|_\besn{2}^{\L} \le C_1(\beta)\exp{\Big(\frac{C}{\sigma^{a_3}}\Big)}\big(\|P^-\|_{\beta,s}^\L\big)^2
$
and
$
\|(III)\|_\besn2^\L \le C(\beta,n)\expc\big(\|P^-\|_{\beta,s}^{\L}\big)^2.
$
Combining all the above estimates, if $\|B\|_\besn{}^{+, \L}\le\frac1{C(\beta)}$ and $0<\sigma<1$, then
\begin{eqnarray}\label{PPlus1}
\|P^+\|_\besn2^\L \le C(\beta,n)\expc\big(\|P^-\|_\besn2^\L\big)^2,
\end{eqnarray}
where $a_3=n+\tau+\frac{\theta(n+\tau+2)}{1-\theta}$ and $\tau>1$. In the following we turn to the estimate on $\|P^+\|_{\lzd,s-2\sigma}^\L$.\\
\indent By \eqref{homoeqn}, Cauchy's estimate and (iv) of  Lemma \ref{algebra} we have
\begin{align*}
\|[A^-,B]\|_{\lzd,s-2\sigma}^\L &\le\|\dot B\|_{\lzt,s-2\sigma}^\L+2\|P^-\|_{\lzd,s}^\L\\
&\le\frac{C(\beta,n,\iota)}{\sigma}\|B\|_\besn{}^{\L,+}+2\|P^-\|_{\lzd,s}^\L,
\end{align*}
together with \eqref{BestiMbeta}, (iv),  (v) of Lemma \ref{algebra}  we obtain
\begin{align*}
&\|[[A^-,B],B]\|_{\lzd,s-2\sigma}^\L\le\|[A^-,B]B\|_{\lzd,s-2\sigma}^\L+\|B[A^-,B]\|_{\lzd,s-2\sigma}^\L\\
\le~& C(\beta,\iota)\|B\|_\besn{}^{\L,+}\|[A^-,B]\|_{\lzd,s-2\sigma}^\L\\
\le~&C(\beta,n,\iota)\expc\|P^-\|_{\beta,s}^{\L}\big(\|P^-\|_{\beta,s}^{\L}+\|P^-\|_{\lzd,s}^\L\big).
\end{align*}
Hence,
\begin{align}
\|(I)\|_{\lzd,s-2\sigma}^\L&\le\int_0^1ds_1\int_0^{s_1}\|e^{-s_2B}[[A^-,B],B]e^{s_2B}\|_{\lzd,s-2\sigma}^\L ds_2\notag\\
&\le C(\beta,n,\iota)\expc\|P^-\|_{\beta,s}^{\L}\big(\|P^-\|_{\beta,s}^{\L}+\|P^-\|_{\lzd,s}^\L\big).\label{Psupsub}
\end{align}
By \eqref{BestiMbeta}, $0\le\delta<\iota-\beta-\frac12$ and (iv), (v) of Lemma \ref{algebra},   we have
\begin{align*}
&\|[P^-,B]\|_{\lzd,s-2\sigma}^\L\le\|P^-B\|_{\lzd,s-2\sigma}^\L+ \|BP^-\|_{\lzd,s-2\sigma}^\L\\
\le& C(~\|P^-\|_{\lzd,s-2\sigma}^\L\|B\|_{\lzt,s-\sigma}^\L+\|B\|_{\ldt,s-\sigma}^\L\|P^-\|_{\lzd,s-2\sigma}^\L)\\
\le&~C(\beta,\iota)\expc\|P^-\|_{\beta,s}^\L\|P^-\|_{\lzd,s}^\L
\end{align*}
and then from Lemma \ref{kaojing1}, one can draw
\begin{align*}
\|(II)\|_{\lzd,s-2\sigma}^\L
&=\Big\|\int_0^1e^{-s_1B}[P^-,B]e^{s_1B}ds_1\Big\|_{\lzd,s-2\sigma}^\L\\
&\le C(\beta,\iota)\expc\|P^-\|_{\beta,s}^\L\|P^-\|_{\lzd,s}^\L.
\end{align*}
Similarly,  with   (iv) of Lemma \ref{algebra}, one obtains
\begin{align*}
\|(III)\|_{\lzd,s-2\sigma}^\L&\le
\|(III)\|_{\lzt,s-2\sigma}^\L\le C(\beta,\iota)\|(III)\|_\besn2^{\L,+}\\
&\le C(\beta,n,\iota)\expc\big(\|P^-\|_{\besn2}^\L\big)^2.
\end{align*}
Thus, if $\|B\|_\besn{}^\L\ll1$ and $0<\sigma<1$, we have
\begin{eqnarray}\label{PPlus2}
\|P^+\|_{\lzd,s-2\sigma}^\L \le C(\beta,n,\iota)\expc\|P^-\|_{\beta,s}^{\L}\big(\|P^-\|_{\beta,s}^{\L}+\|P^-\|_{\lzd,s}^\L\big),
\end{eqnarray}
where $a_3=n+\tau+\frac{\theta(n+\tau+2)}{1-\theta}$ and $\tau>1$.
Combining with (\ref{PPlus1}) and (\ref{PPlus2}), we obtain (3). In the following we turn to prove (4) and  (5). \\
For H1.a)
As \cite{BG}, if choose $\clpl=\clmi-2\big(\|P^-\|_{\beta,s}^\L+\|P^-\|_{\lzd,s}^\L\big)$, then $|\lambda_i^+-\lambda_j^+|\ge\clpl|i^\iota-j^\iota|$.\\
H1.b) Choose $\copl=\comi+2(\|P^-\|_{\beta,s}^\L+\|P^-\|_{\lzd,s}^\L)$, then
$
\sup\limits_{\substack{\omega,\omega'\in \Pi^{-} \\ \omega\neq \omega' }} \frac{|\lambda_i^+(\omega)-\lambda_i^+(\omega')|}{|\dom|}\le\copl i^{2\beta}
$
holds true. \\
H1.c)
Similarly, if $\cmpl=\cmmi+2\big(\|P^-\|_{\beta,s}^\L+\|P^-\|_{\lzd,s}^\L\big)$, then $\|\mu_i^+\|_\s\le\cmpl i^{2\beta}$.
If choose $\copl=\comi+2\big(\|P^-\|_{\beta,s}^\L+\|P^-\|_{\lzd,s}^\L\big)$,  we have
$\sup\limits_{\omega, \omega'\in \Pi^{-}, \omega\neq \omega'}
\frac{\|\mu_i^+(\phom)-\mu_i^+(\phomp)\|_\s}{|\dom|}\le\copl i^{2\beta}$. \\
H3): To check $H3$ for next step, one need to throw away suitable parameter sets. This step is very similar as \cite{BG} and we only give a sketch here.
In fact, if choose $\gamma^+=\gamma^--2\big(\|P^-\|_{\beta,s}^\L+\|P^-\|_{\lzd,s}^\L\big)(1+K^\tau)$, then for any $|k|\le K,~i\ne j$
\begin{equation}\label{secMenikov}
|\lambda_i^+-\lambda_j^++\lra k\omega|\ge\frac{\gamma^+|i^\iota-j^\iota|}{1+|k|^\tau},
\end{equation}
and
$$
|\lra k\omega|\ge\frac{\gamma^+}{|k|^\tau} , \qquad    \forall~k\in\Z^n\backslash\{0\}.$$
For $|k|> K$ and $i\ne j$, we need to throw away a suitable parameter set in $\Pi^-$ to guarantee \eqref{secMenikov} holds true. Clearly, a standard procedure  shows us
\begin{align*}
&\Pi^-\backslash\Pi^+ := \bigcup_{|k|>K}\bigcup_{\substack{i\ne j\\i,j\ge1}}\wt{\cal R}_{ij}^k(\omega):=\bigcup_{|k|>K}\bigcup_{\substack{i\ne j\\i,j\ge1}}\Big\{\omega\in\Pi^-:|\lambda_i^+-\lambda_j^++\lra k\omega|<\frac{\gamma^+|i^\iota-j^\iota|}{1+|k|^\tau}\Big\}\\
&\subset\bigcup_{|k|>K}\bigcup_{\substack{i\ne j\\i,j\ge1}}\Big\{\omega\in\Pi:|\lambda_i^+-\lambda_j^++\lra k\omega|<\frac{\gamma^+|i^\iota-j^\iota|}{1+|k|^\tau}\Big\}
:=\bigcup_{|k|>K}\bigcup_{\substack{i\ne j\\i,j\ge1}}\cal R_{ij}^k(\omega).\notag
\end{align*}
Next we will estimate the measure for $\bigcup\limits_{|k|>K}\bigcup\limits_{\substack{i\ne j\\i,j\ge1}}\cal R_{ij}^k(\omega)$.
Similar as \cite{BG}, if \begin{equation}\label{Clambda}
C_{\lambda,l}\ge\frac12C_\lambda,\quad\text{for any}~l,
\end{equation}
 and
 \begin{equation}\label{gammazero}
\gamma_l<\frac14C_\lambda,\quad\text{for any}~l,
\end{equation}
 and
\begin{equation}\label{Comla}
\copl<\frac{C_\lambda}{16n},
\end{equation}
and
\begin{equation}\label{tagaK}
\tau>n+\frac2{\iota-1},~\gamma^+\le\gamma_0~\text{and}~K\ge2
\end{equation}
are fulfilled, then
$
\text{meas}(\Pi^-\backslash\Pi^+)\le\frac{C\gamma_0}{K^{\tau-n-\frac2{\iota-1}}}.
$
\end{proof}

\subsection{Iteration}
In this section we set up the iteration. First we preassign the value of the various constants. Hence we keep $\ve_0$, $K$, $s_0$ and $\gamma_0$ fixed which satisfy
\begin{equation}\label{initialdata}
0<s_0<1,\quad \ve_0\le\min\{C_*,\exp{(-a_6s_0^{-a_3})}\},\quad \gamma_0<\frac14C_\lambda,\quad K=4^3,
\end{equation}
where $C_*$ and $a_6$ depend on $\beta,C_\lambda,n,\tau,\gamma_0,\iota$ and $a_3\sim n,\tau,\beta,\iota$ and
$\ve(\|P\|_{\beta, s}^{\cal L}+\|P\|^{\cal L}_{\cal{B}(\ell_0^2, \ell_{-2s}^2)}): = \ve_0$.  \\
For $l\ge1$ we define
\begin{align*}
&\ve_{l}=\ve_{l-1}^\frac43,\quad\sigma_l=\left(\frac{3C}{|\ln{\ve_{l-1}}|}\right)^{\frac1{a_3}},\quad s_l=s_{l-1}-2\sigma_l,\quad
K_l=lK,\quad\gamma_l=\gamma_{l-1}-2\ve_{l-1}(1+K_l^\tau),\\
&C_{\lambda,l}=C_{\lambda,l-1}-2\ve_{l-1},\quad C_{\omega,l}=C_{\omega,l-1}+2\ve_{l-1},\quad C_{\mu,l}=C_{\mu,l-1}+2\ve_{l-1}.
\end{align*}
The initial values of the sequences are chosen as follows:
$$\gamma_0:=\gamma,~s_0=s,~C_{\lambda,0}:=C_\lambda,~C_{\omega,0}:=0,~C_{\mu,0}:=0.$$
From these settings, we can obtain that for any $l\ge0$,
a). $\frac{\gamma_0}2\le\gamma_l\le\gamma_0<\frac14 C_\lambda$,
b). $0<\sigma_l\le\min\{C(\gamma_0,n,\tau),\frac1{C(n,\beta)}\}<1$,
c). $0\le C_{\omega,l}\le\min\{\frac{C_\lambda}{16n},1\}$,
d). $C_{\lambda,l}\ge\frac12 C_\lambda$,
e). $K_l\ge2$ with $l \geq 1$,
 and f). $0\leq \ve_l\le\frac1{C(\beta)}.$
\begin{Proposition}
There exist $\ve_*=\ve_*(\gamma,s)>0$ and, for any $l\ge1$, a closed set $\Pi_l^\gamma\subset\Pi$ such that, if $0\le\e<\ve_*$, one can construct for $\omega\in\Pi_l^\gamma$ a unitary transformation $U^l$, analytic and quasi-periodic in $t$ with frequencies $\omega$, mapping the system
${\rm i}\dot{x}=(A+\ve P(\omega t))x$,
into the system
\begin{equation}\label{iterationsystem}
{\rm i}\dot{x}=(A^l+P^l(\omega t))x,
\end{equation}
where\\
(1). $U^l(\omega t)$ is as follow: $U^l(\ot)= e^{B^1(\ot)}e^{B^2(\ot)}\cdots e^{B^l(\ot)}$, and the anti-self-adjoint operator $B^j\in\mathcal{M}_\beta^+$ depending analytically on $\phi\in\T_{s_{j-1}-\sigma_j}^n$, are Lipschitz continuous in $\omega\in\Pi_l^\gamma$ and fulfilling \eqref{Bestimate}  with $P^{j-1},~s_{j-1},~\sigma_j$ in place of $P^-,~s,~\sigma$, respectively.\\
(2). $A^l$ has the form of \eqref{formalA} with the upper index ``minus" replaced by $l$, i.e.
$$A^l:=\text{diag}\{\lambda_1^l(\omega)+\mu_1^l(\omega t,\omega),\lambda_2^l(\omega)+\mu_2^l(\omega t,\omega),\cdots\}.$$
(3). The corresponding $\lambda_i^l$ and $\mu_i^l$ fulfill conditions H1, H3 of the previous section, provided $\lambda_i^-$, $\mu_i^-$ are replaced by $\lambda_i^l$, $\mu_i^l$, respectively.\\
(4). $P^l$ fulfills condition H2 with the upper index ``minus" replaced by $l$
and the following estimates hold:
$
\|P^l\|_{\beta,s_l}^\L+\|P^l\|_{\lzd,s_l}^\L\le \exp{\Big(\frac{2C}{\sigma_l^{a_3}}\Big)}\Big(\|P^{l-1}\|_{\beta,s_{l-1}}^\L+\|P^{l-1}\|_{\lzd,s_{l-1}}^\L\Big)^2\le \ve_l$,
$\|B^l\|_{\beta,s_{l-1}-\sigma_l}^{\L,+}\le \exp{\Big(\frac{C}{\sigma_l^{a_3}}\Big)}\Big(\|P^{l-1}\|_{\beta,s_{l-1}}^\L+\|P^{l-1}\|_{\lzd,s_{l-1}}^\L\Big)\le\ve_{l-1}^{\frac23}$
and
$
|\Pi_l^\gamma\backslash\Pi_{l+1}^\gamma|\le\frac{C\gamma}{K_{l+1}^{b_1}}=\frac{C\gamma}{((l+1)K)^{b_1}}
$
with $b_1=\tau-n-\frac{2}{\iota-1}>1$.
\end{Proposition}
\begin{proof}
We proceed by induction applying Lemma \ref{iterule}. First we apply it to the original system to obtain the system \eqref{iterationsystem} for $l=1$. To this end we notice that all assumptions are satisfied except the non-resonance conditions H3 on the frequencies. We define
\begin{align*}
U_1^0(\gamma)&=\bigcup_{k\ne0}\Big\{\omega\in\Pi:=[0,1]^n:|\la k,\omega\ra|<\frac{\gamma}{|k|^\tau}\Big\},\\
U_2^0(\gamma)&=\bigcup_{k\in\Z^n}\bigcup_{\substack{i,j\ge1\\i\ne j}} \Big\{\omega\in\Pi:|\lambda_i-\lambda_j+\la k,\omega\ra|<\frac{\gamma|i^\iota-j^\iota|}{1+|k|^\tau}\Big\}
\end{align*}
and
$V_0(\gamma)=U_1^0(\gamma)\bigcup U_2^0(\gamma)$.
When $\tau>n+\frac2{\iota-1}$, we obtain
$\text{meas}(V_0(\gamma))\le C(n,\iota)\gamma$.
Denote $\Pi_0=\Pi\backslash V_0(\gamma)$, for $\omega\in\Pi_0$, if the initial data are given suitably as (\ref{initialdata}), then we can apply Lemma \ref{iterule} and the starting point of our induction is established.\\
Now we assume that
\begin{equation}\label{l-thsystem}
{\rm i}\dot{x}=(A^l+P^l(\omega t))x
\end{equation}
and all assumptions are satisfied for Lemma \ref{iterule}. Then there exists an anti-self-adjoint operator $B^{l+1}$ which satisfies
\[
\|B^{l+1}\|_{\beta,s_l-\sigma_{l+1}}^{\L,+}\le \expcl\Big(\|P^l\|_{\beta,s_l}^\L+\|P^l\|_{\lzd,s_l}^\L\Big)
\le\e_l^{-\frac13}\e_l=\e_l^\frac23,
\]
where $B^{l+1}$ is analytically depending on $\phi\in\T_{s_l-\sigma_{l+1}}^n$ and Lipschitz continuous in $\omega\in\Pi_l^\gamma$. By the unitary operator $e^{B^{l+1}(\omega t)}$ transforms the system \eqref{l-thsystem} into the system
${\rm i}\dot{x}=(A^{l+1}+P^{l+1}(\omega t))x$,
the new perturbation $P^{l+1}$ fulfills
$$
\|P^{l+1}\|_{\beta,s_{l+1}}^\L+\|P^{l+1}\|_{\lzd,s_{l+1}}^\L\le \exptcl\left(\|P^l\|_{\beta,s_l}^\L+\|P^l\|_{\lzd,s_l}^\L\right)^2\le \ve_{l+1}.
$$
Moreover, there exists a closed set $\Pi_{l+1}^\gamma\subset\Pi_l^\gamma$ and $b_1=\tau-n-\frac{2}{\iota-1}>1$ fulfilling
$|\Pi_l^\gamma\backslash\Pi_{l+1}^\gamma|\le\frac{C\gamma}{K_{l+1}^{b_1}}$.
If $\omega\in\Pi_{l+1}^\gamma$, then assumptions H1-H3 are fulfilled by $ A^{l+1},~P^{l+1}$ provided that the constants are replaced by the new ones defined by
\[
\gamma_{l+1}=\gamma_l-2\ve_l(1+K_{l+1}^\tau),~C_{\lambda,l+1}=C_{\lambda,l}-2\ve_l,~ C_{\omega,l+1}=C_{\omega,l}+2\ve_l,~C_{\mu,l+1}=C_{\mu,l}+2\ve_l.
\]
\end{proof}
In the following we will prove Theorem \ref{Theorem2.10}, but we first need a series of preparation lemmas.
\begin{lemma}\label{Alconverge}
$A^l(\phi)\to A^\infty(\phi)$ in $\M_{\frac\iota2}(\Pi_*,\frac s2)$ and $\lzi(\Pi_*,\frac s2)$, if $\e_0\ll1$, then
\begin{align*}
&\|A^\infty\|_{\lzi,\frac s2}\le\|A^\infty\|_{\frac\iota2,\frac s2},\\
&\|A^\infty-A^0\|_{\lzi,\frac s2}\le\|A^\infty-A^0\|_{\frac\iota2,\frac s2}\le6\e_0.
\end{align*}
\end{lemma}
\begin{proof}
From the iteration, for $l\ge0$ one has
\begin{eqnarray}\label{piaoyiguji1}
|\lambda_i^{l+1}-\lambda_i^{l}|\le\|P^l\|_{\beta,s_{l}}i^{2\beta}\le\ve_{l}i^{2\beta}
\end{eqnarray}
and
\begin{eqnarray}\label{piaoyiguji2}
|\mu_i^{l+1}-\mu_i^{l}|\le2\|P^l\|_{\beta,s_{l}}i^{2\beta}\le2\ve_{l}i^{2\beta}.
\end{eqnarray}
Thus, for given $i\in\Z_+$, $\{\lambda_i^l\}$ and $\{\mu_i^l\}$ are Cauchy sequences. Define
$
\lambda_i^\infty:=\lambda_i+\sum_{l=0}\lambda_i^{l+1}-\lambda_i^l $ and
$\mu_i^\infty:=\sum_{l=0}\mu_i^{l+1}-\mu_i^l$.
For $l\ge0$, we have
$
|\lambda_i^\infty-\lambda_i^l|  \le 2\e_li^{2\beta}$ and
$|\mu_i^\infty-\mu_i^l|\le  4\e_li^{2\beta}$.
Note $2\beta<\iota-1<\iota$, then
$
|(\lambda_i^\infty+\mu_i^\infty)-(\lambda_i^l+\mu_i^l)| \le 6\e_li^\iota.
$
It follows
\(
\|A^\infty-A^l\|_{\frac\iota2,\frac s2}\le6\e_l,
\)
together with Lemma \ref{normmatope} and thus,
\[
\|A^\infty-A^l\|_{\lzi,\frac s2}\le\|A^\infty-A^l\|_{\frac\iota2,\frac s2}\le6\e_l\to0\qquad(\text{as }l\to\infty).
\]
It follows that $A^l\to A^\infty$ in $\M_{\frac\iota2}(\Pi_*,\frac s2)$  and $\lzi(\Pi_*,\frac s2)$.
Set $l=0$, we have
\[
\|A^\infty-A^0\|_{\lzi,\frac s2}\le\|A^\infty-A^0\|_{\frac\iota2,\frac s2}\le6\e_0.
\]
Then for any $(\phom)\in\T_\frac s2^n\times\Pi_*$, one obtains $A^\infty(\phom)\in\M_{\frac\iota2}\bigcap\lzi$ and
$
\|A^\infty\|_{\lzi,\frac s2}\le\|A^\infty\|_{\frac\iota2,\frac s2}.
$
\end{proof}
\begin{remark}
From the iteration as \cite{BG} one can also have
$
|\lambda_i^{l+1}-\lambda_i^{l}|\le\ve_{l}i^{\delta}$ and
$
|\mu_i^{l+1}-\mu_i^{l}| \le2\ve_{l}i^{\delta}$. But note the assumption $\delta\geq 2\beta$, these estimates are weaker than
(\ref{piaoyiguji1}) and (\ref{piaoyiguji2}).
\end{remark}

Recall $\|P^l\|_{\beta,s_l}^{\Pi_l,\L}+\|P^l\|_{\lzd,s_l}^{\Pi_l,\L}\le\ve_l$,
and $s_l\ge\frac s2$ and $\Pi_*=\bigcap_{l\ge0}\Pi_l$, it follows
\begin{lemma}\label{Plconverge}
$P^l(\phi)\to0$ in $\M_\beta(\Pi_*,\frac s2)$ and $\lzd(\Pi_*,\frac s2)$.
\end{lemma}
\par

In the sequel, let $B^0(\phi)={\bf0}$ and $\ve_{-1}=0$.
From Lemma \ref{algebrapara} and the induction, we have
\begin{lemma}\label{UlId}
For $l\ge0$,~$U^l(\phi)=e^{B^0(\phi)}e^{B^1(\phi)}\cdots e^{B^l(\phi)}$, if $\e_0\ll1$, then
\[
\|U^l(\phi)-\id\|_{\beta,\frac s2}^{\Pi_*,\L,+}\le3(\ve_0^{\frac23}+\cdots+\ve_{l-1}^{\frac23}).
\]
\end{lemma}
Similarly,
\begin{lemma}\label{expBlId}
For $0\le l_1<l_2$, if $\e_0\ll1$, $\|e^{B^{l_1+1}}\cdots e^{B^{l_2}}-\id\|_{\beta,\frac s2}^{\Pi_*,\L,+}\le3(\ve_{l_1}^\frac23+\cdots+\ve_{l_2-1}^\frac23)$.
\end{lemma}
\begin{lemma}\label{UlCauchy}
For $0\le l_1<l_2$, if $\e_0\ll1$, then $\|U^{l_1}(\phi)-U^{l_2}(\phi)\|_{\beta,\frac s2}^{\Pi_*,\L,+}\le8\ve_{l_1}^\frac23$.
\end{lemma}
\begin{proof}
Note $U^{l_1}(\phi)=e^{B^0(\phi)}e^{B^1(\phi)}\cdots e^{B^{l_1}(\phi)},~U^{l_2}(\phi)=U^{l_1}(\phi)e^{B^{l_1+1}(\phi)}\cdots e^{B^{l_2}(\phi)}$, it  follows
\begin{align*}
U^{l_1}(\phi)-U^{l_2}(\phi)&=U^{l_1}(\phi)(\id-e^{B^{l_1+1}(\phi)}\cdots e^{B^{l_2}(\phi)})\\
&=(U^{l_1}(\phi)-\id)(\id-e^{B^{l_1+1}(\phi)}\cdots e^{B^{l_2}(\phi)})+(\id-e^{B^{l_1+1}(\phi)}\cdots e^{B^{l_2}(\phi)}).
\end{align*}
Then by Lemma \ref{algebrapara}, Lemma \ref{UlId} and Lemma \ref{expBlId}, we have
\begin{align*}
&\|U^{l_1}(\phi)-U^{l_2}(\phi)\|_{\beta,\frac s2}^{\L,+}\\
\le&~C(\beta)\|U^{l_1}(\phi)-\id\|_{\beta,\frac s2}^{\L,+}\|\id-e^{B^{l_1+1}(\phi)}\cdots e^{B^{l_2}(\phi)}\|_{\beta,\frac s2}^{\L,+}+\|\id-e^{B^{l_1+1}(\phi)}\cdots e^{B^{l_2}(\phi)}\|_{\beta,\frac s2}^{\L,+}\\
\le&~4(\ve_{l_1}^\frac23+\cdots+\ve_{l_2-1}^\frac23)\le8\ve_{l_1}^\frac23.
\end{align*}
\end{proof}
\indent For $l_2>l_1\ge0$, if $\e_0\ll1$, by Lemma \ref{UlCauchy} one has
\begin{eqnarray}\label{dingyiUinfty}
\|U^{l_1}-U^{l_2}\|_{\beta,\frac s2}^{\Pi_*}\le
\|U^{l_1}-U^{l_2}\|_{\beta,\frac s2}^{\Pi_*,+}\le 8\e_{l_1}^\frac23\to0\quad(\text{as }l_1\to\infty).
\end{eqnarray}
It follows that $\{U^l\}$ is a Cauchy sequence in $\M_\beta(\Pi_*,\frac s2)$. Define $U^\infty=\lim\limits_{l\to \infty}U^l$, then by
(\ref{dingyiUinfty}), we have
\begin{eqnarray}\label{Ulminus}
\|U^\infty-U^l\|_{\beta,\frac s2}^{\Pi_*,+}\le8\e_l^\frac23.
\end{eqnarray}
The following lemma is clear by (iv), (v) of Lemma \ref{algebra} and  $\delta\in[0,\iota-\beta-\frac12)$.
\begin{lemma}\label{Ulconverge}
$U^l(\phi) \to U^\infty(\phi)$ in $\M_\beta(\Pi_*,\frac s2)$, $\lzt(\Pi_*,\frac s2)$ and  $\ldt(\Pi_*,\frac s2)$, if $\e_0\ll1$, then
\begin{align*}
&\|U^\infty-\id\|_{\lzt,\frac s2}^{\Pi_*},~\|U^\infty-\id\|_{\ldt,\frac s2}^{\Pi_*}
\le C\|U^\infty-\id\|_{\beta,\frac s2}^{\Pi_*,+}\le C\e_0^\frac23,\\
&\|U^\infty\|_{\lzt,\frac s2}^{\Pi_*},~\|U^\infty\|_{\ldt,\frac s 2}^{\Pi_*}\le 1+C\e_0^\frac23,
\end{align*}
where $U^\infty(\phi)$ is real analytic on $\T_\frac s2^n$ and Lipschitz continuous on $\omega\in\Pi_*$.
\end{lemma}
In the following we denote $V^l(\phi)=e^{-B^l(\phi)}e^{-B^{l-1}(\phi)}\cdots e^{-B^0(\phi)}$ for $l\ge0$.
\begin{lemma}\label{VlId}
For $l\ge0$, if $\e_0\ll1$, then $\|V^l(\phi)-\id\|_{\beta,\frac s2}^{\Pi_*,\L,+}\le3(\ve_0^{\frac23}+\cdots+\ve_{l-1}^{\frac23})$.
\end{lemma}
Similarly, we have
\begin{lemma}\label{VlCauchy}
For $0\le l_1<l_2$,  if $\e_0\ll1$, then $\|V^{l_1}(\phi)-V^{l_2}(\phi)\|_{\beta,\frac s2}^{\Pi_*,\L,+}\le8\ve_{l_1}^\frac23$.
\end{lemma}
It follows that $\{V^l\}$ is a Cauchy sequence in $\M_\beta(\Pi_*,\frac s2)$. Define $V^\infty=\lim\limits_{l\to \infty}V^l$, then we have
\begin{lemma}\label{Vlconverge}
$V^l(\phi)\to V^\infty(\phi)$ in $\mbhs$,  $\lzt(\Pi_*,\frac s2)$ and  $\ldt(\Pi_*,\frac s2)$, if $\e_0\ll1$, then
\begin{align*}
&\|V^\infty-\id\|_{\lzt,\frac s2}^{\Pi_*},~\|V^\infty-\id\|_{\ldt,\frac s2}^{\Pi_*}
\le C\|V^\infty-\id\|_{\beta,\frac s2}^{\Pi_*,+}\le C\e_0^\frac23,\\
&\|V^\infty\|_{\lzt,\frac s2}^{\Pi_*},~\|V^\infty\|_{\ldt,\frac s 2}^{\Pi_*}\le 1+C\e_0^\frac23,
\end{align*}
where $V^\infty(\phi)$ is real analytic on $\T_\frac s2^n$ and Lipschitz continuous on $\omega\in\Pi_*$.
\end{lemma}
The following is clear by (\ref{Ulminus}) and Cauchy's estimate.
\begin{lemma}\label{dtUlj}
For $l\ge0$, if $\e_0\ll1$, then
$\big\|\frac{\partial U^l(\phi)}{\partial\phi_j}-\frac{\partial U^\infty(\phi)}{\partial\phi_j}\big\|_{\beta,\frac{3s}8}^{\Pi_*,+}\le Cs^{-1}\ve_l^\frac23,\quad\forall j\in\{1,2,\cdots,n\}$.
\end{lemma}
By Lemma \ref{dtUlj}, Cauchy estimate and  (iv) and (v) of Lemma \ref{algebra}, we have
\begin{lemma}\label{dUl}
$\frac d{dt}U^l(\omega t)\to\frac d{dt}U^\infty(\omega t)$ in $\cal{M}_{\beta}^{+}(\Pi_{*})$, $\lzt(\Pi_*)$ and  $\ldt(\Pi_*)$ when $l\rightarrow \infty$.
\end{lemma}
From Lemma \ref{dUl}, we have
\begin{Corollary}\label{dtUlconverge}
$\frac d{dt}U^l(\omega t)\to\frac d{dt}U^\infty(\omega t)$ in $\lzi(\Pi_*)$ when $l\rightarrow \infty$.
\end{Corollary}
\begin{lemma}\label{A0Ulconverge}
$A^0U^l(\phi)\to A^0U^\infty(\phi)$ in $\lzi(\Pi_*,\frac s2)$ when $l\rightarrow \infty$.
\end{lemma}
\begin{proof}
From  Lemma \ref{normmatope}, one has $\|A_0\|_{\lzi}\le C$. For any $\phi\in\T_{\frac s2}^n$, by Lemma \ref{Ulconverge} we have
\begin{align*}
\|A^0U^l(\phi)-A^0U^\infty(\phi)\|_{\lzi,\frac s2}&\le\|A^0\|_{\lzi,\frac s2}\|U^l(\phi)-U^\infty(\phi)\|_{\lzt,\frac s2}\\
&\le C\e_l^\frac23\to0\quad\text{as }l\to\infty.
\end{align*}
\end{proof}
\begin{lemma}\label{P0Ulconverge}
$P^0U^l(\phi)\to P^0U^\infty(\phi)$ in $\lzd(\Pi_*,\frac s2)$ and $\lzi(\Pi_*,\frac s2)$ when $l\rightarrow \infty$.
\end{lemma}
\begin{proof}
For any $\phi\in\T_{\frac s2}^n$, by Lemma \ref{Ulconverge} one has
\begin{align*}
&\|P^0(\phi)U^l(\phi)-P^0(\phi)U^\infty(\phi))\|_{\lzi,\frac s2}\le
\|P^0(\phi)U^l(\phi)-P^0(\phi)U^\infty(\phi))\|_{\lzd,\frac s2}\\
\le&~\|P^0\|_{\lzd,\frac s2}\|U^l-U^\infty\|_{\lzt,\frac s2}\le C\ve_0\ve_l^\frac23\to0\quad\text{as }l\to\infty.
\end{align*}
\end{proof}

\begin{lemma}\label{UlPlconverge}
$U^lP^l(\phi)\to 0$ in $\lzd(\Pi_*,\frac s2)$ and $\lzi(\Pi_*,\frac s2)$.
\end{lemma}
\begin{proof}
For any $\phi\in\T_{\frac s2}^n$, by Lemma \ref{UlId} and (v) of Lemma \ref{algebra} we have
\[
\|U^l\|_{\ldt,\frac s2}\le 1+C\|U^l-\id\|_{\beta,\frac s2}^+\le 1+C\e_0^\frac23,
\]
then one obtains
\begin{align*}
&\|U^l(\phi)P^l(\phi)\|_{\lzi,\frac s2}\le \|U^l(\phi)P^l(\phi)\|_{\lzd,\frac s2}\\
\le&~\|U^l\|_{\ldt,\frac s2}\|P^l\|_{\lzd,\frac s2}\le C\e_l\to0\quad\text{as }l\to\infty.
\end{align*}
\end{proof}

\begin{lemma}\label{UA1}
For any $l\ge1$, $U^lA^l\in\lzi$.
\end{lemma}
\begin{proof}
When $l=0$, we have $U^0A^0=A^0\in\lzi$.
Suppose $U^lA^l\in \lzi$, we consider $U^{l+1}A^{l+1}$.
From
$
U^{l+1}A^{l+1}=U^le^{B^{l+1}}A^{l+1}$ and
$A^{l+1}=A^l+\text{diag}(P^l)$,
it results in
\[
e^{B^{l+1}}A^{l+1}=e^{B^{l+1}}A^l+e^{B^{l+1}}\text{diag}(P^l).
\]
Besides,
\[
e^{B^{l+1}}A^le^{-B^{l+1}}-A^l=[B^{l+1},A^l]+\int_0^1ds_1\int_0^{s_1}e^{s_2B^{l+1}}[B^{l+1},[B^{l+1},A^l]]e^{-s_2B^{l+1}}ds_2:=\wt P^l.
\]
 and
 \begin{eqnarray}\label{homologicalequations}
 [A^l,B^{l+1}]-\rmi\dot B^{l+1}+P^l-\text{diag}(P^l)=0.
 \end{eqnarray}
Thus, we have
\begin{align}
e^{B^{l+1}}A^l&=A^le^{B^{l+1}}+\wt P^le^{B^{l+1}},\notag\\
U^{l+1}A^{l+1}&=U^le^{B^{l+1}}\big(A^l+\text{diag}(P^l)\big)=U^le^{B^{l+1}}A^l+U^le^{B^{l+1}}\text{diag}(P^l)\notag\\
&=U^lA^le^{B^{l+1}}+U^l\wt P^le^{B^{l+1}}+U^{l+1}\text{diag}(P^l).\label{UAinduce}
\end{align}
From $U^lA^l\in\lzi $, Lemma \ref{algebra}, we have $U^lA^le^{B^{l+1}}\in \lzi$. From (\ref{homologicalequations}) and assumption B3, we have
$ [A^l,B^{l+1}]\in \cal{B}(\ell_0^2, \ell_{-2\delta}^2)$. From $B^{l+1}\in \cal{M}_{\beta}^+$ and Lemma \ref{algebra}, we have $\wt P^l\in \cal{B}(\ell_0^2, \ell_{-2\delta}^2)$.
Note $e^{B^{l+1}}\in \cal{B}(\ell_0^2)$, we obtain
$\wt P^le^{B^{l+1}}\in \cal{B}(\ell_0^2, \ell_{-2\delta}^2)$.  From above, $U^{l}\in \cal{B}(\ell_{-2\delta}^2)$, it follows $U^l\wt P^le^{B^{l+1}}\in \cal{B}(\ell_0^2, \ell_{-2\delta}^2). $
Similarly, $U^{l+1}\text{diag}(P^l)\in \cal{B}(\ell_0^2, \ell_{-2\delta}^2).$ From above and (\ref{UAinduce}), one can draw $U^{l+1}A^{l+1}\in\lzi$.
 \end{proof}

 \begin{remark}
 If we only assume B1, B2 without B3, one can't prove
$[A^l,B^{l+1}]\in \cal{B}(\ell_0^2, \ell_{-2\delta}^2)$ without any further assumption.
 \end{remark}

 \begin{lemma}\label{UA2}
For any $l\ge1$, if $\e_0\ll1$, then one has
\[
||U^lA^l||_{\lzi,s_l}\le C\prod_{i=0}^{l-1}(1+\e_i^\frac12)+\sum_{i=0}^{l-1}\e_i^\frac12.
\]
\end{lemma}
\begin{proof}  From the prove of Lemma \ref{iterule}, one has $\|\wt P^l\|_{\lzd,s_{l+1}}\le C\e_l^\frac23$ (see \eqref{Psupsub}).\\
When $l=0$, we have $U^0A^0=A^0$ and $\|A^0\|_{\lzi,s_0}\le C$. When $l=1$, by \eqref{UAinduce} we have
\[
U^1A^1=A^0e^{B^1}+\wt P^0e^{B^1}+e^{B^1}\text{diag}(P^0),
\]
where $\|\wt P^0\|_{\lzd,s_1}\le C\e_0^\frac23$ and $\|\text{diag}(P^0)\|_{\lzd,s_1}\le\|P^0\|_{\lzd,s_1}\le\e_0$.
Combining  (iv) of Lemma \ref{algebra} with  (ii) of Lemma \ref{algebrapara}, one has
$
\|e^{B^{1}}\|_{\lzt,s_1}\le1+\|e^{B^1}-\id\|_{\beta,s_0-\sigma_1}^+\le2.
$
Similarly, we have $\|e^{B^{1}}\|_{\ldt,s_1}\le2$. Hence,
\begin{align*}
&\|U^1A^1\|_{\lzi,s_1}\le\|A^0e^{B^1}\|_{\lzi,s_1}+\|\wt P^0e^{B^1}\|_{\lzi,s_1}+\|e^{B^1}\text{diag}(P^0)\|_{\lzi,s_1}\\
\le&~\|A^0\|_{\lzi,s_1}\|e^{B^1}\|_{\lzt,s_1}+\|\wt P^0e^{B^1}\|_{\lzd,s_1}+\|e^{B^1}\text{diag}(P^0)\|_{\lzd,s_1}\\
\le&~ C(1+\e_0^\frac12)+\e_0^\frac12.
\end{align*}
By induction we finish the proof.
\end{proof}
\noindent From  \eqref{UAinduce} and  Lemma \ref{UA2},  we have
\begin{align*}
&\|U^{l+1}A^{l+1}-U^lA^l\|_{\lzi,\frac s2}\le\|U^lA^l\|_{\lzi,s_l}\|e^{B^{l+1}}-\id\|_{\lzt,\frac s2}\\
+&\|U^l\|_{\ldt,\frac s2}\|\wt P^l\|_{\lzd,s_{l+1}}\|e^{B^{l+1}}\|_{\lzt,\frac s2}+\|U^{l+1}\|_{\ldt,\frac s2}\|\text{diag}(P^l)\|_{\lzd,s_{l+1}}\\
\le&C\e_l^\frac23+4C\e_l^\frac23+2\e_l\le6C\e_l^\frac23.
\end{align*}
It follows $\{U^lA^l\}$ is a Cauchy sequence in $\lzi(\Pi_*,\frac s2)$ and thus define
$
U^\infty A^\infty:= \lim\limits_{l\rightarrow \infty} U^{l}A^{l}$
in the norm of $|\cdot |_{\lzi,\frac s2}$ below.
It is clear that
$$\|U^\infty A^\infty-U^lA^l\|_{\lzi,\frac s2}
\le\sum_{t=l}^\infty\|U^{t+1}A^{l+1}-U^tA^t\|_{\lzi,\frac s2}
\le C\e_l^\frac23.
$$
Thus we have the following.
\begin{lemma}\label{UlAlconverge}
When $l\rightarrow \infty$, $U^lA^l(\phi)\to U^\infty A^\infty(\phi)$ in $\lzi(\Pi_*,\frac s2)$ and
$U^\infty A^\infty$ satisfies
$
\|U^\infty A^\infty-A^0\|_{\lzi,\frac s2}\le C\e_0^\frac23
$
for some positive constant $C$.
\end{lemma}

From the construction we can prove the reducibility identity
\begin{eqnarray}\label{yuehuahengdeng}
U^l(\ot)(A^l(\ot)+P^l(\ot))=-\rmi\frac d{dt}U^l(\ot)+(A^0+P^0)U^l(\ot)\quad\text{in}~\lzi(\Pi_*).
\end{eqnarray}
From the above lemmas and let $l\to \infty$ in (\ref{yuehuahengdeng}),  one has
\begin{equation}\label{transformid}
 U^\infty(\ot)A^\infty(\ot)=-\rmi\frac d{dt}U^\infty(\ot)+(A^0+P^0(\ot))U^\infty(\ot), \qquad \omega\in\Pi_*,
\end{equation}
where the  identity holds in $\lzi$.\\
Proof of Theorem \ref{Theorem2.10}.  The measure estimate for $\Pi_{*}$ is similar as \cite{BG}. The estimate (\ref{gujiforUinfty}) is clear from  Lemma \ref{Ulconverge}
and the proof of Lemma \ref{Alconverge}.  We only need to prove the equivalence of two relative equations. \\
\indent If $t~\mapsto~y(t)\in\cal C^0(\R,\l_0^2)\cap\cal C^1(\R,\l_{-2\iota}^2)$ satisfies the equation \eqref{reduedeqn}, define $x(t)=U^\infty(\ot)y(t)$. By a straightforward computation, we have
\begin{align*}
\rmi \dot x&=\rmi\Big(\frac d{dt}U^\infty(\ot)\Big)y(t)+U^\infty(\ot)\rmi\dot y\\
\underline{\text{by \eqref{reduedeqn}}}\quad&=
\rmi\Big(\frac d{dt}U^\infty(\ot)\Big)y(t)+U^\infty(\ot)A^\infty(\ot)y(t)\\
\underline{\text{by \eqref{transformid}}}\quad&=(A^0+P^0(\ot))U^\infty(\ot)y(t)\\
\quad&=(A^0+P^0(\ot))x(t).
\end{align*}
From $y(t)\in\cal C^0(\R,\l_0^2)\cap\cal C^1(\R,\l_{-2\iota}^2)$, we can draw $x(t)\in\cal C^0(\R,\l_0^2)\cap\cal C^1(\R,\l_{-2\iota}^2)$.
On the contrary, if $\cal C^0(\R,\l_0^2)\cap\cal C^1(\R,\l_{-2\iota}^2)\ni x(t)$ satisfies \eqref{redueqn}, we define $y(t)=V^\infty(\ot)x(t)$. By Lemma \ref{Vlconverge} one has  $y(t)\in\cal C^0(\R,\l_0^2)\bigcap\cal C^1(\R,\l_{-2\iota}^2)$.
Since $V^\infty(\ot)U^\infty(\ot)=\id$, it follows
\begin{equation}\label{diffid}
\big(\frac d{dt}V^\infty(\ot)\big)U^\infty(\ot)=-V^\infty(\ot)\frac d{dt}U^\infty(\ot),
\end{equation}
and
\begin{align*}
\rmi\dot{y}(t)&=\rmi\Big(\frac d{dt}V^\infty(\ot)\Big)x(t)+V^\infty(\ot)\rmi \dot{x}(t)\\
\underline{\text{by \eqref{redueqn}}}\quad
&=\rmi\Big(\frac d{dt}V^\infty(\ot)\Big)U^\infty(\ot)y(t)+V^\infty(\ot)(A^0+P^0(\ot))x(t)\\
\underline{\text{by \eqref{diffid}}}\quad
&=-\rmi V^\infty(\ot)\Big(\frac d{dt}U^\infty(\ot)\Big)y(t)+V^\infty(\ot)(A^0+P^0(\ot))U^\infty(\ot)y(t)\\
\underline{\text{by \eqref{transformid}}}\quad
&=V^\infty(\ot)U^\infty(\ot)A^\infty(\ot)y(t)\\
&=A^\infty(\ot)y(t).
\end{align*}
\qed

\begin{remark}\label{jieshao2.29}
Lemma \ref{UlAlconverge} and \ref{VlA0converge} prove that $U^\infty(\ot)A^\infty(\ot)y$ and $V^\infty(\ot)A^0 x$ belong to $\l_{-2\iota}^2$ if $x,y\in\l_0^2$.
\end{remark}
The proof of Corollary \ref{tuilun2.12}, see \cite{BG}.

\subsection{Proof of Main Theorems}\hfill\par\noindent
We first prove Theorem \ref{mainthm1} based on Theorem \ref{Theorem2.10}.  All the assumptions B1- B3 should be checked.  Define $H_0=-\frac{d^2}{dx^2}+\cal V(x)$. $H_0$ is self-adjoint in $L^2(\R)$ and $spec(H_0)$ is discrete, and all eigenvalues
$
0<\lambda_1<\lambda_2<\lambda_3<\cdots\quad\text{are simple},
$
and $\lambda_j\sim cj^{\frac{2\l}{\l+1}}$ when $j\to \infty$ and all eigenfunctions $\{h_j(x)\}_{j\ge1}$ form a complete basis in $ L^2$.
As \cite{BG}, the equation \eqref{introduction01} can be written as (\ref{redueqn}), where
$A=\text{diag}\{\lambda_1,\lambda_2,\cdots\}$
and $P(\phi)=(P_i^j(\phi))_{i,j\ge1}$ with
$
P_i^j(\phi)=\int_\R\la x\ra^\mu W(\nu x,\phi)h_i(x)h_j(x)dx$  and $P_i^j(\phi)=P_j^i(\phi)$ for $\phi\in\T^n$.
As \cite{BG} and \cite{BamII}, the assumption B1 is satisfied. In the following we will show the assumptions B2 - B3 are fulfilled for the equation (\ref{redueqn}).
\begin{lemma}\label{yanzhengB2}
The map $P(\phi)$  is analytic from $\T_{s}^n$ into $\M_\beta$,  where $s=\rho-\delta_0$ with $s>0$ and
\[
\beta=
\begin{cases}
\frac{\mu}{2(\l+1)}-\frac1{2(\l+1)}\left(\frac13\bigwedge\frac{\mu+1}{2\mu+2\l+1}\right),&  \frac13\bigwedge\frac{\sqrt{\l^2+2}-\l}2\le\mu<\left(\l-\frac23\right)\bigwedge\frac{\sqrt{4\l^2-2\l+1}-1}2,\\
0 ,&\mu<\frac13\bigwedge\frac{\sqrt{\l^2+2}-\l}2.
\end{cases}
\]
\end{lemma}

\begin{proof}
We discuss the case  when $\frac13\bigwedge\frac{\sqrt{\l^2+2}-\l}2\le\mu<(\l-\frac23)\bigwedge\frac{\sqrt{4\l^2-2\l+1}-1}2$.
In this case, if $|\Im\phi|<\rho-\delta_0$, then
 \begin{align*}
|P_i^j(\phi)|&=\Big|\int_\R\la x\ra^\mu\sum_{k\in\Z^ d,l\in\Z^n}\widehat W(k,l)e^{\rmi k\cdot\nu x}e^{\rmi l\phi}h_i(x)h_j(x)dx\Big|\\
&\le\sum_{l\in\Z^n}e^{|l|(\rho-\delta_0)}\sum_{k\in\Z^ d}|\widehat W(k,l)|\Big|\int_\R e^{\rmi k\cdot\nu x}\la x\ra^\mu h_i(x)h_j(x)dx\Big|\\
\underline{ W(-\varphi,\phi)=-W(\varphi,\phi)}\quad&\le \sum_{l\in\Z^n}e^{|l|(\rho-\delta_0)}\sum_{0\ne k\in\Z^ d}|\widehat W(k,l)|\Big|\int_\R e^{\rmi k\cdot\nu x}\la x\ra^\mu h_i(x)h_j(x)dx\Big|\\
\underline{ \text{Lemma} ~ \ref{mainlem}}\quad&\le C\sum_{l\in\Z^n}e^{|l|(\rho-\delta_0)}\sum_{k\ne0}|\widehat W(k,l)|(|k\cdot\nu|\vee|k\cdot\nu|^{-1})(ij)^\beta\\
\underline{ A2}\quad&\le\frac C{\bar\gamma}\sum_{l\in\Z^n}e^{|l|(\rho-\delta_0)}\sum_{k\ne0}|\widehat W(k,l)||k|^{1\vee\tau_1}(ij)^\beta\\
\underline{A3}\quad&\le\frac C{\bar\gamma}\sum_{l\in\Z^n}e^{|l|(\rho-\delta_0)}\sum_{k\ne0}\frac{e^{-|l|\rho}|k|^{1\vee\tau_1}(ij)^\beta}{\la k_1\ra^{[1\vee\tau_1]+ d+2}\cdots\la k_d\ra^{[1\vee\tau_1]+ d+2}}\\
\underline{\exists~i_0\in\{1,\cdots,d\},~|k_{i_0}|\ge\frac{|k|}{d}}\quad&\le\frac{C}{\bar\gamma}(ij)^\beta.
\end{align*}
It follows  $P(\phi)$ is an analytic map from $\T_{s}^n$ into $\M_\beta$ with $0\leq 2\beta<\frac{\l-1}{\l+1}$. The rest is similar.
\end{proof}

\begin{remark}
From Assumption {\rm A1} we can show that when $|x|\geq R_0>0$ large enough, Assumption 1.1 in Lemma \ref{mainlem} is satisfied for the potential $\cal V(x)$ and thus, we can apply it  in the above proof.
From Lemma \ref{yanzhengB2} we prove that  B2 is satisfied for the equation (\ref{redueqn}).
\end{remark}
In the following we will show that B3 is satisfied. Following \cite{BamII, BamIII}, we have
\begin{lemma}\label{CVT}
Let $g\in S^{m_1,m_2}$, then one has $$g^{w}(x,-\rmi\px)\in\B(\H^{s+s_1},\H^s),\quad\forall~s\in\R,~\forall~s_1\ge m_1+[m_2]\quad
\text{with}\quad[m_2]:=m_2\vee0.$$
\end{lemma}

\begin{Definition}
An operator $G$ will be said to be pseudodifferential of class $OPS^{m_1,m_2}$ if there
exists a symbol $g\in S^{m_1,m_2}$ such that $G = g^w(x,-\rmi\px)$.
\end{Definition}

\begin{Definition}
An operator $F$ will be said to be a pseudodifferential operator of class $\cal O^{m_1,m_2}$
if there exists a sequence $f\in S^{m_1^{(j)}, m_2^{(j)}}$ with $m_1^{(j)}+[m_2^{(j)}]\leq m_1^{(j-1)}+[m_2^{(j-1)}] $ and,  for any $\kappa$
there exist $N$ and an operator $R_N\in \cal{B}(\H^{s-\kappa}, \H^s)$ for any $s$ such that
$F=\sum\limits_{j=1}^N f_j^{w}+R_N$.
\end{Definition}

\begin{lemma}\label{multiplythm}
If  $a(x,\xi)\in S^{m_1,m_2}$ and $b(x,\xi)\in S^{m_1',m_2'}$, then  $a\cdot b\in S^{m_1+m_1',m_2+m_2'}.$
\end{lemma}
Given a symbol $g\in S^{m_1,m_2}$ we will write
\[
g\sim\sum_{j\ge0}g_j,\quad g_j\in S^{m_1^{(j)},m_2^{(j)}},\quad m_1^{(j+1)}+[m_2^{(j+1)}]\le m_1^{(j)}+[m_2^{(j)}],
\]
if $\forall~\kappa$ there exist $N$ and $r_N\in S^{-\kappa,0}$ such that $g=\sum_{j=0}^Ng_j+r_{N}$.
The following lemma is from \cite{BamI}.
\begin{lemma}\label{moyalproduct}
Given a couple of symbols $a\in S^{m_1,m_2}$  and $b\in S^{m_1',m_2'}$, then there exists a symbol $c$, denoted by $c = a\#b$
such that
\[
(a\#b)^w(x,D_x)=a^w(x,D_x)b^w(x,D_x),
\]
furthermore one has
\(
a\#b\sim\sum_{j\ge0}c_j
\)
with
\[
c_j=\sum_{k_1+k_2=j}\frac1{k_1!k_2!}(\frac12)^{k_1}(-\frac12)^{k_2}(\partial_\xi^{k_1}D_x^{k_2}a)(\partial_\xi^{k_2}D_x^{k_1}b)
\in S^{m_1+m_1'-j\l,m_2+m_2'-j}
\]
where $D_x=-\rmi\px$.
\end{lemma}
From Lemma \ref{moyalproduct}, we have
\begin{Corollary}\label{opsmproduct}
If   $A\in OPS^{m_1,m_2},~B\in OPS^{m_1',m_2'}$, then $AB\in OPS^{m_1+m_1',m_2+m_2'}.$
\end{Corollary}

\begin{lemma}\label{l0ldeltanorm}
 If $0\leq \mu\leq \delta(\l+1)$, then
$(\lambda^{w}(x,D_x))^{-\delta(\ell+1)} \la x\ra^{\mu} \in \cal B(\cal{H}^0)$.
\end{lemma}
\begin{proof}
It is clear that
\begin{eqnarray}\label{zhongyaoguancha}
(\lambda^{w}(x,D_x))^{-\delta(\ell+1)} \in \cal O^{-\delta(\ell+1),0}.
\end{eqnarray} From the definition, there exists $N$ and an operator
$R_N\in \cal{B}(\cal H^{-\mu}, \cal H^0)$ and $(\lambda^{w}(x,D_x))^{-\delta(\ell+1)} =\sum\limits_{j\geq 0}^{N}f_j^{w}+R_N$.  From (\ref{zhongyaoguancha}) we have
$\sum\limits_{j\geq 0}^{N}f_j^{w}\in OPS^{-\delta(\l+1), 0}$. On the other hand, $\la x\ra^{\mu}\in OPS^{\mu,0}$ since $\mu\geq 0$(see \cite{BamIII}).  From Lemma \ref{CVT}, it follows
$\la x\ra^{\mu}\in \cal{B}(\cal H^0, \cal H^{-\mu})$ and thus $(\sum\limits_{j\geq 0}^Nf_j^{w} ) \cdot \la x\ra^{\mu}\in OPS^{\mu-\delta(\l+1),0}\subset OPS^{0,0}$ by $\mu\leq \delta(\l+1)$. Therefore,
$(\sum\limits_{j\geq 0}^{N}f_j^{w}) \cdot  \la x\ra^{\mu}\in \B(\H^{0})$. For the second part it is easy to check that
$R_{N} \la x\ra^{\mu} \in \cal B(\H^{0})$.  Combining with the two parts we finish the proof.
\end{proof}

\begin{lemma}\label{xmuxingzhi}
If $0\leq \mu\leq \delta(\l+1)$, then the multiplication operator $\la x\ra^\mu\in\cal{B}(\H^0,\H^{-\delta})$.
\end{lemma}
\begin{proof}
From the self-adjointness, H\"older inequality and Lemma  \ref{l0ldeltanorm} one has
\begin{align*}
&\|\la x\ra^\mu\|_{\B(\H^0, \H^{-\delta})}=\sup_{\|f\|_{\H^0}=1}\|\la x\ra^\mu f\|_{\H^{-\delta}}
=\sup_{\substack{\|f\|_{\H^0}=1,\\ \|g\|_{\H^\delta}=1\phantom,}} \big|\la g,\la x\ra^\mu f\ra\big|\\
=&~\sup_{\substack{\|f\|_{\H^0}=1,\\ \|g\|_{\H^\delta}=1\phantom,}}
\big|\la[\lambda^w]^{\delta(\l+1)} g,[\lambda^w]^{-\delta(\l+1)}\la x\ra^\mu f\ra\big| \\
\le&~\sup_{\substack{\|f\|_{\H^0}=1,\\ \|g\|_{\H^\delta}=1\phantom,}}
\|[\lambda^w]^{\delta(\l+1)} g\|_{\H^0}\cdot\|[\lambda^w]^{-\delta(\l+1)}\la x\ra^\mu f\|_{\H^0} \\
\le&~\sup_{\|f\|_{\H^0}=1}\|[\lambda^w]^{-\delta(\l+1)}\la x\ra^\mu f\|_{\H^0}= \|[\lambda^w]^{-\delta(\l+1)}\la x\ra^\mu\|_{\B(\H^0)}<\infty.
\end{align*}
\end{proof}

\begin{lemma}\label{chengfasuanzi}
Suppose that $g(x,\phi)$ is continuous on $x\in\R$ and analytic on $\phi\in\T_s^n$ and there exists a positive constant $C>0$ such that
$
|g(x,\phi)|\le C $ on $(x,\phi)\in\R\times\overline{\T_s^n}$,
then if $0\leq \mu\leq \delta(\l+1)$, for any $\phi\in \T_s^n$, $\la x\ra^\mu g(x,\phi)$ is an analytic map from $\T_s^n$ to $ \B( \H^0,\H^{-\delta})$. On the other hand, if
$\mu<0$, then $\la x\ra^\mu g(x,\phi)$ is an analytic map from $\T_s^n$ to $ \B( \H^0)$.
\end{lemma}
\begin{proof}
From the boundedness of $g(x,\phi)$ on  $\phi\in \T_s^n$ and the definition, we can draw that the multiplication operator $g(x,\phi)\in \B(\H^0)$ on $\T_s^n$.   Together with Lemma \ref{xmuxingzhi} and $0\leq \mu\leq \delta(\l+1)$, one has
the multiplication operator $\la x\ra^{\mu}g(x,\phi) \in \B( \H^0,\H^{-\delta})$ for $\phi\in \T_s^n$.   The rest is clear.
\end{proof}

\noindent From Lemma \ref{chengfasuanzi}  it follows the map $\T^n\ni\phi~\mapsto~P(\phi)\in\lzd$ is analytic on $\T_s^n$  if  $0\leq \mu\leq \delta(\l+1)$.  As above we discuss two cases. When $\frac13\bigwedge\frac{\sqrt{\l^2+2}-\l}2\le\mu<(\l-\frac23)\bigwedge\frac{\sqrt{4\l^2-2\l+1}-1}2$,  one has $\beta=\frac\mu{2(\l+1)}-\frac1{2(\l+1)}(\frac13\bigwedge\frac{\mu+1}{2\mu+2\l+1})$. If choose $\delta=\frac{\l}{\l+1}$,  we have $\mu\leq \delta(\l+1)$,  $0\le\delta<\frac{2\l}{\l+1}-\beta-\frac12$ and  $2\beta\leq \delta$. When  $0\leq \mu<\frac13\bigwedge\frac{\sqrt{\l^2+2}-\l}2$,  one has $\beta=0$ and $\delta=\frac{\l}{\l+1}$.  If $\mu<0$, set $\beta=\delta=0$. This confirms the assumption B3. \\
\noindent Proof of Theorem \ref{mainthm1}.  As we mentioned above, the equation \eqref{introduction01} can be written as (\ref{redueqn}). Since all the assumptions B1 - B3 are checked,
we can use Theorem \ref{Theorem2.10} and Corollary \ref{tuilun2.12} to finish the proof. For details, see \cite{BG}.   \\
\noindent Proof of Theorem \ref{mainthm2}. It is similar. \\
\noindent Proof of Corollary \ref{maincorBam}. As above
\(\dss
P_i^j(\phi)=\int_\R\la x\ra^\mu g(x,\phi)h_i(x)h_j(x)dx.
\)
Clearly, for all $\phi\in\T_\rho^n$, there exists some positive constant $C$ such that $|\la x\ra^\mu g(x,\phi)|\le C(\mu)|x|^\mu$ is satisfied for $|x|\ge 1$ and $\mu\geq 0$. By  Lemma \ref{mainlem2} one has
\[
|P_i^j(\phi)|\le C\Big|\int_\R\la x\ra^\mu g(x,\phi) h_i(x)\overline{h_j(x)}dx\Big|\le C(ij)^{\frac\mu{2(\l+1)}}, \qquad 0\leq \mu<\l-1.
\]
Thus, we have $0\le2\beta = \frac{\mu}{\l+1}<\frac{\l-1}{\l+1}$.  By Lemma \ref{chengfasuanzi}, B3 is satisfied if we choose $\delta=\frac{\l}{\l+1}$. The following is similar as above.\\
\indent If $\mu<0$,  we set $\beta=0$ and $\delta=0$. The rest is similar. \qed

\section{Estimates on eigenfunctions}\label{section3}
In this section we will prove Lemma \ref{mainlem} and \ref{mainlem2} based on Langer's turning point method and oscillatory integrals. For the proof the rough idea is that we first rewrite the eigenfunction into
the sum of two different functions, and then use Lemma \ref{oscint} to estimate the relative integrals, if necessary.
\subsection{Langer's turning point and the new form of the eigenfunctions}
\indent Consider the function
\begin{eqnarray}\label{tezhengfangcheng}
h_n^{\prime\prime}(x)+(\lambda_n-V(x))h_n(x)=0,\quad x\geq 0,
\end{eqnarray}
where $V(x)$ satisfies Assumption \ref{poteass}.   From Lemma \ref{potesim} there exists a positive constant $R\ge 2\widetilde{R}\geq 2R_0$ such that the following conditions are satisfied:
\begin{flalign}
(\rm i).\phantom{\rm ii}&~V(x)\le xV'(x),\quad\text{for}~x\in[{\txs\frac R2},\infty),&\label{potecd1}\\
(\rm ii).\phantom{\rm i}&~|V(x)|<V(R),\quad\text{for}~x\in[0,R).&\label{potecd2}
\end{flalign}
Let $n_0 : = \min\big\{n\in\Z_+|\lambda_n\ge V(R)\big\}$ and $\lambda_n=V(X_n)$ for $n\geq n_0$. From the above, $X_n$ is unique when $n\geq n_0$ as the figure \ref{potefig} below.
\begin{figure}[H]
\centering
\epsfig{file=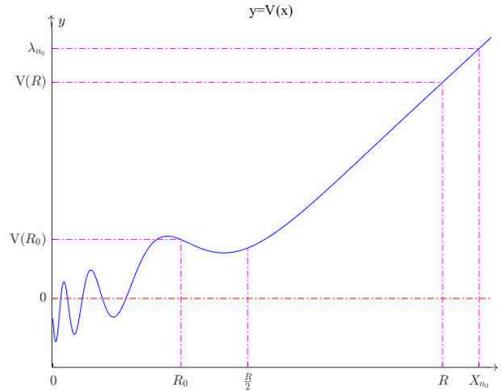,width=0.45\textwidth}
\caption{Potential Function}\label{potefig}
\end{figure}
\begin{lemma}\label{eigenfunction}
For $n\ge n_0$, if $x\ge0$, we have
\begin{equation}\label{eigenfunexp}
h_n(x)=\psi_1^{(n)}(x)+\psi_2^{(n)}(x)\quad\text{with}\quad|\psi_2^{(n)}(x)|\le\frac{C}{X_n^{\l+1}}|\psi_1^{(n)}(x)|,
\end{equation}
where
\begin{align}
\psi^{(n)}_1(x)&:=C_n(\lambda_n-V(x))^{-\frac{1}{4}}\Big(\frac{\pi\zeta_n}{2}\Big)^{\frac{1}{2}}H_{\frac{1}{3}}^{(1)}(\zeta_n)\label{eigenfunmain}
\end{align}
and
\(\dss
C_n\sim X_n^{\frac{\l-1}2},~\zeta_n(x)=\int_{X_n}^x(\lambda_n-V(t))^\frac12dt~\text{with}~\arg\zeta_n(x)=
\begin{cases}
\frac\pi2,&x>X_n,\\
-\pi,&x<X_n.
\end{cases}
\)
\end{lemma}
For the proof see section \ref{section4}.
\begin{lemma}\label{eigens}
For $1\le n<n_0$, if $x\ge2R$, we have
\begin{equation}\label{eigenexps}
h_n(x)=\psi_1^{(n)}(x)+\psi_2^{(n)}(x)\quad\text{with}\quad
|\psi_2^{(n)}(x)|\le\frac{C}{xV^\frac12(x)}|\psi_1^{(n)}(x)|,
\end{equation}
where
$
\zeta_n(x)=\int_{X_n}^x(\lambda_n-V(t))^\frac12dt$
with $\arg\zeta_n(x)=\frac\pi2$ and $X_n=R$.
\end{lemma}
\subsection{Proof of Lemma \ref{mainlem}: Part 1}\label{introduction1}
\subsubsection{the integral on $[0, +\infty)$ in the three  cases}In the following lengthy proof  we first estimate the integral on $[0,+\infty)$ under Assumption \ref{pertass}. The rest integral estimation on $(-\infty,0)$ can be obtained by a coordinate transformation. In this way we prove the estimate \eqref{mainesto} of Lemma \ref{mainlem}. \\
\begin{Corollary}\label{eigenbound}
Given $n\in\Z_+$, $h_n(x)$ is bounded on $[0,\infty)$.
\end{Corollary}
\begin{proof}
If $n\ge n_0$, by \eqref{eigenfunmain},~\eqref{eigenfunexp} and Lemma \ref{Bessel} we obtain $|h_n(x)|\le C$ for $x\ge0$.
If $1\le n<n_0$ and $x\in [0,~2R]$ one has $|h_n(x)|\le C$. For $x\ge2R$ one also has $|h_n(x)|\le C$ by  \eqref{eigenfunmain}, \eqref{eigenexps} and Lemma \ref{Bessel}.
\end{proof}
As in \cite{YajimaZhang}, we have
\begin{lemma}\label{qQesti}
Assume $V(x)$ satisfies Assumption \ref{poteass}, then there exists constants $a_1,a_2,A_1,A_2$ such that the following estimates
are satisfied uniformly for $n\ge n_0$:
\begin{equation}\label{qesti}
\begin{alignedat}{5}
a_1X_n^{2\l-1}(X_n-x)\le\lambda_n-V(x)&\le a_2X_n^{2\l-1}(X_n-x),&\qquad&\text{for~}0\le x<X_n,\\
V(x)-\lambda_n&\ge a_1X_n^{2\l-1}(x-X_n),&&\text{for~}x\ge X_n
\end{alignedat}
\end{equation}
and
\begin{equation}\label{Qesti}
\begin{alignedat}{5}
A_1X_n^{\l-\frac12}(X_n-x)^{\frac32}\le-\zeta_n(x)&\le A_2X_n^{\l-\frac12}(X_n-x)^{\frac32},&\qquad&\text{for~}0\le x<X_n,\\
-\rmi \zeta_n(x)&\ge A_1X_n^{\l-\frac12} (x-X_n)^{\frac32},&&\text{for~}x\ge X_n,
\end{alignedat}
\end{equation}
where $0<a_1\le1\le a_2$ and $0<A_1\le1\le A_2$.
\end{lemma}
\noindent For the proof see section \ref{section4}.\\
For the integral on $[0,+\infty)$, we have the following lemma.
\begin{lemma}\label{realplus}
If $f(x)$ satisfies Assumption \ref{pertass}, then
\[
\Big|\int_{0}^{+\infty} f(x) e^{{\rm i}kx}h_m(x)\overline{h_n(x)}dx\Big|
\le C(|k|^{-1}\vee|k|)(\lambda_m\lambda_n)^{\frac\mu{4\l}-\frac1{4\l}(\frac13\bigwedge\frac{\mu+1}{2\mu+2\l+1})},\quad \forall~k\ne0,
\]
where $C$ only depends on $(\mu,\l)$.
\end{lemma}
Define
$\dss n_1:=\min\big\{n>n_0:X_n^{\frac13}\ge 2X_{n_0}\big\}$.
Assume $m\le n$ in the following. We now prepare to prove Lemma \ref{realplus} in three different cases, which are
$m,n<n_1$, $m<n_0$ and $n\ge n_1$ and $m,n\ge n_0$.
For  the first case we have
\begin{lemma}\label{mnsmall} If $f(x)$ satisfies Assumption \ref{pertass}, then one has
\[
\Big|\int_{0}^{+\infty}f(x) e^{{\rm i}kx}h_m(x)\overline{h_n(x)}dx\Big|
\le C(X_mX_n)^{\frac\mu2-\frac12},\quad \forall~k\in\R,
\]
where $C$ only depends on $(\mu,\l)$ and $1\le m\le n<n_1$.
\end{lemma}
\begin{proof}
By Lemma \ref{eigenfunction} and \ref{eigens}, for $x\ge 2X_{n_1}$ there exists a positive constant $C$ such that
\[
h_n(x)=\psi_1^{(n)}(x)+\psi_2^{(n)}(x)\quad\text{with}\quad|\psi_2^{(n)}(x)|\le C|\psi_1^{(n)}(x)|.
\]
When $x\ge 2X_{n_1}$,
 there exists a $C_0>0$ such that
$V(x)-\lambda_n\ge C_0$ and $
|\zeta_n(x)|\ge C_0(x-X_n)$,
then
\begin{align*}
&\Big|\int_{2X_{n_1}}^{+\infty}f(x) e^{{\rm i}kx}h_m(x)\overline{h_n(x)}dx\Big|\le C\int_{2X_{n_1}}^\infty x^\mu(V(x)-\lambda_n)^{-\frac12}e^{-(|\zeta_m(x)|+|\zeta_n(x)|)}dx\\
\le&~C\int_{2X_{n_1}}^\infty(x-X_n)^\mu e^{-C_0(x-X_n)}dx\le C(X_mX_n)^{\frac\mu2-\frac12}.
\end{align*}
Besides, by H\"older inequality we have
\(\dss
\Big|\int_{0}^{2X_{n_1}}f(x) e^{{\rm i}kx}h_m(x)\overline{h_n(x)}dx\Big|\le C(X_mX_n)^{\frac\mu2-\frac12}.
\)
\end{proof}

For the second case we have
\begin{lemma}\label{msmall} If $f(x)$ satisfies Assumption \ref{pertass}, then
\[
\Big|\int_{0}^{+\infty}f(x) e^{{\rm i}kx}h_m(x)\overline{h_n(x)}dx\Big|
\le C(X_mX_n)^{\frac\mu2-\frac16},\quad \forall~k\in\R,
\]
where $C$ only depends on $(\mu,\l)$ and $1\le m<n_0,~n\ge n_1$.
\end{lemma}
\begin{proof}
By Corollary \ref{eigenbound} one has $h_m(x)$ are bounded in $[0,\infty)$ uniformly for $m<n_0$, then
\begin{align*}
&\Big|\int_{0}^{X_n^\frac13}f(x) e^{{\rm i}kx}h_m(x)\overline{h_n(x)}dx\Big|\le C X_n^{\frac{\l-1}2+\frac\mu3}\int_0^{X_n^\frac13}(\lambda_n-V(x))^{-\frac14}dx\le C(X_mX_n)^{\frac\mu2-\frac16}.
\end{align*}
Note that $X_n^\frac13\ge2X_{n_0}$ from $n\ge n_1$. For any $m\le n_0$,
there exists a $C_0>0$ such that
\[
V(x)-\lambda_m\ge C_0,\quad
|\zeta_m(x)|\ge C_0(x-X_m)\ge\frac{C_0}2x,\quad\forall~x\ge X_n^\frac13.
\]
By H\"older inequality we obtain
\begin{align*}
&\Big|\int_{X_n^\frac13}^{+\infty}f(x) e^{{\rm i}kx}h_m(x)\overline{h_n(x)}dx\Big|\le C\Big(\int_{X_n^\frac13}^\infty x^{2\mu} h_m^2(x)dx\Big)^\frac12\\
\le&~ C\Big(\int_{X_n^\frac13}^\infty (x-X_m)^{2\mu}e^{-2C_0(x-X_m)}dx\Big)^\frac12\le Ce^{-\frac{C_0}4X_n^\frac13}
\le C(X_mX_n)^{\frac\mu2-\frac16}.
\end{align*}
\end{proof}
For the third case, we have
\begin{lemma}\label{mlarge} Assume $f(x)$ satisfies Assumption \ref{pertass}, then one has
\[
\Big|\int_{0}^{+\infty}f(x) e^{{\rm i}kx}h_m(x)\overline{h_n(x)}dx\Big|
\le C (|k|^{-1}\vee|k|)(X_mX_n)^{\frac\mu2-\frac12(\frac13\bigwedge\frac{\mu+1}{2\mu+2\l+1})},\quad\forall~k\ne0,
\]
where $C$ only depends on $(\mu,\l)$ and $n_0\le m\le n$.
\end{lemma}
To prove above Lemma \ref{mlarge}, we split the integral into two parts which is delayed in the following.
\subsubsection{the Integral on $[X_n, +\infty)$}\label{intxninf}
For the following part in this section, we will denote
$\cal{F}(x) : = f(x)e^{{\rm i}kx}\psi_1^{(m)}(x)\overline{\psi_1^{(n)}(x)}$ for simplicity.
Our main result in this part  is the following.
\begin{lemma}\label{xninf}
If  $f(x)$ satisfies Assumption \ref{pertass}, then
\[
\Big|\int_{X_n}^{+\infty}f(x) e^{{\rm i}kx}h_m(x)\overline{h_n(x)}dx\Big|
\le C(X_mX_n)^{\frac\mu2-\frac13},\quad \forall~k\in\R,
\]
where $C$ only depends on $(\mu,\l)$ and $n_0\le m\le n$.
\end{lemma}
Lemma \ref{xninf} is the direct corollary of the following two lemmas.
\begin{lemma}\label{cxninf}
If  $f(x)$ satisfies Assumption \ref{pertass}, then
\[
\Big|\int_{2X_n}^{+\infty}f(x)e^{{\rm i}kx}h_m(x) \overline{h_n(x)}dx\Big|\le Ce^{-C_0X_n^{\l+1}},
\quad\forall~k\in\R,
\]
where $C_0,C$ only depend on $(\mu,\l)$ and $n_0\le m\le n$.
\end{lemma}
\begin{proof}
Recall that
$\psi_1^{(n)}(x)=C_n(\lambda_n - V(x))^{-\frac{1}{4}}\big(\frac{\pi \zeta_n}{2}\big)^\frac12H^{(1)}_\frac{1}{3}(\zeta_n)$ with $C_n\sim X_n^{\frac{\l-1}2}$.\\
From Lemma \ref{qQesti}, for any $n\ge n_0$, if $x\ge2X_n$, we have
$
V(x)-\lambda_n\ge a_1X_n^{2\l}$ and
$|\zeta_n(x)|\ge A_1X_n^{\l+1}$
and thus
\begin{align*}
&\Big|\int_{2X_n}^{+\infty}\cal F(x) dx\Big|\le CX_m^{\frac{\l-1}{2}}X_n^{\frac{\l-1}{2}}\int_{2X_n}^{+\infty}x^\mu(V(x)-\lambda_m)^{-\frac{1}{4}}(V(x)-\lambda_n)^{-\frac{1}{4}} e^{-|\zeta_m(x)|}e^{-|\zeta_n(x)|}dx \\
\le~&CX_n^{\l-1}\int_{2X_n}^{+\infty}x^\mu(V(x)-\lambda_n)^{-\frac{1}{2}} e^{-|\zeta_n(x)|}dx\\
\le~&Ce^{-\frac{A_1}2X_n^{\l+1}}X_n^{-1} \int_{2X_n}^{+\infty}x^\mu e^{-\frac{A_1}2X_n^\l(x-X_n)}dx \\
\le~& C e^{-\frac{A_1}2X_n^{\l+1}}\int_{2X_n}^{+\infty}(x-X_n)^{\mu} e^{-C_0(x-X_n)}dx\le Ce^{-C_0X_n^{\l+1}},
\end{align*}
where $C_0 >0$ only depends on $(\mu,\l)$. Note that $h_n(x)=\psi^{(n)}_1(x)+O(X_n^{-(\l+1)})\psi^{(n)}_1(x)$, then we obtain
\[
\Big|\int_{2X_n}^{+\infty}f(x)e^{{\rm i}kx}h_m(x) \overline{h_n(x)}dx\Big|\le Ce^{-C_0X_n^{\l+1}}.
\]
\end{proof}
\begin{lemma}\label{xncxn} If  $f(x)$ satisfies Assumption \ref{pertass}, then
\[
\Big|\int^{2X_n}_{X_n}f(x) e^{{\rm i}kx}h_m(x) \overline{h_n(x)}dx\Big|
\le C(X_mX_n)^{\frac\mu2-\frac13},\quad\forall~k\in\R,
\]
where $C$ only depends on $\mu, \l$ and $n_0\le m\le n$.
\end{lemma}
\begin{proof} Since $X_n\ge1$ from $n\ge n_0$, then $X_n+X_n^{-\frac13}\le2X_n$.
We split the integral into two parts as:
\[\int^{2X_n}_{X_n}f(x) e^{{\rm i}kx}h_m(x) \overline{h_n(x)}dx=
\Big(\int^{X_n+X_n^{-\frac{1}{3}}}_{X_n}+\int^{2X_n}_{X_n+X_n^{-\frac{1}{3}}}\Big)dx.
\]
From Lemma \ref{qQesti}, if $x\ge X_n+X_n^{-\frac{1}{3}}$, we have
\[
V(x)-\lambda_n\ge a_1X_n^{2\l-1}(x-X_n),\quad
|\zeta_n(x)|\ge A_1X_n^{\l-\frac{1}{2}}(x-X_n)^\frac{3}{2}\ge A_1X_n^{\l-1}.
\]
Then
\begin{align*}
&\Big|\int_{X_n+X_n^{-\frac{1}{3}}}^{2X_n}\cal F(x) dx\Big|\le CX_m^{\frac{\l-1}{2}}X_n^{\frac{\l-1}{2}}\int^{2X_n}_{X_n+X_n^{-\frac{1}{3}}}x^\mu(V(x)-\lambda_m)^{-\frac{1}{4}}(V(x)-\lambda_n)^{-\frac{1}{4}} e^{-|\zeta_n(x)|}dx \\
\le~&CX_n^{\mu+\l-1}\int^{2X_n}_{X_n+X_n^{-\frac{1}{3}}}(V(x)-\lambda_n)^{-\frac{1}{2}}e^{-|\zeta_n(x)|}dx \\
\le~& Ce^{-A_1X_n^{\l-1}}X_n^{\mu-\frac12}\int^{2X_n}_{X_n+X_n^{-\frac{1}{3}}}(x-X_n)^{-\frac{1}{2}}dx\\
\le~&Ce^{-A_1X_n^{\l-1}}X_n^{\mu}\le  C e^{-C_0 X_n^{\l-1}},
\end{align*}
where $C_0$  only depends on $(\mu,\l)$. Then we estimate the remainder integral under two cases. If $X_m\le X_n< 2X_m$, we have
\begin{align*}
&\Big|\int^{X_n+X_n^{-\frac{1}{3}}}_{X_n}\cal F(x) dx\Big|\le CX_m^{\frac{\l-1}{2}}X_n^{\mu+\frac{\l-1}{2}}\int^{X_n+X_n^{-\frac{1}{3}}}_{X_n}(V(x)-\lambda_m)^{-\frac{1}{4}}(V(x)-\lambda_n)^{-\frac{1}{4}}dx \\
\le~& CX_n^{\mu+\l-1}\int^{X_n+X_n^{-\frac{1}{3}}}_{X_n} (V(x)-\lambda_n)^{-\frac{1}{2}}dx \\
\le~& C  X_n^{\mu-\frac23}\le C (X_mX_n)^{\frac\mu2-\frac13}.
\end{align*}
Otherwise, $X_n\ge2X_m$. From the Lemma \ref{qQesti}, for any $m\ge n_0$, if $x\ge X_n$, we have
\[
V(x)-\lambda_m\ge a_1X_m^{2\l-1}(x-X_m)\ge a_1X_m^{2\l},\quad
|\zeta_m(x)|\ge A_1X_m^{\l-\frac12}(x-X_m)^{\frac32}\ge A_1X_m^\l(x-X_m).
\]
It follows
\begin{align*}
&\Big|\int^{X_n+X_n^{-\frac{1}{3}}}_{X_n} \cal F(x) dx\Big|\le C X_m^{\frac{\l-1}{2}}X_n^{\frac{\l-1}{2}}\int_{X_n}^{X_n+X_n^{-\frac{1}{3}}}(x-X_m)^\mu(V(x)-\lambda_m)^{-\frac{1}{4}}(V(x)-\lambda_n)^{-\frac{1}{4}}e^{-|\zeta_m|}dx \\
\le~&
C X_m^{-\frac{1}{2}}X_n^{-\frac14}\int_{X_n}^{X_n+X_n^{-\frac{1}{3}}}(x-X_m)^\mu(x-X_n)^{-\frac{1}{4}}e^{-A_1X_m^\l(x-X_m)}dx \\
\le~& CX_m^{-\frac12}X_n^{-\frac14}\int^{X_n+X_n^{-\frac{1}{3}}}_{X_n}(x-X_n)^{-\frac{1}{4}}dx
\le C(X_mX_n)^{\frac\mu2-\frac12}.
\end{align*}
Note $h_n(x)=\psi^{(n)}_1(x)+O(X_n^{-(\l+1)})\psi^{(n)}_1(x)$,  we obtain
\[
\Big|\int^{2X_n}_{X_n}f(x)e^{{\rm i}kx}h_m(x) \overline{h_n(x)}dx\Big|
\le C(X_mX_n)^{\frac\mu{2}-\frac{1}{3}}.
\]
\end{proof}
\subsection{Proof of Lemma \ref{mainlem}: Part 2}
\subsubsection{the Integral on $[0, X_n)$: preparations}\label{intzeroxn}
In this part we will prove the following.
\begin{lemma}\label{zeroxn}
If $f(x)$ satisfies Assumption \ref{pertass}, then
\[
\Big|\int_0^{X_n}f(x) e^{{\rm i}kx}h_m(x)\overline{h_n(x)}dx\Big|
\le C (|k|^{-1}\vee|k|)(X_mX_n)^{\frac\mu2-\frac12(\frac13\bigwedge\frac{\mu+1}{2\mu+2\l+1})},\quad\forall~k\ne0,
\]
where $C$ only depends on $(\mu,\l)$ and $n_0\le m\le n$.
\end{lemma}
For the simplicity we will introduce the following notations.  For  $m\ge n_0$ we  denote
$
f_m(x):=\int_0^{\infty}e^{-t}t^{-\frac{1}{6}}\Big(1+\frac{{\rm i}t}{2\zeta_m}\Big)^{-\frac16}dt.
$
Then for $x\in [0,X_{m})$, one has
\begin{align*}
\psi_1^{(m)}(x) &=C_m(\lambda_m-V(x))^{-\frac{1}{4}}\big(\frac{\pi\zeta_m}{2}\big)^\frac12 H_\frac{1}{3}^{(1)}(\zeta_m)\\
&=C_m (\lambda_m-V(x))^{-\frac{1}{4}} \frac{e^{{\rm i}(\zeta_m-\frac{5\pi}{12})}}{\Gamma{\big(\frac{5}{6}\big)}}
\int_0^{\infty}e^{-t}t^{-\frac{1}{6}}\Big(1+\frac{{\rm i}t}{2\zeta_m}\Big)^{-\frac{1}{6}}dt\\
&=C_m(\lambda_m-V(x))^{-\frac{1}{4}}e^{{\rm i}\zeta_m(x)}f_m(x),
\end{align*}
where $C_m\sim X_m^{\frac{\l-1}2}$. Similarly for $x\in [0,X_{m})$ we have
$
\overline{\psi_1^{(n)}(x)} =C_n(\lambda_n-V(x))^{-\frac{1}{4}}e^{-{\rm i}\zeta_n(x)}\overline{f_n(x)} $
with
$C_n\sim X_n^{\frac{\l-1}2}$.  For $x\in[0,X_m)$, denote
\begin{align*}
g(x)&:=(\zeta_n(x)-\zeta_m(x)-kx)'=(\lambda_n-V(x))^{\frac12}-(\lambda_m-V(x))^{\frac12}-k,\\
\Psi(x)&:=(\lambda_m-V(x))^{-\frac{1}{4}} (\lambda_n-V(x))^{-\frac{1}{4}}f_m(x)\overline{f_n(x)}.
\end{align*}
By a straightforward computation we  have
$
g'(x)=\frac{V'(x)}{2}\big((\lambda_m-V(x))^{-\frac12}-(\lambda_n-V(x))^{-\frac12}\big)$ and
\begin{align*}
\Psi'(x)&=\frac14V'(x)(\lambda_m-V(x))^{-\frac{5}{4}} (\lambda_n-V(x))^{-\frac{1}{4}} f_m(x)\overline{f_n(x)}\\
&+\frac14V'(x)(\lambda_m-V(x))^{-\frac{1}{4}} (\lambda_n-V(x))^{-\frac{5}{4}}f_m(x)\overline{f_n(x)}\\
&+(\lambda_m-V(x))^{-\frac{1}{4}} (\lambda_n-V(x))^{-\frac{1}{4}} \Big(f_m'(x)\overline{f_n(x)}+f_m(x)\overline{f_n'(x)}\Big).
\end{align*}
From $x\in [0, X_m)$ it is easy to obtain
 $|f_m(x)| \le \Gamma(\frac{5}{6})$ and  $|f_n(x)| \le \Gamma(\frac{5}{6})$. From Lemma \ref{qQesti} we have the following.
\begin{Corollary}\label{psiesti}
For $x\in [0,X_m)$ and $n_0\le  m\le  n$, one has
$$
|\Psi(x)|\le C(\lambda_m-V(x))^{-\frac12}\le C X_m^{-\l+\frac12}(X_m-x)^{-\frac12}$$
 and
$$|\Psi'(x)|\le C\big(J_1+J_2+J_3+J_4\big)\le C(J_1+J_3),$$
where
\begin{align*}
&J_1:={\la x\ra}^{2\l-1}(\lambda_m-V(x))^{-\frac{5}{4}}(\lambda_n-V(x))^{-\frac{1}{4}},&&J_2:={ \la x\ra}^{2\l-1}(\lambda_m-V(x))^{-\frac{1}{4}} (\lambda_n-V(x))^{-\frac{5}{4}},\\
&J_3:=\frac{(\lambda_m-V(x))^{\frac{1}{4}} (\lambda_n-V(x))^{-\frac{1}{4}}}{X_m^{2\l-1}(X_m-x)^3},&&J_4:=
\frac{(\lambda_m-V(x))^{-\frac{1}{4}} (\lambda_n-V(x))^{\frac{1}{4}}}{X_n^{2\l-1}(X_n-x)^3}.
\end{align*}
\end{Corollary}
Lemma \ref{zeroxn} is a direct corollary of Lemma \ref{xnsimple}  and \ref{xncomplex} in which  we suppose $X_n>4X_m$ and $X_m\le X_n \le4X_m$ respectively.
 \subsubsection{\bf the Integral on $[0, X_n)$ when $X_n>4X_m$}
In this case, from \eqref{potecd1} and Lemma \ref{potesim} we have
\begin{eqnarray}\label{yinyong0222}
\lambda_m=V(X_m)\le V(\frac14X_n)\le\frac14V(X_n)=\frac14\lambda_n,
\end{eqnarray}
for $m\ge n_0$.
\begin{lemma}\label{xnsimple}
If $f(x)$ satisfies Assumption \ref{pertass} and $X_n>4X_m$, then
\[
\Big|\int^{X_n}_0 f(x)e^{{\rm i}kx}h_m(x)\overline{h_n(x)}dx\Big|\le
C(|k|^{\frac1{\l}}\vee 1)(X_mX_n)^{\frac\mu2-\frac12},\quad \forall~k\ne0,
\]
where $C$ only depend on $(\mu,\l)$ and $n_0\le m\le n$.
\end{lemma}
Lemma \ref{xnsimple} is a direct corollary from Lemma \ref{xmbotsimple}, \ref{xmcritic} and \ref{xmxnsimple}.
\begin{lemma}\label{xmbotsimple}
If $f(x)$ satisfies Assumption \ref{pertass} and $X_n>4X_m$,  then
\[
\Big|\int_{0}^{X_m-X_m^{-\frac{1}{3}}} f(x)e^{{\rm i}kx}h_m(x)\overline{h_n(x)}dx\Big|\le C(|k|^{\frac1{\l}}\vee1)(X_mX_n)^{\frac\mu2-\frac12},\quad \forall~k\ne0,
\]
where $C$ only depends on $(\mu,\l)$ and $n_0\le m\le n$.
\end{lemma}
\begin{proof}
Clearly, we have
\[
\int^{X_m-X_m^{-\frac{1}{3}}}_{0} \cal F(x) dx=
C_mC_n\int^{X_m-X_m^{-\frac{1}{3}}}_{0} f(x)e^{{\rm i}(\zeta_m-\zeta_n+kx)}\Psi(x)dx,
\]
where $C_m\sim X_m^{\frac{\l-1}2},~C_n\sim X_n^{\frac{\l-1}2}$.  In the following discussion we always suppose $k\neq 0$.
We estimate it under two cases: $k<\frac{\sqrt{2D_1}}{8}X_n^\l$ and $k\ge\frac{\sqrt{2D_1}}{8}X_n^\l$ . \\
Case 1): $0\ne k<\frac{\sqrt{2D_1}}{8}X_n^\l$.
When $x\in[0,X_m-X_m^{-\frac13}]$, by \eqref{potecd2} we have $V(x)\ge-V(R)$. Together with
$V(R)\le V(X_{n_0})\le\lambda_m\le\lambda_n$ and (\ref{yinyong0222}), one has
\begin{align*}
g(x)&=\frac{\lambda_n-\lambda_m}{\sqrt{\lambda_m-V(x)}+\sqrt{\lambda_n-V(x)}}-k\ge\frac{\lambda_n-\lambda_m}{\sqrt{2\lambda_m}+\sqrt{2\lambda_n}}-k\\
&=\frac{\sqrt2}2(\sqrt{\lambda_n}-\sqrt{\lambda_m})-k\ge\frac{\sqrt2}4\sqrt{\lambda_n}-k\ge\frac{\sqrt{2D_1}}4X_n^\l-k\ge\frac{\sqrt{2D_1}}8X_n^\l.
\end{align*}
Thus, by Lemma \ref{oscint} we have
\begin{align*}
&\Big|\int^{X_m-X_m^{-\frac{1}{3}}}_{0}f(x)e^{{\rm i}(\zeta_m-\zeta_n+kx)}\Psi(x)dx\Big|\\
\le&~CX_n^{-\l}\bigg(X_m^\mu\Big(\big|\Psi({\scs X_m-X_m^{-\frac{1}{3}}})\big| +\int_{0}^{X_m-X_m^{-\frac{1}{3}}}(J_1+J_3)dx\Big)+\int_{0}^{X_m-X_m^{-\frac{1}{3}}}\la x\ra^{\mu-1}|\Psi(x)|dx\bigg).
\end{align*}
By Corollary \ref{psiesti} we have
\(
\big|\Psi({\scs X_m-X_m^{-\frac13}})\big|\le C X_m^{-\l+\frac23}
\) and
\[
 \int_{0}^{X_m-X_m^{-\frac13}}{\la x\ra}^{\mu-1}|\Psi(x)|dx
\le CX_m^{-\l+\frac23}\int_{0}^{X_m-X_m^{-\frac13}}{\la x\ra}^{\mu-1}dx\\
\le C X_m^{\mu-\l+\frac23}.
\]
Besides, one has
\begin{align*}
\int_{0}^{X_m-X_m^{-\frac{1}{3}}}J_1\,dx &\le C \int_{0}^{X_m-X_m^{-\frac{1}{3}}} {\la x\ra}^{2\l-1}(\lambda_m-V(x))^{-\frac32}dx\\
&\le C X_m^{2\l-1} X_m^{-3\l+\frac32} \int_{0}^{X_m-X_m^{-\frac{1}{3}}}(X_m-x)^{-\frac{3}{2}}dx \le C X_m^{-\l+\frac{2}{3}}
\end{align*}
and
$
\int_{0}^{X_m-X_m^{-\frac{1}{3}}}J_3\,dx \le C X_m^{-\l+\frac{2}{3}}.
$
It follows that
\[
\Big|\int^{X_m-X_m^{-\frac{1}{3}}}_{0} \cal F(x) dx\Big|\le C X_m^{-\frac{\l}{2}+\frac16+\mu}X_n^{-\frac{\l+1}{2}}
\le C(X_mX_n)^{\frac\mu2-\frac12}.
\]
For the remainder terms, since $\lambda_m\le\frac14\lambda_n$, then
\begin{align*}
&\Big|\int^{X_m-X_m^{-\frac{1}{3}}}_{0} f(x) e^{{\rm i}kx}\psi_2^{(m)}(x)\overline{\psi_1^{(n)}(x)}dx\Big|\\
\le&~ CX_m^{\frac{\l-1}2}X_n^{\frac{\l-1}2}X_m^{\mu-(\l+1)}\int_{0}^{X_m-X_m^{-\frac{1}{3}}}(\lambda_m-V(x))^{-\frac{1}{4}}(\lambda_n-V(x))^{-\frac{1}{4}}dx
\le C(X_mX_n)^{\frac\mu2-\frac12}.
\end{align*}
In the same way we have
$\Big|\int^{X_m-X_m^{-\frac{1}{3}}}_{0} f(x) e^{{\rm i}kx}\psi_1^{(m)}(x)\overline{\psi_2^{(n)}(x)}dx\Big|\le C(X_mX_n)^{\frac\mu{2}-\frac1{2}}$
and
$$\Big|\int^{X_m-X_m^{-\frac{1}{3}}}_{0} f(x) e^{{\rm i}kx}\psi_2^{(m)}(x)\overline{\psi_2^{(n)}(x)}dx\Big|\le C(X_mX_n)^{\frac\mu{2}-\frac1{2}}.$$
Therefore, for the first case we have
\(\dss
\Big|\int^{X_m-X_m^{-\frac13}}_{0} f(x) e^{{\rm i}kx}h_m(x)\overline{h_n(x)}dx\Big|\le C(X_mX_n)^{\frac\mu{2}-\frac1{2}}.
\)\\
Case 2): $k\ge\frac{\sqrt{2D_1}}{8}X_n^\l $. By H\"older inequality we have
\[
\Big|\int_0^{X_m-X_m^{-\frac13}} f(x) e^{{\rm i}kx}h_m(x)\overline{h_n(x)}dx\Big| \le CX_m^\mu
\le C(k^{\frac1{\l}}\vee 1)(X_mX_n)^{\frac\mu{2}-\frac1{2}}.
\]
Combining above estimations, we obtain
\[
\Big|\int_0^{X_m-X_m^{-\frac13}} f(x) e^{{\rm i}kx}h_m(x)\overline{h_n(x)}dx\Big| \le C(|k|^{\frac1\l}\vee1)(X_mX_n)^{\frac\mu{2}-\frac12},\quad \forall~k\ne0.
\]
\end{proof}
\begin{lemma}\label{xmcritic}
 If $f(x)$ satisfies Assumption \ref{pertass} and $X_n>4X_m$, then
\[
\Big|\int_{X_m-X_m^{-\frac{1}{3}}}^{X_m} f(x) e^{{\rm i}kx}h_m(x) \overline{h_n(x)}dx\Big| \le
C(X_mX_n)^{ \frac\mu2-\frac12},\quad \forall~k\ne0,
\]
where $C$ only depend on $(\mu,\l)$ and $n_0\le m\le n$.
\end{lemma}
\begin{proof}
Recall that $\lambda_m\le\frac14\lambda_n$, then  we have
\begin{align*}
&\Big|\int_{X_m-X_m^{-\frac{1}{3}}}^{X_m} \cal F(x) dx\Big| \le CX_m^{\frac{\l-1}2}X_n^{\frac{\l-1}2}\int_{X_m-X_m^{-\frac{1}{3}}}^{X_m} x^\mu(\lambda_m-V(x))^{-\frac{1}{4}}(\lambda_n-V(x))^{-\frac{1}{4}}dx\\
\le&~CX_m^{-\frac{1}{4}+\mu}X_n^{\frac{\l-1}2} (\lambda_n-\lambda_m)^{-\frac{1}{4}}\int_{X_m-X_m^{-\frac{1}{3}}}^{X_m}(X_m-x)^{-\frac{1}{4}}dx\le C(X_mX_n)^{\frac\mu2-\frac{1}{2}}.
\end{align*}
Similarly, we have
\[
\Big|\int_{X_m-X_m^{-\frac{1}{3}}}^{X_m} f(x) e^{{\rm i}kx}\psi_{j_1}^{(m)}(x) \overline{\psi_{j_2}^{(n)}(x)}dx\Big| \le C (X_mX_n)^{\frac\mu2-\frac{1}{2}},~\text{for}~j_1,j_2\in\{1,2\}~\text{and}~j_1+j_2\ge3.
\]
Hence,
\(\dss
\Big|\int_{X_m-X_m^{-\frac{1}{3}}}^{X_m} f(x) e^{{\rm i}kx}h_m(x) \overline{h_n(x)}dx\Big| \le
C(X_mX_n)^{\frac\mu{2}-\frac{1}{{2}}}.
\)
\end{proof}

\begin{lemma}\label{xmxnsimple}
If  $f(x)$ satisfies Assumption \ref{pertass} and $X_n>4X_m$, then
\[
\Big|\int_{X_m}^{X_n}f(x)e^{{\rm i}kx}h_m(x)\overline{h_n(x)}dx\Big|
\le C(X_mX_n)^{\frac\mu2-\frac12},\quad \forall~k\ne0,
\]
where $C$ only depends on $(\mu,\l)$ and $n_0\le m\le n$.
\end{lemma}
For the proof see section \ref{section4}.
\subsubsection{\bf the Integral on $[0, X_n)$ when $X_m\le X_n\le4X_m$}
We will prove  the following lemma in this part.
\begin{lemma}\label{xncomplex}
If  $f(x)$ satisfies Assumption \ref{pertass} and $X_m\le X_n\le4X_m$, then
\[
\Big|\int_{0}^{X_n} f(x) e^{{\rm i}kx}h_m(x)\overline{h_n(x)}dx\Big|\le C(|k|^{-1}\vee |k|)(X_mX_n)^{\frac\mu2-\frac12(\frac13\bigwedge\frac{\mu+1}{2\mu+2\l+1})},\quad \forall~k\ne0,
\]
where $C$ only depends on $(\mu,\l)$ and $n_0\le m\le n$.
\end{lemma}
\indent If $R\geq 8$, $X_m\ge8$ from $m\ge n_0$ and  $X_m-X_m^{-\frac13}\ge X_m^{\frac23}$. Hence, we can split the integral into three parts as:
\[
\int^{X_n}_0 f(x) e^{{\rm i}kx}h_m(x)\overline{h_n(x)}dx = \bigg(\int^{X_m^\frac23}_0+\int_{X_m^\frac23}^{X_m-X_m^\frac{1}{3}}+ \int_{X_m-X_m^\frac{1}{3}}^{X_n}\bigg)dx.
\]
Lemma \ref{xncomplex} comes from Lemma \ref{xmbotcomplex}, \ref{xmcomplex} and \ref{xmxncomplex}.  For the first part of the above integral, we have
\begin{lemma}\label{xmbotcomplex}
If  $f(x)$ satisfies Assumption \ref{pertass} and $X_m\le X_n\le4X_m$, then
\[
\Big|\int_{0}^{X_m^\frac23} f(x) e^{{\rm i}kx}h_m(x)\overline{h_n(x)}dx\Big|\le C(X_mX_n)^{\frac\mu2-\frac16},\quad \forall~k\ne0,
\]
where $C$ only depends on $(\mu,\l)$ and $n_0\le m\le n$.
\end{lemma}
\begin{proof}
By Lemma \ref{qQesti} we have
\begin{align*}
&\Big|\int_{0}^{X_m^\frac23} \cal F(x) dx\Big|\le CX_m^{\frac{\l-1}2}X_n^{\frac{\l-1}2}X_m^{\frac{2\mu}3}\int_{0}^{X_m^\frac23} (\lambda_m-V(x))^{-\frac{1}{4}}(\lambda_n-V(x))^{-\frac{1}{4}}dx\le C X_m^{\mu-\frac13}.
\end{align*}
Similarly, we have
\(\dss
\Big|\int_{0}^{X_m^\frac23}f(x) e^{{\rm i}kx} \psi_{j_1}^{(m)}(x) \overline{\psi_{j_2}^{(n)}(x)}dx\Big| \le C X_m^{\mu-\frac13},~\text{for}~j_1,j_2\in\{1,2\}~\text{and}~j_1+j_2\ge3.
\)\\
Thus, we obtain
\(\dss
\Big|\int_0^{X_m^\frac23} f(x) e^{{\rm i}kx}h_m(x)\overline{h_n(x)}dx\Big|\le C(X_mX_n)^{\frac\mu{2}-\frac{1}{{6}}}.
\)
\end{proof}
Next we estimate the integral on $[X_m^\frac23,X_m-X_m^{\frac{1}{3}}]$, and obtain the following lemma.
\begin{lemma}\label{xmcomplex}
If $f(x)$ satisfies Assumption \ref{pertass} and $X_m\le X_n\le4X_m$, then one has
\[
\Big|\int_{X_m^\frac23}^{X_m-X_m^\frac{1}{3}} f(x) e^{{\rm i}kx}h_m(x)\overline{h_n(x)}dx\Big|\le C(|k|^{-1}\vee|k|)(X_m X_n)^{\frac\mu2-\frac12(\frac13\bigwedge\frac{\mu+1}{2\mu+2\l+1})},\quad \forall~k\ne0,
\]
where $C$ only depends on $(\mu,\l)$ and $n_0\le m\le n$.
\end{lemma}
If $k>X_m^\frac13$, then $X_n^{\frac13}\le Ck$. By H\"older inequality we have
\begin{equation}\label{kxmcomplex}
\Big|\int_{X_m^{\frac23}}^{X_m-X_m^{\frac13}}f(x) e^{{\rm i}kx}h_m(x)\overline{h_n(x)}dx\Big| \le C X_m^\mu\le Ck(X_mX_n)^{\frac\mu2-\frac16}.
\end{equation}
Thus we prove Lemma \ref{xmcomplex} when $ k>X_m^\frac13$. In the following we turn to the case when $0<k \le X_m^\frac13$.
From Lemma \ref{xmcomplex1}  to Lemma \ref{xmcomplex5}
we always suppose the following assumptions:
1.  $f(x)$ satisfies Assumption \ref{pertass}; 2.  $X_m\le X_n\le4X_m$; 3. $0<k \le X_m^\frac13$.
\begin{lemma}\label{xmcomplex1}
If $0\le \lambda_n-\lambda_m<\sqrt{a_1}kX_m^{\l-\frac{1}{3}}$, then
\[
\Big|\int_{X_m^\frac23}^{X_m-X_m^\frac{1}{3}} f(x)e^{{\rm i}kx}h_m(x)\overline{h_n(x)}dx\Big|\le C(k^{-1}\vee 1)(X_mX_n)^{\frac\mu2-\frac1{3}},\quad\forall~k\in(0,X_m^\frac13],
\]
where $C$ only depends on $(\mu,\l)$ and $n_0\le m\le n$.
\end{lemma}
\begin{proof}
\indent Write $\dss I:=\int_{X_m^\frac23}^{X_m-X_m^\frac{1}{3}} \cal F(x) dx$ and
$I=C_mC_n\int_{X_m^\frac23}^{X_m-X_m^\frac{1}{3}}f(x) e^{{\rm i}(\zeta_m-\zeta_n+kx)}\Psi(x)dx$ with $C_m\sim X_m^{\frac{\l-1}2}$ and $C_n\sim X_n^{\frac{\l-1}2}$.
Since $\lambda_m-V(X_m-X_m^{\frac13})\ge a_1X_m^{2\l-\frac23}$, then
\begin{align*}
g({\scs X_m-X_m^\frac{1}{3}})&=
\frac{\lambda_n-\lambda_m}{\sqrt{\lambda_n-V(X_m-X_m^{\frac{1}{3}})}+\sqrt{\lambda_m-V(X_m-X_m^{\frac13})}}-k\\
&\le\frac{\sqrt{a_1}k X_m^{\l-\frac13}}{2\sqrt{\lambda_m-V(X_m-X_m^{\frac13})}}-k
\le\frac{\sqrt{a_1}k X_m^{\l-\frac13}}{2\sqrt{a_1}X_m^{\l-\frac{1}{3}}}-k= -\frac{k}{2}.
\end{align*}
By $g'(x)\ge0$  one obtains
$|g(x)|\ge\frac{k}{2}$ for $x\in[X_m^\frac23,X_m-X_m^\frac{1}{3}]$ as the figure \ref{case1} below.
\begin{figure}[H]
\centering
\epsfig{file=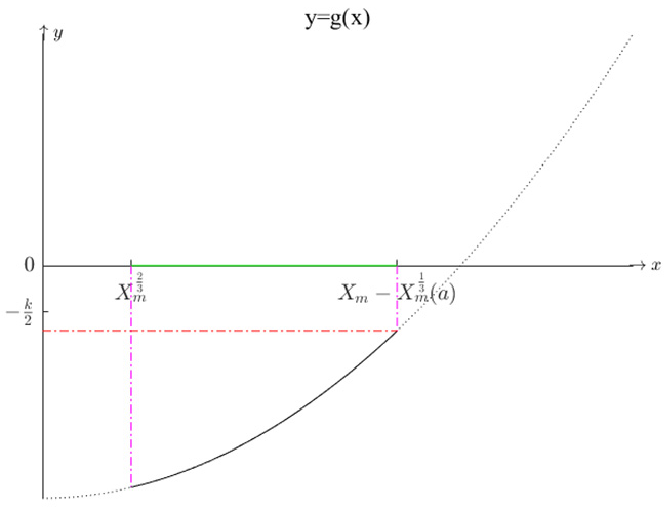,width=0.45\textwidth}
\caption{Phase in Lemma \ref{xmcomplex1}}\label{case1}
\end{figure}
Then by Lemma \ref{oscint} one obtains
\begin{align*}
&\Big| \int_{X_m^\frac23}^{X_m-X_m^\frac{1}{3}}f(x) e^{{\rm i}(\zeta_m-\zeta_n+kx)}\Psi(x)dx\notag\\
\le&~\frac{C}{k}\bigg(X_m^{\mu}\Big(\big|\Psi({\scs X_m-X_m^{\frac{1}{3}}})\big|
+\int_{X_m^\frac23}^{X_m-X_m^\frac{1}{3}}|\Psi'(x)|dx\Big)
+\int_{X_m^\frac23}^{X_m-X_m^\frac{1}{3}}x^{\mu-1}|\Psi(x)|dx\bigg).
\end{align*}
By corollary \ref{psiesti} we have
\(
\big|\Psi({\scs X_m-X_m^{\frac{1}{3}}})\big|\le C X_m^{-\l+\frac{1}{3}}
\) and
$
\int_{X_m^\frac23}^{X_m-X_m^\frac{1}{3}}x^{\mu-1}|\Psi(x)|dx\le  CX_m^{\mu-\l+\frac13}.
$
Besides, one has
\begin{align*}
\int_{X_m^\frac23}^{X_m-X_m^\frac{1}{3}}J_1dx &\le  C \int_{X_m^\frac23}^{X_m-X_m^\frac{1}{3}} x^{2\l-1}(\lambda_m-V(x))^{-\frac{5}{4}}(\lambda_n-V(x))^{-\frac{1}{4}}dx\\
&\le C \int_{X_m^\frac23}^{X_m-X_m^\frac{1}{3}} x^{2\l-1}(\lambda_m-V(x))^{-\frac{3}{2}}dx \le C X_m^{-\l+\frac{1}{3}},
\end{align*}
and
$$
\int_{X_m^\frac23}^{X_m-X_m^\frac{1}{3}}J_3dx\le C\int_{X_m^\frac23}^{X_m-X_m^\frac{1}{3}}
\frac{(\lambda_m-V(x))^{\frac{1}{4}} (\lambda_n-V(x))^{-\frac{1}{4}}}{X_m^{2\l-1}(X_m-x)^3}dx\le C X_m^{-\l+\frac13}.
$$
It follows that
$
\Big|\int_{X_m^\frac23}^{X_m-X_m^\frac{1}{3}}\cal F(x) dx\Big| \le \frac{C}{k} X_m^{\mu-\frac{2}{3}}.
$
Similarly,
$$
\Big|\int_{X_m^\frac23}^{X_m-X_m^\frac{1}{3}}f(x) e^{{\rm i}kx}\psi^{(m)}_2(x)\overline{\psi^{(n)}_1(x)}dx\Big| \le CX_m^{\mu-2}\int_{X_m^\frac23}^{X_m-X_m^\frac{1}{3}}(\lambda_m-V(x))^{-\frac{1}{2}}dx \le C (X_mX_n)^{\frac\mu2-\frac\l2-\frac12}.$$
The other two terms have same estimates. Therefore,
\(\dss
\Big|\int_{X_m^\frac23}^{X_m-X_m^\frac{1}{3}} f(x)e^{{\rm i}kx}h_m(x)\overline{h_n(x)}dx\Big|\le C(k^{-1}\vee 1)(X_mX_n)^{\frac\mu2-\frac13}.
\)
\end{proof}

\begin{lemma}\label{xmcomplex2}
If $\sqrt{a_1}kX_m^{\l-\frac13}\le \lambda_n-\lambda_m<\sqrt{a_1}kX_m^{\l-\frac16}$, then
\[
\Big|\int_{X_m^\frac23}^{X_m-X_m^\frac{1}{3}}f(x) e^{{\rm i}kx}h_m(x)\overline{h_n(x)}dx\Big|\le C(k^{-1}\vee 1)(X_mX_n)^{\frac\mu2-\frac1{6}},\quad\forall~k\in(0,X_m^\frac13],
\]
where $C$ only depends on $(\mu,\l)$ and $n_0\le m\le n$.
\end{lemma}
We put the proof of Lemma \ref{xmcomplex2} to Lemma \ref{xmcomplex5} into section \ref{section4}.
\begin{lemma}\label{xmcomplex3}
If $\sqrt{a_1}kX_m^{\l-\frac16}\le \lambda_n-\lambda_m<\sqrt{a_1}kX_m^\l$, then
\[
\Big|\int_{X_m^\frac23}^{X_m-X_m^\frac{1}{3}}f(x)e^{{\rm i}kx}h_m(x)\overline{h_n(x)}dx\Big|\le C(k^{-1}\vee1)(X_mX_n)^{\frac\mu2-\frac1{6}},\quad\forall~k\in(0,X_m^\frac13],
\]
where $C$ only depends on $(\mu,\l)$ and $n_0\le m\le n$.
\end{lemma}

\begin{lemma}\label{xmcomplex4}
If $\sqrt{a_1}kX_m^\l\le\lambda_n-\lambda_m <3\sqrt{D_2}kX_n^\l$, then
\[
\Big|\int_{X_m^\frac23}^{X_m-X_m^\frac{1}{3}}f(x)e^{{\rm i}kx}h_m(x)\overline{h_n(x)}dx\Big|\le C(k^{-1}\vee 1)(X_mX_n)^{\frac\mu2-\frac{\mu+1}{2(2\mu+2\l+1)}},\quad\forall~k\in(0,X_m^\frac13],
\]
where $C$ only depends on $(\mu,\l)$ and $n_0\le m\le n$.
\end{lemma}

\begin{lemma}\label{xmcomplex5}
If $\lambda_n-\lambda_m \ge3\sqrt{D_2}kX_n^\l$, then
\[
\Big|\int_{X_m^\frac23}^{X_m-X_m^\frac{1}{3}}f(x)e^{{\rm i}kx}h_m(x)\overline{h_n(x)}dx\Big|\le  C(k^{-1}\vee 1)(X_mX_n)^{\frac\mu2-\frac13},\quad\forall~k\in(0,X_m^\frac13],
\]
where $C$ only depends on $(\mu,\l)$ and $n_0\le m\le n$.
\end{lemma}

\begin{lemma}\label{xmposi}
If  $f(x)$ satisfies Assumption \ref{pertass} and $X_m\le X_n\le4X_m$, then one has
\[
\Big|\int_{X_m^\frac23}^{X_m-X_m^\frac{1}{3}}f(x) e^{{\rm i}kx}h_m(x)\overline{h_n(x)}dx\Big|\le  C(|k|^{-1}\vee|k|)(X_mX_n)^{\frac\mu2-\frac13},\quad\forall~k<0,
\]
where $C$ only depends on $(\mu,\l)$ and $n_0\le m\le n$.
\end{lemma}

\begin{proof}
Clearly, if $x\in[X_m^\frac23,~X_m-X_m^\frac13]$, then
\[
g(x)=\sqrt{\lambda_n-V(x)}-\sqrt{\lambda_m-V(x)}-k\ge-k=|k|.
\]
The proof is similar as Lemma \ref{xmcomplex1}.
\end{proof}
Combining Lemma \ref{xmcomplex1}-\ref{xmposi} with \eqref{kxmcomplex}, one completes the proof of Lemma \ref{xmcomplex}.
Then we turn to the last term and obtain the following lemma.
\begin{lemma}\label{xmxncomplex}
If  $f(x)$ satisfies Assumption \ref{pertass} and $X_m\le X_n\le4X_m$, then
\[
\Big|\int_{X_m-X_m^\frac{1}{3}}^{X_n}f(x) e^{{\rm i}kx}h_m(x)\overline{h_n(x)}dx\Big|\le C(X_mX_n)^{\frac\mu{2}-\frac{1}{6}},\quad \forall~k\ne0,
\]
where $C$ only depends on $(\mu,\l)$ and $n_0\le m\le n$.
\end{lemma}
\begin{proof}
By Lemma \ref{qQesti} we have $\lambda_m-V(x)\ge a_1X_m^{2\l-1}(X_m-x)$ for $x\in[X_m-X_m^{\frac13},X_m]$, then
\begin{align*}
\Big|\int_{X_m-X_m^\frac{1}{3}}^{X_m} \cal F(x) dx\Big| & \le CX_m^{\frac{\l-1}2}X_n^{\frac{\l-1}2}
\int_{X_m-X_m^\frac{1}{3}}^{X_m}x^\mu(\lambda_m-V(x))^{-\frac{1}{4}}(\lambda_n-V(x))^{-\frac{1}{4}}dx\\
& \le C(X_mX_n)^{\frac\mu2-\frac16}.
\end{align*}
For the integral on $[X_m,X_n]$, we  discuss it under two cases.\\
1). $X_n - X_n^{-\frac{1}{3}} \ge X_m + X_m^{-\frac{1}{3}}$. We split the integral into three parts as:
\[
\int^{X_n}_{X_m} f(x) e^{{\rm i}kx}\psi_1^{(m)}(x)\overline{\psi_1^{(n)}(x)}dx=
\big(\int^{X_m+X_m^{-\frac{1}{3}}}_{X_m}+\int_{X_m+X_m^{-\frac{1}{3}}}^{X_n-X_n^{-\frac13}}+\int_{X_n-X_n^{-\frac13}}^{X_n}\big)dx.
\]
For the first part, since $\lambda_n-V(X_m+X_m^{-\frac13})\ge\lambda_n-V(X_n-X_n^{-\frac13})\ge a_1X_n^{2\l-\frac43}$, then
\begin{align*}
&\Big|\int^{X_m+X_m^{-\frac{1}{3}}}_{X_m}  \cal F(x) dx\Big| \le~CX_m^{\frac{\l-1}2+\mu}X_n^{\frac{\l-1}2}\int^{X_m+X_m^{-\frac{1}{3}}}_{X_m} (V(x)-\lambda_m)^{-\frac{1}{4}}(\lambda_n-V(x))^{-\frac{1}{4}}dx\\
\le~&C X_m^{-\frac{1}{4}+\mu} X_n^{\frac{\l-1}{2}} (\lambda_n-V(X_m+X_m^{-\frac{1}{3}}))^{-\frac{1}{4}}\int^{X_m+X_m^{-\frac{1}{3}}}_{X_m}(X_m-x)^{-\frac{1}{4}}dx \le C (X_mX_n)^{\frac\mu2-\frac{1}{3}}.
\end{align*}
By lemma \ref{potesim} we have $|\zeta_m(x)|\ge A_1X_m^{\l-\frac12}(x-X_m)^\frac32\ge A_1X_m^{\l-\frac23}(x-X_m)$ for $x\ge X_m+X_m^{-\frac13}$. Thus,
\begin{align*}
&\Big|\int_{X_m+X_m^{-\frac{1}{3}}}^{X_n-X_n^{-\frac{1}{3}}} \cal F(x) dx\Big| \le CX_m^{\frac{\l-1}2}X_n^{\frac{\l-1}2+\mu}\int_{X_m+X_m^{-\frac{1}{3}}}^{X_n-X_n^{-\frac{1}{3}}} (V(x)-\lambda_m)^{-\frac{1}{4}} (\lambda_n-V(x))^{-\frac{1}{4}}e^{-|\zeta_m|}dx\\
\le&~C X_m^{-\frac{1}{4}}X_n^{\frac{\l-1}2+\mu}\Big(\lambda_n-V\big(X_n-X_n^{-\frac{1}{3}}\big)\Big)^{-\frac{1}{4}} \int_{X_m+X_m^{-\frac{1}{3}}}^{X_n-X_n^{-\frac{1}{3}}} (x-X_m)^{-\frac{1}{4}} e^{-C_0(x-X_m)}dx\\
\le&~C X_m^{-\frac{1}{4}}X_n^{-\frac{1}{6}+\mu}\int_0^\infty t^{-\frac{1}{4}}e^{-t}dt\le C (X_mX_n)^{\frac\mu2-\frac{1}{6}}.
\end{align*}
For the last part, since
\(\dss
V(X_n-X_n^{-\frac13})-\lambda_m\ge V(X_m+X_m^{-\frac13})-\lambda_m\ge a_1X_m^{2\l-\frac43}
\), then
\begin{align*}
&\Big|\int_{X_n-X_n^{-\frac{1}{3}}}^{X_n}\cal F(x) dx\Big| \le CX_m^{\frac{\l-1}2}X_n^{\frac{\l-1}2+\mu}\int_{X_n-X_n^{-\frac{1}{3}}}^{X_n} (V(x)-\lambda_m)^{-\frac{1}{4}}(\lambda_n-V(x))^{-\frac{1}{4}}dx\\
\le&~C X_m^{\frac{\l-1}2}\Big(V\big(X_n-X_n^{-\frac{1}{3}}\big)-\lambda_m\Big)^{-\frac{1}{4}}X_n^{-\frac{1}{4}+\mu} \int_{X_n-X_n^{-\frac{1}{3}}}^{X_n} (X_n-x)^{-\frac{1}{4}}dx\\
\le&~C X_m^{-\frac16}X_n^{-\frac12+\mu}\le C (X_mX_n)^{\frac\mu2-\frac{1}{3}}.
\end{align*}
Thus, for the first case we have
$
\Big|\int_{X_m}^{X_n} \cal F(x) dx\Big|\le C(X_mX_n)^{\frac\mu2-\frac16}.
$\\
2). $X_n - X_n^{-\frac{1}{3}} < X_m + X_m^{-\frac{1}{3}}$. Note $X_m\ge8$, it follows $X_n-X_m\le X_n^{-\frac13}+X_m^{-\frac13}\le1$. Hence,
\begin{align*}
&\Big|\int_{X_m}^{X_n} \cal {F}(x) dx\Big|\le C X_m^{\frac{\l-1}2}X_n^{\frac{\l-1}2+\mu}\int_{X_m}^{X_n} (V(x)-\lambda_m)^{-\frac{1}{4}}(\lambda_n-V(x))^{-\frac{1}{4}}dx\\
&\le C X_m^{-\frac{1}{4}}X_n^{-\frac{1}{4}+\mu}\int_{X_m}^{X_n} (x-X_m)^{-\frac{1}{4}}(X_n-x)^{-\frac{1}{4}}dx  \le C (X_mX_n)^{\frac\mu2-\frac{1}{4}}.
\end{align*}
 Thus, in the second case we have
$
\Big|\int_{X_m}^{X_n} \cal F(x) dx\Big|\le C(X_mX_n)^{\frac\mu2-\frac{1}{4}}$.
Since the other three integrals have better estimates, we finish the proof.
\end{proof}
\subsection{Proof of Lemma \ref{mainlem} and Lemma \ref{mainlem2}}
 Proof of Lemma \ref{mainlem}.
If denote $\wt V(x)=V(-x),~\wt{h_n}(x)=h_n(-x)$ and $\wt f(x)=f(-x)$, we have
\[
\Big(-\frac{d^2}{dx^2}+\widetilde V(x)\Big)\widetilde {h_n}(x)=\lambda_n\widetilde{h_n}(x),\quad x\in\R.
\]
Clearly, $\wt V(x)$ satisfies Assumption \ref{poteass} and $\wt f(x)$ satisfies Assumption \ref{pertass}.
Applying Lemma \ref{realplus}, we obtain
\[
\Big|\int_{0}^{+\infty} \wt f(y) e^{-{\rmi}ky}\wt{h_m}(y)\overline{\wt{h_n}(y)}dx\Big|\le C(|k|^{-1}\vee|k|)(\lambda_m\lambda_n)^{\frac\mu{4\l}-\frac1{4\l}(\frac13\bigwedge\frac{\mu+1}{2\mu+2\l+1})},\quad \forall~k\ne 0.
\]
Note that
$
\int_{0}^{+\infty} \wt f(y) e^{{\rmi}(-k)y}\wt{h_m}(y)\overline{\wt{h_n}(y)}dy= \int^{0}_{-\infty} f(x) e^{{\rmi}kx}h_m(x)\overline{h_n(x)}dx,
$
it follows
\begin{eqnarray}\label{knegative}
\Big|\int^{0}_{-\infty} f(x) e^{{\rm i}kx}h_m(x)\overline{h_n(x)}dx\Big|
\le C(|k|^{-1}\vee|k|)(\lambda_m\lambda_n)^{\frac\mu{4\l}-\frac1{4\l}(\frac13\bigwedge\frac{\mu+1}{2\mu+2\l+1})},\quad\forall~k\ne0.
\end{eqnarray}
Combining Lemma \ref{realplus} with (\ref{knegative}), we finish the proof. \qed\\
\noindent Proof of  Lemma \ref{mainlem2}.  Similar as the above proof, we only need to estimate the integral on $[0,+\infty)$.
As Lemma \ref{mnsmall},~\ref{msmall} and \ref{xninf} one has
\begin{align*}
&\Big|\int_{0}^{+\infty} f(x)h_m(x)\overline{h_n(x)}dx\Big|\le C(X_mX_n)^{\frac\mu2-\frac12}
\le C(X_mX_n)^{\frac\mu2},\quad\text{for}~m\le n<n_1;\\
&\Big|\int_{0}^{+\infty} f(x)h_m(x)\overline{h_n(x)}dx\Big|\le C(X_mX_n)^{\frac\mu2-\frac16}
\le C(X_mX_n)^{\frac\mu2},\quad\text{for}~m<n_0,~n\ge n_1;\\
&\Big|\int_{X_n}^{+\infty} f(x)h_m(x)\overline{h_n(x)}dx\Big|\le C(X_mX_n)^{\frac\mu2-\frac13}
\le C(X_mX_n)^{\frac\mu2},\quad\text{for}~n\ge m\ge n_0.
\end{align*}
Next we estimate the integral on $[0,~X_n]$ for $n\ge m\ge n_0$. If $X_m\le X_n\le 4X_m$, by H\"older inequality we have
\[
\Big|\int_0^{X_n}f(x)h_m(x)\overline{h_n(x)}dx\Big|\le CX_n^{\mu}\le C(X_mX_n)^{\frac\mu2}.
\]
When $X_n>4X_m$ we have $X_n-X_n^{-\frac13}\ge\frac12X_n\ge2X_m$, and thus we split the integral into three parts as:
\[
\int_{0}^{X_n}f(x)h_m(x)\overline{h_n(x)}dx = \Big(\int_{0}^{2X_m}+\int_{2X_m}^{X_n-X_n^{-\frac13}}+\int_{X_n-X_n^{-\frac13}}^{X_n}\Big)f(x)h_m(x)\overline{h_n(x)}dx.
\]
By H\"older inequality one has
\(\dss
\Big|\int_{0}^{2X_m}f(x)h_m(x)\overline{h_n(x)}dx\Big|\le C X_m^{\mu}\le C(X_mX_n)^{\frac\mu2}.
\)\\
By Lemma \ref{qQesti} we have
$
\lambda_n-V(X_n-X_n^{-\frac13})\ge a_1X_n^{2\l-\frac43}$ and
$|\zeta_m(x)|\ge A_1X_m^\l(x-X_m)$  for $x\ge2X_m$.
It follows
\begin{align*}
&\Big|\int_{2X_m}^{X_n-X_n^{-\frac13}}\cal F(x) dx\Big| \leq CX_m^{\frac{\l-1}{2}}X_n^{\frac{\l-1}{2}} \int_{2X_m}^{X_n-X_n^{-\frac13}}x^\mu(V(x)-\lambda_m)^{-\frac14}(\lambda_n-V(x))^{-\frac14}e^{-|\zeta_m|}dx\\
\leq ~& C  X_m^{-\frac14}X_n^{\frac{\l-1}{2}}(\lambda_n-V(X_n-X_n^{-\frac13}))^{-\frac14} \int_{2X_m}^{X_n-X_n^{-\frac13}}(x-X_m)^{-\frac14+\mu} e^{-C_0(x-X_m)}dx\\
\leq ~& C X_m^{-\frac14}X_n^{-\frac16}\int_0^\infty t^{-\frac{1}{4}+\mu}e^{-t}dt\le CX_m^{-\frac14}X_n^{-\frac16}.
\end{align*}
Similarly, from  Lemma \ref{qQesti} we have
$
V(X_n-X_n^{-\frac13})-\lambda_m\ge a_1X_m^{2\l}$ and
$|\zeta_m(x)|\ge\frac{A_1}4X_m^\l X_n$ for $x\ge X_n-X_n^{-\frac13}$.
Therefore,
\begin{align*}
&\Big|\int_{X_n-X_n^{-\frac13}}^{X_n}\cal F(x) dx\Big| \leq CX_m^{\frac{\l-1}{2}}X_n^{\frac{\l-1}{2}} \int_{X_n-X_n^{-\frac13}}^{X_n}x^\mu(V(x)-\lambda_m)^{-\frac14}(\lambda_n-V(x))^{-\frac14}e^{-|\zeta_m|}dx \\
\leq ~& C X_n^{\frac \l2-\frac34}(V(X_n-X_n^{-\frac13})-\lambda_m)^{-\frac14}{e^{-C_0X_n}}\int_{X_n-X_n^{-\frac13}}^{X_n}(x-X_m)^\mu e^{-C_0(x-X_m)}(X_n-x)^{-\frac14}dx \\
\leq ~&  C X_n^{-\frac12}X_m^{-\frac12}.
\end{align*}
Since the other integrals have better estimates, we obtain
$
\Big|\int_0^{X_n}f(x)h_m(x)\overline{h_n(x)}dx\Big|\le C(X_mX_n)^{\frac\mu2}.
$
Thus we finish the proof. \qed

\section{Appendix}\label{section4}
\subsection{Lemma  \ref{algebra} (iv) and (v) }\noindent Proof of Lemma \ref{algebra} (iv). Since $A\in\M_\beta^+$, then
$|A_i^j|\le|A|_\beta^+(ij)^\beta(1+|i-j|)^{-1}(i^{\iota-1}+j^{\iota-1})^{-1}$ for any $i,j\in\Z_+$.
For a conjugated pair $(p,q)$, i.e. $\frac{1}{p}+\frac{1}{q}=1$ with $p,q\geq 1$,  given $i\in\Z_+$,
$\sum_{j\ge1}|A_i^j|(\fij)^{\frac sp}\le|A|_\beta^+\sum_{j\ge1}\frac{(ij)^\beta}{(1+|i-j|)(i^{\iota-1}+j^{\iota-1})}(\fij)^{\frac sp}$.
We first consider the sum
\[
I:=\sum_{j\ge1}\frac{(ij)^\beta}{(1+|i-j|)(i^{\iota-1}+j^{\iota-1})}(\fij)^{\frac sp}=\Big(\sum_{j\ge i}+\sum_{\frac i2\le j<i}+\sum_{j<\frac i2}\Big):=I_1+I_2+I_3.
\]
For $I_1$, by $\frac sp\ge0$, we have
\[
I_1\le\sum_{j\ge i}\frac{j^{2\beta}}{(1+|i-j|)j^{\iota-1}}
\le\sum_{j\ge1}\frac1{(1+|i-j|)j^{\iota-1-2\beta}}\le C(\beta,\iota).
\]
For $I_2$, we have
\[
I_2\le C(\beta,s)\sum_{\frac i2\le j<i}\frac{j^{2\beta}}{(1+|i-j|)j^{\iota-1}}\le C(\beta,s)
\sum_{j\ge1}\frac1{(1+|i-j|)j^{\iota-1-2\beta}}\le C(\beta,\iota,s).
\]
For $I_3$, if $\iota-\beta-\frac sp\ge0$, then
\begin{align*}
I_3&\le\sum_{j<\frac i2}\frac{2(ij)^\beta}{i^\iota}(\fij)^{\frac sp}
=\sum_{j<\frac i2}\frac2{i^{\iota-\beta-\frac sp}j^{\frac sp-\beta}}\\
\quad&\le\sum_{j<\frac i2}\frac2{j^{\iota-\beta-\frac sp}j^{\frac sp-\beta}}
\le\sum_{j\ge1}\frac2{j^{\iota-2\beta}}\le C(\beta,\iota).
\end{align*}
Hence, we obtain $I\le C(\beta,\iota,s)$ provided
\begin{equation}\label{seriesjCD}
s, \iota-\beta-\frac sp\ge0.
\end{equation}
For given $j\in\Z_+$,
since $\sum_{i\ge1}|A_i^j|(\fij)^{\frac sq}\le|A|_\beta^+\sum_{i\ge1}\frac{(ij)^\beta}{(1+|i-j|)(i^{\iota-1}+j^{\iota-1})}(\fij)^{\frac sq}$, as above we consider the sum
\[
J:=\sum_{i\ge1}\frac{(ij)^\beta}{(1+|i-j|)(i^{\iota-1}+j^{\iota-1})}(\fij)^{\frac sq}=\Big(\sum_{i<j}+\sum_{j\le i<2j}+\sum_{i\ge2j}\Big):=J_1+J_2+J_3.
\]
We discuss two cases in order to obtain the estimates of $I $ and $J$. \\
Case 1: $s\in [\beta,2\iota-2\beta-1)$.
For $J_1$, we choose $q\geq 1$ such that
$
\frac sq\ge \beta
$
and then \begin{align*}
J_1&\le\sum_{i<j}\frac{(ij)^\beta}{(1+|i-j|)i^{\iota-1}}(\fij)^\frac sq=\sum_{i<j}\frac1{(1+|i-j|)i^{\iota-1-\beta-\frac sq}j^{\frac sq-\beta}}\\
   &\le\sum_{i\ge1}\frac1{(1+|i-j|)i^{\iota-1-2\beta}}\le C(\beta,\iota).
\end{align*}
For $J_2$, we have
\[
J_2\le C(\beta,s)\sum_{j\le i<2j}\frac{i^{2\beta}}{(1+|i-j|)i^{\iota-1}}\le C(\beta,s)\sum_{i\ge1}\frac1{(1+|i-j|)i^{\iota-1-2\beta}}\le C(\beta,\iota,s).
\]
For $J_3$,
if we choose $q$ such that $\iota-\beta-\frac sq>1,~\frac sq-\beta\ge0$,  then
\[
J_3\le\sum_{i\ge2j}\frac{2(ij)^\beta}{i^\iota}(\fij)^\frac sq\le\sum_{i\ge1}\frac2{i^{\iota-\beta-\frac sq}j^{\frac sq-\beta}}\le C(\beta,\iota,s).
\]
Hence, when $s\in [\beta,2\iota-2\beta-1)$,  we obtain $I, J\le C(\beta,\iota,s)$ provided
\begin{equation}\label{seriesiCD}
\iota-\beta-\frac sp\ge0, \qquad \iota-\beta-\frac sq>1,\quad\frac sq-\beta\ge0.
\end{equation}
More precisely, one needs to choose $q\geq 1$ such that
$
\frac1q\in
\begin{cases}
[\frac\beta s,\frac{\iota-\beta-1}s),&\beta\le s< \iota,\\
[\frac{s-\iota+\beta}s,\frac{\iota-\beta-1}s),&\iota\le s<2\iota-2\beta-1.
\end{cases}
$
In fact one can choose
$
q(s)=
\begin{cases}
\frac s\beta,&\beta\le s<\iota,\\
\frac s{s-\iota+\beta},&\iota\le s<2\iota-2\beta-1
\end{cases}
$
such that (\ref{seriesiCD}) holds true.\\
Case 2: $s\in [0,\beta)$. In this case we estimate $J_1$ to $J_3$ again.  In fact,
For $J_1$, we choose $q=\infty$ and then
\begin{align*}
J_1&\le\sum_{i<j}\frac{(ij)^\beta}{(1+|i-j|)(i^{\iota-1}+j^{\iota-1})} \le\sum_{i\ge1}\frac1{(1+|i-j|)j^{\iota-1-2\beta}}\\
 &\le\sum_{i\ge1}\frac1{(1+|i-j|)i^{\iota-1-2\beta}}\le C(\beta,\iota).
\end{align*}
For $J_2$, the proof is the same.
For $J_3$,
if  $q=\infty$,  then
$
J_3\le C\sum_{i\ge1}\frac{1}{i^{\iota-2\beta}} \le C(\beta,\iota).
$
We remark in this case (\ref{seriesjCD}) also holds since  $\iota>2\beta+1$ when $q=\infty$.
Hence, when $s\in [0,\beta)$,  we still obtain $I, J\le C(\beta,\iota,s)$.
Thus, for any $s\in [0,2\iota-2\beta-1)$,  one can choose $p,q\geq 1$ such that   $1/p+1/q=1$ and
$
\sum_{j}|A_i^j|(\fij)^{\frac sp},~\sum_{i}|A_i^j|(\fij)^{\frac sq}\le C(\beta,\iota,s) |A|_{\beta}^{+}.
$
By H\"older inequality, we have
\begin{align*}
\|A\xi\|_s^2&=\sum_{i\ge1}i^s |\sum_{j\ge1}A_i^j\xi_j |^2
\le\sum_{i\ge1}\Big(\sum_{j\ge1}|A_i^j|(\fij)^{\frac s2}|\xi_j|j^{\frac s2}\Big)^2\\
&=\sum_{i\ge1}\Big(\sum_{j\ge1}|A_i^j|^\frac12(\fij)^{\frac s{2p}}|A_i^j|^\frac12(\fij)^{\frac s{2q}}|\xi_j|j^{\frac s2}\Big)^2\\
&\le\sum_{i\ge1}\Big(\sum_{j\ge1}|A_i^j|(\fij)^{\frac s{p}}\Big)
\Big(\sum_{j\ge1}|A_i^j|(\fij)^{\frac s{q}}|\xi_j|^2j^s\Big)\\
&\le C(\beta,\iota,s)|A|_{\beta}^{+}\sum_{i\ge1}\sum_{j\ge1}|A_i^j|(\fij)^{\frac s{q}}|\xi_j|^2j^s\\
&=C(\beta,\iota,s) |A|_{\beta}^{+}\sum_{j\ge1}|\xi_j|^2j^s\sum_{i\ge1}|A_i^j|(\fij)^{\frac s{q}}\le C^2(\beta,\iota,s)( |A|_{\beta}^{+})^2\|\xi\|_s^2.
\end{align*}
It follows $\|A\|_{\lst}\le C(\beta,\iota,s)|A|_{\beta}^{+}$.\qed \\
\noindent Proof of Lemma \ref{algebra} (v). As above  we define conjugated pair $(p,q)$, where $p,q$ is chosen as above.  From the proof in (iv) we have
for any $s\in[0,2\iota-2\beta-1)$,
$
\sum_{j\geq 1}|A_j^i|(\fij)^{\frac sp},~\sum_{i\ge 1}|A_j^i|(\fij)^{\frac sq}\le C(\beta,\iota,s) |A|_{\beta}^{+},
$
i.e.
\begin{eqnarray}\label{sxiaoyuling}
\sum_{i\ge 1}|A_i^j|(\fji)^{\frac sp}, \sum_{j\ge 1}|A_i^j|(\fji)^{\frac sq}\le C(\beta,\iota,s)|A|_{\beta}^{+}.
\end{eqnarray}
By  H\"older inequality, (\ref{sxiaoyuling}) and a similar method as (iv),  we obtain
\begin{align*}
\|A\xi\|_{-s}^2&\le \sum_{i\ge1}\Big(\sum_{j\ge1}|A_i^j|^\frac12(\fji)^{\frac s{2q}}|A_i^j|^\frac12(\fji)^{\frac s{2p}}|\xi_j|j^{-\frac s2}\Big)^2\\
&\le\sum_{i\ge1}\Big(\sum_{j\ge1}|A_i^j|(\fji)^{\frac sq}\Big)
\Big(\sum_{j\ge1}|A_i^j|(\fji)^{\frac sp}|\xi_j|^2j^{-s}\Big)\le (C(\beta,\iota,s)|A|_{\beta}^{+})^2\|\xi\|_{-s}^2.
\end{align*}
It results in $\|A\|_{\lms}\le C(\beta,\iota,s)|A|_\beta^+$.\qed
\subsection{some lemmas for section \ref{section2}}
\begin{lemma}\label{normmatope}
For $s\ge0$, if $A\in\M_\frac s2$ is a diagonal matrix, then $A\in\B(\l_0^2,\l_{-2s}^2)$ and satisfies
\[
\|A\|_{\B(\l_0^2,\l_{-2s}^2)}\le|A|_\frac s2.
\]
\end{lemma}

\begin{proof}
Since $A\in\M_\frac s2$ is a diagonal matrix, then $|A_i^i|\le|A|_\frac s2(ii)^\frac s2=|A|_\frac s2i^s$.
Given $\xi\in\l_0^2$, one has
\[
\|A\xi\|_{-2s}^2=\sum_{i\ge1}i^{-2s}\Big|\sum_{j\ge1}A_i^j\xi_j\Big|^2
\le\sum_{i\ge1}i^{-2s}|A_i^i|^2|\xi_i|^2
\le|A|_\frac s2^2\sum_{i\ge1}|\xi_i|^2=|A|_\frac s2^2\|\xi\|_0^2.
\]
It follows that $\|A\|_{\B(\l_0^2,\l_{-2s}^2)}\le|A|_\frac s2$.
\end{proof}

We present the following lemmas to prove  Lemma \ref{VlA0converge} and complete the proof of Theorem \ref{Theorem2.10}. The following proof is similar as   Lemma \ref{UA1} to Lemma \ref{UlAlconverge} and we don't give the details here.
\begin{lemma}\label{Psubl}
For $l\ge1$, if $\e_0\ll1$, then $V^lA^0=A^lV^l+Q_l$, where $Q_l\in\lzd$ and satisfies
$||Q_l||_{\lzd,s_l}\le\prod_{i=0}^{l-1}(1+\e_i^\frac12)+\sum_{i=0}^{l-1}\e_i^\frac12$.
\end{lemma}
From a straightforward computation, we have
$
\|Q_{l+1}-Q_l\|_{\lzi,\frac s2}\le 3C\e_l^\frac23.
$
It follows  $\{Q_l\}$ is a Cauchy sequence in $\lzd(\Pi_*,\frac s2)$ and $\lzi(\Pi_*,\frac s2)$.\\
 Define $Q_\infty:=Q_1+\sum_{l=1}^{\infty}Q_{l+1}-Q_l$. We have  the following
\begin{lemma}\label{Psublconverge}
 For $l\to \infty$, $Q_l(\phi)\to Q_\infty(\phi)$ in $\lzd(\Pi_*,\frac s2)$ and $\lzi(\Pi_*,\frac s2)$,
 where
\begin{align*}
 &\|Q_\infty\|_{\lzi,\frac s2}\le\|Q_\infty\|_{\lzd,\frac s2}\le C\e_0^\frac23,\\
 &\|Q_\infty-P^0\|_{\lzi,\frac s2}\le\|Q_\infty-P^0\|_{\lzd,\frac s2}\le C\e_0^\frac23
\end{align*}
for some positive constant $C$.
\end{lemma}

From
$
A^{l+1}V^{l+1}-A^lV^l=A^l\big(e^{-B^{l+1}}-\id\big)V^l+\text{diag}(P^l) V^{l+1}
$
and Lemma \ref{Alconverge}, \ref{VlId},  one obtains
$
\|A^{l+1}V^{l+1}-A^lV^l\|_{\lzi,\frac s2}\le2C\e_l^\frac23.
$
 Thus, $\{A^lV^l\}$ is a Cauchy sequence in $\lzi$. We now define
$
A^\infty V^\infty=A_0+\sum_{l=0}^\infty A^{l+1}V^{l+1}-A^lV^l,
$
it follows
$
\|A^\infty V^\infty-A^0\|_{\lzi,\frac s2}\le4C\e_0^\frac23.
$
Therefore, we have
\begin{lemma}\label{AlVlconverge}
When $l\to \infty$, $A^lV^l(\phi)\to A^\infty V^\infty(\phi)$ in $\lzi(\Pi_*,\frac s2)$, where
$$
\|A^\infty V^\infty-A^0\|_{\lzi,\frac s2}\le C\e_0^\frac23, \qquad C>0.
$$
\end{lemma}
Combining Lemma \ref{Psubl}, \ref{Psublconverge} with \ref{AlVlconverge}  one has
$
V^lA^0=A^lV^l+Q_l,
$
where $\{A^lV^l\},~\{Q_l\}$ are Cauchy sequences in $\lzi(\Pi_*,\frac s2)$ and
$
\|A^\infty V^\infty-A^0\|_{\lzi,\frac s2}\le C\e_0^\frac23
$
and
$
\|Q_\infty\|_{\lzi,\frac s2}\le C\e_0^\frac23
$
with some $C>0$. Thus,  $\{V^lA^0\}$ is a Cauchy sequence in $\lzi(\Pi_*,\frac s2)$. Define
$
\lim\limits_{l\to \infty}V^l A^0:=V^{\infty} A^0 = A^\infty V^\infty+Q_\infty,
$
then one has
\[
\|V^\infty A^0-A^0\|_{\lzi,\frac s2}\le\|A^\infty V^\infty-A^0\|_{\lzi,\frac s2}+\|Q_\infty\|_{\lzi,\frac s2}\le 2C\e_0^\frac23.
\]
It can be written as the following.
\begin{lemma}\label{VlA0converge}
When $l\to \infty$, $V^lA^0(\phi)\to V^\infty A^0(\phi)$ in $\lzi(\Pi_*,\frac s2)$, where
\[
\|V^\infty A^0-A^0\|_{\lzi,\frac s2}\le C\e_0^\frac23, \qquad C>0.
\]
\end{lemma}
\subsection{some lemmas for section \ref{section3}}

\noindent Proof of Lemma \ref{eigenfunction}. The proof is almost from \cite{T2}(see \cite{YajimaZhang}, Lemma 2.1). As in \cite{T2} we set
$\eta(x)=(\lambda-V(x))^\frac14h(x)$ and $\dss\zeta(x)=\int_X^x (\lambda-V(t))^\frac12 dt$,
where
\(
\arg\zeta(x)=
\begin{cases}
\frac\pi2,&x>X,\\
-\pi,&x<X.
\end{cases}
\)
Then the equation (\ref{tezhengfangcheng}) is transformed into
$\frac{d^2\eta}{d\zeta^2}+\eta+ [\frac{V^{\prime\prime}(x)}{4(\lambda-V(x))^2}+\frac{5V^\prime(x)^2}{16(\lambda-V(x))^3}]\eta=0$
and it can be rewritten as
\begin{eqnarray}\label{xinfangcheng}
\frac{d^2\eta}{d\zeta^2}+(1+\frac{5}{36\zeta^2})\eta=f(x)\eta,
\end{eqnarray}
where
$f(x)=\frac{5}{36\zeta^2}-\frac{V^{\prime\prime}(x)}{4(\lambda-V(x))^2}-\frac{5V^\prime(x)^2}{16(\lambda-V(x))^3}$.
As we know, Bessel equation
$\frac{d^2 G}{d\zeta^2}+(1+\frac{5}{36\zeta^2})G=0$
has two linearly independent solutions $(\frac{\pi\zeta}{2})^\frac12 J_\frac13(\zeta)$ and $(\frac{\pi\zeta}{2})^\frac12 H_\frac13^{(1)}(\zeta)$, where $J_\nu(x)$ and $H_\nu^{(1)}(x)$ are the first kind Bessel function and one of the third kind Bessel function, respectively. By the property of Bessel function that
$x(J_\nu(x)H_\nu^{(1)\prime}(x)-J_\nu^\prime(x)H_\nu^{(1)}(x))=\frac{2{\rm i}}{\pi}$,
then (\ref{xinfangcheng}) is formally equivalent to the integral equation
\begin{align*}
\eta = (\frac{\pi\zeta}{2})^\frac12H_{\frac13}^{(1)}(\zeta)+\frac{\pi{\rm i}}{2}\int_x^\infty \bigg(H_\frac13^{(1)}(\zeta)J_\frac13(\theta)-J_\frac13(\zeta)H_\frac13^{(1)}(\theta)\bigg)\zeta^\frac12\theta^\frac12 f(t)(\lambda-V(t))^\frac12\eta(t)dt,
\end{align*}
where we write $\zeta=\zeta(x)$ and $\theta=\zeta(t)$ for convenience. Set
$$\alpha(x)=e^{-{\rm i}\zeta}(\frac{\pi\zeta}{2})^\frac12 H_\frac13^{(1)}(\zeta),\quad \beta(x)=e^{{\rm i}\zeta}(\frac{\pi\zeta}{2})^\frac12 J_\frac13(\zeta),\quad \chi(x)=e^{-{\rm i}\zeta}\eta(x).$$
then
$$\chi(x)=\alpha(x)+{\rm i}\int_x^\infty \bigg(\alpha(x)\beta(t)-e^{2{\rm i}(\theta-\zeta)}\beta(x)\alpha(t)\bigg) f(t)(\lambda-V(t))^\frac12\chi(t)dt.$$
From \cite{T2}, $\alpha(x)$, $\beta(x)$ are bounded, and
$\Im(\theta-\zeta)=\Im(\int_x^t(\lambda-q(u))^\frac12du)\geq0$ together with Lemma \ref{6.5} and \ref{2.2sub}, we can prove that the iteration converges. In fact, if we denote
$\int_0^\infty |f(t)||\lambda-V(t)|^\frac12dt= M_0=O\left(\frac1{X\lambda^\frac12}\right)$,
and
$\left|\alpha(x)\beta(t)-e^{2{\rm i}(\theta-\zeta)}\beta(x)\alpha(t)\right|\le M$
uniformly, then
$|\chi_0(x)|=|\alpha(x)|\le C$, $|\chi_1(x)-\chi_0(x)|\le C M M_0$,
and generally, if
$|\chi_n(x)-\chi_{n-1}(x)|\le CM^n M_0^n$,
then
\begin{align*}
&|\chi_{n+1}(x)-\chi_n(x)|\\
=&\left|\int_x^\infty \bigg(\alpha(x)\beta(t)-e^{2{\rm i}(\theta-\zeta)}\beta(x)\alpha(t)\bigg) f(t)(\lambda-V(t))^\frac12(\chi_n(t)-\chi_{n-1}(t))dt\right|\\
\le& CM^{n+1}M_0^n\int_x^\infty \left|f(t)(\lambda-V(t))^\frac12\right|dt\le CM^{n+1}M_0^{n+1}.
\end{align*}
Thus,
\begin{align*}
|\chi_n(x)|&\le|\chi_0(x)|+|\chi_1(x)-\chi_0(x)|+\cdots+|\chi_n(x)-\chi_{n-1}(x)|\\
&\le C(1+MM_0+\cdots+M^n M_0^n)\le\frac{C}{1-MM_0}.
\end{align*}
If $n\geq n_0$, then $X_n\geq R$ and $X_n\sim \lambda_n^{\frac{1}{2\l}}$. Thus, choose $ n\geq n_*\geq n_0$ large enough such that $MM_0<1$, and by the theorem of dominated convergence, when $n\to\infty$, $\chi_n(x)\to\chi(x)=\alpha(x)+O(\frac{1}{X\lambda^\frac12})$ uniformly w.r.t $x$, which means that $\chi(x)$ is bounded.\\
Next we show that
\begin{eqnarray}\label{biaodaforchi}
\chi(x)=\alpha(x)\left(1+O\left(\frac{1}{X\lambda^\frac12}\right)\right).
\end{eqnarray}
In fact, similar as Lemma \ref{Bessel}, if $\zeta(x)<-c_0$ or ${\rm i}\zeta(x)<-c_0$, where $c_0$ are arbitrary two positive constants, then we can prove that $|\alpha(x)|>C$ and (\ref{biaodaforchi}) holds. While for $0<|\zeta(x)|\leq c_0$ we have
$|\beta(x)|\le C|\alpha(x)|$.
Thus,
$$|\chi(x)-\alpha(x)|=\left|\int_x^\infty \bigg(\alpha(x)\beta(t)-e^{2{\rm i}(\theta-\zeta)}\beta(x)\alpha(t)\bigg) f(t)(\lambda-V(t))^\frac12\chi(t)dt\right|\le \frac{C|\alpha(x)|}{X\lambda^\frac12}.$$
From the above proof when $\lambda>c_*>0$ large enough such that
\begin{eqnarray}\label{forremark1}
M\int_0^{\infty}|f(t)||\lambda-V(t)|^\frac12dt\le\frac{CM}{X\lambda^{\frac12}}<\frac{C_1}{\lambda^{\frac12+\frac{1}{2\l}}}<1,
\end{eqnarray}
the solution of (\ref{tezhengfangcheng}) can be written as
$h(x)=(\lambda-V(x))^{-\frac14}(\frac{\pi\zeta}{2})^\frac12 H_\frac13^{(1)}(\zeta)(1+O(\frac{1}{X\lambda^\frac12}))$. It follows that
$h_n(x)=(\lambda_n-V(x))^{-\frac14}(\frac{\pi\zeta_n}{2})^\frac12 H_\frac13^{(1)}(\zeta_n)(1+O(\frac{1}{X_n\lambda_n^\frac12}))$ when $n\geq n_0$ large enough.
 Note $C_nh_n(x)$ is also the solution of (\ref{tezhengfangcheng}),
Titchmarsh(\cite{T1,T2}) shows  $ C_n\sim X_n^{\frac{\l-1}2}$ and thus finish the proof.\qed
\begin{remark}\label{n0defi}
For (\ref{forremark1}) we let  $R$ large enough.   It follows when $n\geq n_0$,
$\lambda_n\geq \lambda_{n_0}\geq V(R)\geq c_*>0$.
\end{remark}
\noindent Proof of Lemma \ref{qQesti}.   From Remark \ref{n0defi} we obtain $X_n\ge R$, then we discuss it under 3 cases.\\ \indent
Case 1. When $x\ge X_n$, by  \eqref{potecd1} one obtains
\[
V(x)-\lambda_n=V'(\xi)(x-X_n)
\ge V'(X_n)(x-X_n)\ge\frac{V(X_n)}{X_n}(x-X_n)\ge D_1X_n^{2\l-1}(x-X_n),\quad\xi\in(X_n,x).
\]
From a straightforward integral estimation, we obtain the second one  in \eqref{Qesti}.\\ \indent
Case 2. When $\frac{X_n}2\le x<X_n$, we have
\[
\lambda_n-V(x)=V'(\xi)(X_n-x)\le V'(X_n)(X_n-x)\le\frac{C_1V(X_n)}{X_n}(X_n-x)\le C_1D_2X_n^{2\l-1}(X_n-x),\quad\xi\in(x,X_n).
\]
Similarly, combining \eqref{potecd1} with Lemma \ref{potesim}  we have $V(\frac{X_n}2)\ge 2^{-C_1}V(X_n)$, then
\[
\lambda_n-V(x)\ge V'({\scs\frac{X_n}2})(X_n-x)\ge\frac{2V({\scs\frac{X_n}2})}{X_n}(X_n-x)\ge2^{1-C_1}D_1X_n^{2\l-1}(X_n-x).
\]
\indent
Case 3. When $0\le x<\frac{X_n}2$, then $\frac{X_n}2\le X_n-x\le X_n$. We discuss it under 2 subcases as follows.\\
Subcase 3.1: $X_n\ge2R$.  By \eqref{potecd2} and Lemma \ref{potesim} we have $|V(x)|\le V({\scs\frac{X_n}2})\le\frac12\lambda_n$, then
\begin{align*}
&\lambda_n-V(x)\ge\frac12\lambda_n\ge\frac{D_1}2X_n^{2\l}\ge\frac{D_1}2X_n^{2\l-1}(X_n-x),&&\forall~x\in[0,{\scs\frac{X_n}2}),\\
&\lambda_n-V(x)\le\frac32\lambda_n\le\frac{3D_2}2X_n^{2\l}\le3D_2X_n^{2\l-1}(X_n-x),&&\forall~x\in[0,{\scs\frac{X_n}2}).
\end{align*}
Subcase 3.2: $R\le X_n<2R$. Denote $\cal N=\{n\in\Z_+:~R\le X_n<2R\}$, then $|\cal N|<\infty$. From \eqref{potecd2} one has $|V(x)|\le\lambda_n$, then
$
\lambda_n-V(x)\le2\lambda_n\le4D_2X_n^{2\l-1}(X_n-x),\quad\forall~x\in[0,{\scs\frac{X_n}2}).
$
On the other side, consider the function $\lambda_n-V(x)$ on $x\in [0,{\scs\frac{X_n}2}]$. Since $\frac{X_n}2<R$, then from \eqref{potecd2} one has $\lambda_n-V(x)>\lambda_n-V(R)\ge0$. By continuity of $V(x)$ there exists a
positive $c_n>0$ such that
$
\lambda_n-V(x)\ge c_n$  for any $x\in[0,~{\scs\frac{X_n}2}]$.
Define $c_0:=\min\limits_{n\in\cal N}\{c_n\}$, then for any $n\in\cal N$ one obtains
\[
\lambda_n-V(x)\ge\frac{c_0}{X_n^{2\l-1}(X_n-x)}X_n^{2\l-1}(X_n-x)\ge\frac{c_0}{(2R)^{2\l}}X_n^{2\l-1}(X_n-x),
\quad\forall~x\in[0,{\scs\frac{X_n}2}).
\]
Combining all the above cases we obtain the first estimate  in \eqref{qesti}. The rest is clear. \qed\\
Proof of Lemma \ref{xmxnsimple}. Since $X_m+X_m^{-\frac13}\le2X_m<\frac12X_n\le X_n-X_n^{-\frac13}$, then we split the integral into five parts as:
\[
\int_{X_m}^{X_n}f(x)e^{{\rm i}kx}h_m(x)\overline{h_n(x)}dx=
\Big(\int_{X_m}^{X_m+X_m^{-\frac13}}+\int_{X_m+X_m^{-\frac13}}^{2X_m}+\int_{2 X_m}^{\frac{X_n}2}+\int_{\frac{X_n}2}^{X_n-X_n^{-\frac13}}+\int_{X_n-X_n^{-\frac13}}^{X_n}\Big)dx.
\]
By Lemma \ref{qQesti} we have $\lambda_n-V({\scs X_m+X_m^{-\frac13}})\ge\lambda_n-V({\scs\frac{X_n}2})\ge\frac{a_1}2X_n^{2\l}$, then
\begin{align*}
&\Big|\int^{X_m+X_m^{-\frac{1}{3}}}_{X_m} \cal F(x) dx\Big| \le CX_m^{\frac{\l-1}2}X_n^{\frac{\l-1}2}X_m^\mu\int^{X_m+X_m^{-\frac{1}{3}}}_{X_m}(V(x)-\lambda_m)^{-\frac{1}{4}} (\lambda_n-V(x))^{-\frac{1}{4}}dx\\
\le~& CX_m^{-\frac{1}{4}+\mu}X_n^{\frac{\l-1}2}(\lambda_n-V(X_m+X_m^{-\frac{1}{3}}))^{-\frac{1}{4}} \int_{X_m}^{X_m+X_m^{-\frac{1}{3}}}(x-X_m)^{-\frac{1}{4}}dx\le C(X_mX_n)^{\frac\mu2-\frac12}.
\end{align*}
Then we turn to the second part.
By Lemma \ref{qQesti} we have
$
\lambda_n-V(2X_m)\ge\lambda_n-V({\scs\frac{X_n}2})\ge\frac{a_1}2X_n^{2\l}$ and
$|\zeta_m(x)|\ge A_1X_m^{\l-\frac23}(x-X_m)\ge A_1X_m^{\l-1}$ for $x\ge X_m+X_m^{-\frac13}$.
Therefore,
\begin{align*}
&\Big|\int_{X_m+X_m^{-\frac{1}{3}}}^{2X_m} \cal F(x) dx\Big| \le  CX_m^{\frac{\l-1}2}X_n^{\frac{\l-1}2}X_m^\mu\int_{X_m+X_m^{-\frac{1}{3}}}^{2X_m} (V(x)-\lambda_m)^{-\frac{1}{4}} (\lambda_n-V(x))^{-\frac{1}{4}}e^{{\rm i}\zeta_m}dx\\
\le&~CX_m^{-\frac{1}{4}+\mu}X_n^{\frac{\l-1}2}\Big(\lambda_n-V(2X_m)\Big)^{-\frac{1}{4}} e^{-C_0X_m^{\l-1}}\int_{X_m+X_m^{-\frac{1}{3}}}^{2X_m} (x-X_m)^{-\frac{1}{4}} e^{-C_0(x-X_m)}dx\\
\le&~CX_m^{-\frac12+\mu}X_n^{-\frac12}\int_0^\infty t^{-\frac{1}{4}}e^{-t}dt\le C (X_mX_n)^{\frac\mu2-\frac12}.
\end{align*}
For the third part, from $x\ge2X_m$ we have $\frac x2\le x-X_m\le x$ and
$
|\zeta_m(x)|\ge A_1X_m^{\l}(x-X_m)\ge A_1X_m^{\l+1}.
$
It follows
\begin{align*}
&\Big|\int_{2X_m}^{\frac{X_n}2} \cal F(x) dx\Big| \le CX_m^{\frac{\l-1}2}X_n^{\frac{\l-1}2}\int_{2X_m}^{\frac{X_n}2} x^\mu(V(x)-\lambda_m)^{-\frac{1}{4}} (\lambda_n-V(x))^{-\frac{1}{4}}e^{{\rm i}\zeta_m}dx\\
\le&~CX_m^{-\frac{1}{4}}X_n^{\frac{\l-1}2}\big(\lambda_n-V({\scs\frac{X_n}2})\big)^{-\frac{1}{4}}e^{-C_0X_m^{\l+1}} \int_{2X_m}^{\frac{X_n}2} (x-X_m)^{-\frac14+\mu} e^{-C_0(x-X_m)}dx\le C(X_mX_n)^{\frac\mu2-\frac12}.
\end{align*} Then we turn to the fourth part, from $x\ge\frac{X_n}2$ we have $x-X_m\ge\frac{X_n}2-\frac{X_n}4=\frac{X_n}4$ and
$
\lambda_n-V({\scs X_n-X_n^{-\frac13}})\ge a_1X_n^{2\l-\frac43}$ and
$|\zeta_m(x)|\ge A_1X_m^\l(x-X_m)\ge\frac{A_1}4X_m^{\l}X_n$. It follows
\begin{align*}
&\Big|\int_{\frac{X_n}2}^{X_n-X_n^{-\frac13}}  \cal F(x) dx\Big| \le CX_m^{\frac{\l-1}2}X_n^{\frac{\l-1}2+\mu}\int_{\frac{X_n}2}^{X_n-X_n^{-\frac13}}(V(x)-\lambda_m)^{-\frac14} (\lambda_n-V(x))^{-\frac{1}{4}}e^{{\rm i}\zeta_m}dx\\
\le&~C X_m^{-\frac{1}{4}}X_n^{\frac{\l-1}2+\mu}\big(\lambda_n-V({\scs X_n-X_n^{-\frac13}})\big)^{-\frac{1}{4}}e^{-C_0X_n} \int_{\frac{X_n}2}^{X_n-X_n^{-\frac13}} (x-X_m)^{-\frac14} e^{-C_0(x-X_m)}dx\\
\le&~CX_m^{-\frac12}X_n^{-\frac12}\int_0^\infty t^{-\frac14}e^{-t}dt\le C(X_mX_n)^{\frac\mu2-\frac12}.
\end{align*}
For the last part, from $X_n-X_n^{-\frac13}\ge \frac12X_n\ge2X_m$, we have
$
V(X_n-X_n^{-\frac13})-\lambda_m\ge V(2X_m)-V(X_m)\ge a_1X_m^{2\l}$ and
$|\zeta_m(x)|\ge\frac{A_1}4X_m^\l X_n$. It follows
\begin{align*}
&\Big|\int_{X_n-X_n^{-\frac{1}{3}}}^{X_n} \cal F(x) dx\Big| \le CX_m^{\frac{\l-1}2}X_n^{\frac{\l-1}2}\int_{X_n-X_n^{-\frac{1}{3}}}^{X_n}x^\mu (V(x)-\lambda_m)^{-\frac{1}{4}}(\lambda_n-V(x))^{-\frac{1}{4}}e^{-|\zeta_m(x)|}dx\\
\le&~CX_m^{\frac{\l-1}2}X_n^{-\frac14+\mu}e^{-C_0X_n}\Big(V\big(X_n-X_n^{-\frac{1}{3}}\big)-\lambda_m\Big)^{-\frac{1}{4}} \int_{X_n-X_n^{-\frac{1}{3}}}^{X_n} (X_n-x)^{-\frac{1}{4}}dx\\
\le&~CX_m^{-\frac12}X_n^{-\frac14}\int_{X_n-X_n^{-\frac{1}{3}}}^{X_n} (X_n-x)^{-\frac{1}{4}}dx
\le C(X_mX_n)^{\frac\mu2-\frac12}.
\end{align*}
Similarly, we obtain
\[
\Big|\int_{X_m}^{X_n}f(x) e^{{\rm i}kx}\psi_{j_1}^{(m)}(x) \overline{\psi_{j_2}^{(n)}(x)} dx \Big|
\le C(X_mX_n)^{\frac\mu{2}-\frac12},~\text{for}~j_1,j_2\in\{1,2\}~\text{and}~j_1+j_2\ge3.
\]
Hence, we have
\(\dss
\Big|\int_{X_m}^{X_n}f(x)e^{{\rm i}kx}h_m(x)\overline{h_n(x)}dx\Big|\le C(X_mX_n)^{\frac\mu2-\frac12}.
\)\qed \\
Proof of Lemma \ref{xmcomplex2}.  Let $b=X_m-{\txs\frac{a_1}{16a_2}}X_m^{\frac13}$, then
$$
\Big|\int_{X_m-X_m^{\frac13}}^{b}\cal F(x) dx\Big|\le C X_m^{\mu+\l-1} \int_{X_m-X_m^{\frac13}}^{b}(\lambda_m-V(x))^{-\frac{1}{2}}dx
\le C X_m^{\mu-\frac13}.
$$
For the remainder integral on $[X_m^\frac23, b]$, since
\[
\lambda_m-V(b)\le a_2X_m^{2\l-1}(X_m-b)\le\frac{a_1}{16}X_m^{2\l-\frac23},\quad
\lambda_m-V(X_m-X_m^{\frac23})\ge a_1X_m^{2\l-\frac13},
\]
then
\begin{align*}
g(b)&\ge \frac{\sqrt{a_1}kX_m^{\l-\frac13}}{\sqrt{\lambda_n-V(b)}+\sqrt{\lambda_m-V(b)}}-k\\
&\ge \frac{\sqrt{a_1}kX_m^{\l-\frac13}}{\sqrt{\lambda_n-\lambda_m+\frac{a_1}{16}X_m^{2\l-\frac23}}+\sqrt{\frac{a_1}{16}X_m^{2\l-\frac23}}}-k \ge\frac{kX_m^{\l-\frac13}}{\frac{X_m^{\l-\frac13}}{2}+\frac{X_m^{\l-\frac13}}{4}}-k=\frac k3,\\
g({\scs X_m-X_m^\frac23})
&\le\frac{\sqrt{a_1}kX_m^{\l-\frac16}}{2\sqrt{\lambda_m-V(X_m-X_m^\frac23)}}-k
\le\frac{k\sqrt{a_1}X_m^{\l-\frac16}}{2\sqrt{a_1}X_m^{\l-\frac16}}-k=-\frac k2.
\end{align*}
Let $g(a)=-kX_m^{-\frac13}$, by the monotonicity of $g(x)$ we have $X_m-X_m^\frac23<a<b$ as the figure \ref{case2} below.
In the following we first estimate the integral on $[a,b]$. Note that $X_m\ge8$, then $g(a)\ge-\frac k2$. Since
\[
\lambda_m-V(a)\le a_2X_m^{2\l-1}(X_m-a)\le a_2X_m^{2\l-\frac13}
\]
and $g''(x)>0$, then
\[
g'(x)\ge g'(a)=\frac{V'(a)(g(a)+k)}{2\sqrt{\lambda_n-V(a)}\sqrt{\lambda_m-V(a)}}
\ge \frac{C k X_m^{2\l-1}}{X_m^{2\l-\frac13}} \ge C kX_m^{-\frac23},\quad\text{for}~x\in[a,b].
\]
By Lemma \ref{oscint}, ones obtain
\[
\Big|\int_a^{b}f(x)e^{\rmi(\zeta_m-\zeta_n+kx)}\Psi(x)dx\Big|
\le Ck^{-\frac{1}{2}}X_m^{\frac13}\bigg(X_m^\mu\Big(|\Psi(b)|+\int_a^{b}|\Psi'(x)|dx\Big)+
\int_a^{b}x^{\mu-1}|\Psi(x)|dx\bigg).
\]
By Corollary \ref{psiesti} we have
\(
|\Psi(b)|\le C X_m^{-\l+\frac13}
\) and
$
\int_a^{b}x^{\mu-1}|\Psi(x)|dx\le C X_m^{\mu-\l+\frac13}.
$
Besides, one has
$
\int_a^{b}J_1dx \le C \int_a^{b} x^{2\l-1}(\lambda_m-V(x))^{-\frac{5}{4}}(\lambda_m-V(x))^{-\frac{1}{4}}dx \le CX_m^{-\l+\frac13}$
and
$\int_a^{b}J_3dx \le C X_m^{-2\l+\frac13}$.
It follows that
$
\Big|\int_a^{b} \cal F(x) dx\Big| \le Ck^{-\frac{1}{2}}X_m^{\mu-\frac13}.
$
Next we estimate the integral on $[X_m^\frac23,a]$, where $\left|g(x)\right| \ge kX_m^{-\frac13}$. By Lemma \ref{oscint} one has
\[
\Big|\int_{X_m^\frac23}^af(x) e^{{\rm i}(\zeta_m-\zeta_n+kx)}\Psi(x)dx\Big|
\le Ck^{-1}X_m^{\frac13}\bigg(X_m^\mu\Big(|\Psi(a)| + \int_{X_m^\frac23}^a |\Psi'(x)|dx\Big)+\int_{X_m^\frac23}^ax^{\mu-1}|\Psi(x)|dx\bigg).
\]
Similarly, we have
$
\Big|\int_{X_m^\frac23}^a \cal F(x) dx\Big| \le C k^{-1}X_m^{\mu-\frac13}.
$
It follows that
\[
\Big|\int_{X_m^\frac23}^{X_m-X_m^{\frac{1}{3}}} \cal F(x) dx\Big| \le C(k^{-1}\vee 1)(X_mX_n)^{\frac\mu2-\frac1{6}}.
\]
From a straightforward computation, we obtain
\[
\Big|\int_{X_m^\frac23}^{X_m-X_m^\frac{1}{3}}f(x) e^{{\rm i}kx}\psi^{(m)}_{j_1}(x)\overline{\psi^{(n)}_{j_2}(x)}dx\Big|\le C(X_mX_n)^{\frac\mu2-\frac\l2-\frac12},
~\text{for}~j_1,j_2\in\{1,2\}~\text{and}~j_1+j_2\ge3.
\]
Hence, we have
\(\dss
\Big|\int_{X_m^\frac23}^{X_m-X_m^{\frac{1}{3}}}f(x)e^{{\rm i}kx}h_m(x)\overline{h_n(x)}dx\Big| \le C (k^{-1}\vee 1)(X_mX_n)^{\frac\mu2-\frac16}.
\)\qed
\begin{figure}[H]
\begin{minipage}{0.49\textwidth}
\centering
\epsfig{file=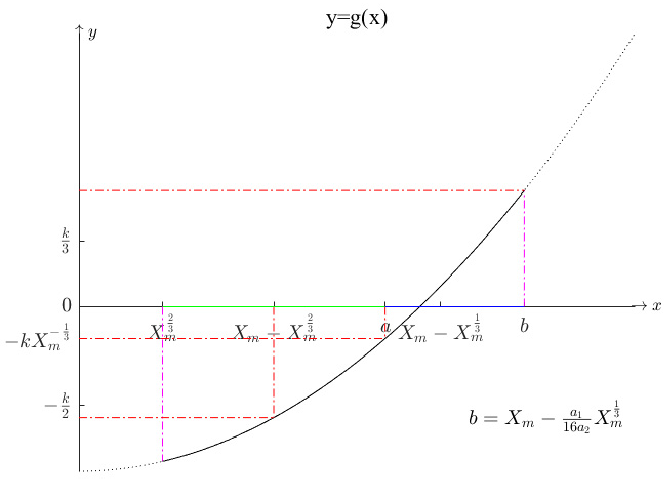,width=\textwidth}
\caption{Phase in Lemma \ref{xmcomplex2}}\label{case2}
\end{minipage}\hfill
\begin{minipage}{0.49\textwidth}
\centering
\epsfig{file=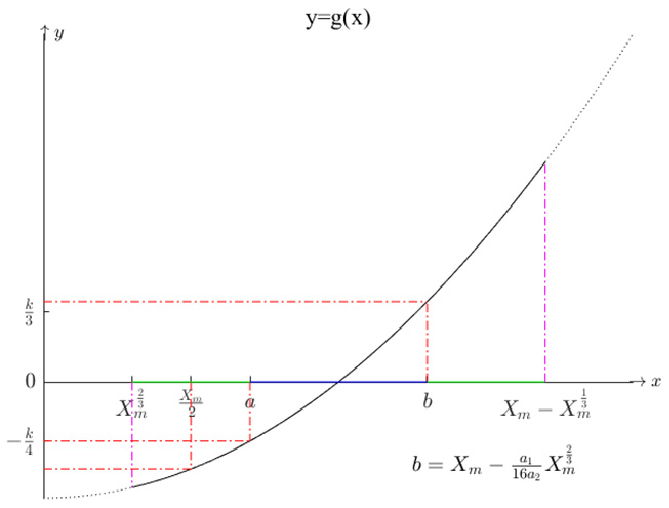,width=\textwidth}
\caption{Phase in Lemma \ref{xmcomplex3}}\label{case3}
\end{minipage}
\end{figure}
Proof of Lemma \ref{xmcomplex3}. Let $b=X_m-{\txs\frac{a_1}{16a_2}}X_m^\frac23$. By Lemma \ref{qQesti} we have
$
\lambda_m-V(b)\le\frac{a_1}{16}X_m^{2\l-\frac13}$,
$\frac{a_1}2X_m^{2\l}\leq \lambda_m-V({\scs\frac{X_m}2})\le\frac{a_2}2X_m^{2\l}$. Similar as before, we have
$g(b)\ge\frac k3$ and
$g({\scs\frac{X_m}{2}})\le-\frac k4$.
Let $g(a)=-\frac k4$, then $\frac{X_m}{2}<a<b$
as the figure \ref{case3} above. We first estimate the integral on $[a,b]$. Since
\[
g'(a)=\frac{V'(a)(g(a)+k)}{2\sqrt{\lambda_n-V(a)}\sqrt{\lambda_m-V(a)}}
\ge\frac{CkX_m^{2\l-1}}{X_m^{2\l}}\ge CkX_m^{-1},
\]
then $g'(x)\ge CkX_m^{-1}$, for $x\in(a,b)$. By Lemma \ref{oscint} we have
\[
\Big|\int_a^{b}f(x)e^{{\rm i}(\zeta_m-\zeta_n+kx)}\Psi(x)dx\Big|
\le Ck^{-\frac{1}{2}}X_m^{\frac12}\bigg(X_m^\mu\Big(|\Psi(b)| + \int_a^{b} |\Psi'(x)|dx\Big)
+\int_a^{b} x^{\mu-1}|\Psi(x)|dx\bigg).
\]
By corollary \ref{psiesti} we have
$
|\Psi(b)|\le CX_m^{-\l+\frac16}
$ and
$
\int_a^{b}x^{\mu-1}|\Psi(x)|dx\le CX_m^{\mu-\l+\frac16}.
$
Besides,
$
\int_a^{b}J_1dx  \le CX_m^{-\l+\frac16}$ and
$\int_a^{b}J_3dx \le CX_m^{-\l+\frac16}$.
It follows that
$
\Big|\int_a^{b} \cal F(x) dx\Big| \le Ck^{-\frac{1}{2}}X_m^{\mu-\frac13}.
$
Next we estimate the integral on $[b,X_m-X_m^{\frac13}]$. From $\left|g(x)\right| \ge \frac k3$ and  Lemma \ref{oscint},
we have
\begin{align*}
&\Big|\int_{b}^{X_m-X_m^{\frac13}}f(x)e^{{\rm i}(\zeta_m-\zeta_n+kx)}\Psi(x)dx\Big| \notag \\
&\le~Ck^{-1}\bigg(X_m^{\mu}\Big(\big|\Psi({\scs X_m-X_m^{\frac13}})\big| + \int_{b}^{X_m-X_m^{\frac13}} |\Psi'(x) |dx\Big)
+\int_{b}^{X_m-X_m^{\frac13}}x^{\mu-1}|\Psi(x)|dx\bigg)\\
&\le Ck^{-1}X_m^{\mu-\frac23}.
\end{align*}
Similarly,
$
\Big|\int_{a}^{X_m^{\frac23}}\cal F(x) dx\Big| \le Ck^{-1}X_m^{\mu-\frac56}.
$
It follows that
$
\Big|\int_{X_m^\frac23}^{X_m-X_m^{\frac{1}{3}}} \cal F(x) dx\Big| \le C(k^{-1}\vee 1)(X_mX_n)^{\frac\mu2-\frac{1}{6}}.
$
Since the other three integrals have better estimates, we finish the proof. \qed\\
Proof of Lemma \ref{xmcomplex4}. Let $a=X_m^{\frac{2\mu+2\l}{2\mu+2\l+1}},~b=(1-{\scs\frac{a_1}{16a_2}})X_m$, then by Lemma \ref{qQesti} we have
\(\dss
\lambda_m-V(b)\le a_2X_m^{2\l-1}(X_m-b)=\frac{a_1}{16}X_m^{2\l}.
\) Hence,
$
g(b)\ge\frac{\sqrt{a_1}kX_m^\l}{\sqrt{\lambda_n-\lambda_m+\frac{a_1X_m^{2\l}}{16}}+\sqrt{\frac{a_1X_m^{2\l}}{16}}}-k
\ge\frac k3
$
as the  figure \ref{case4} below. We first estimate the integral on $[b,X_m-X_m^\frac13]$. By Lemma \ref{oscint} one has
\begin{align*}
&\Big|\int_{b}^{X_m-X_m^\frac13}f(x)e^{{\rm i}(\zeta_m-\zeta_n+kx)}\Psi(x)dx\Big|  \\
\le&~Ck^{-1}\bigg(X_m^\mu\Big(\big|\Psi({\scs X_m-X_m^\frac13})\big| + \int_b^{X_m-X_m^\frac13} |\Psi'(x)|dx\Big)
+\int_b^{X_m-X_m^\frac13}x^{\mu-1}|\Psi(x)|dx\bigg).
\end{align*}
From similar computations, we obtain
$
\Big|\int_{b}^{X_m-X_m^\frac13} \cal F(x) dx\Big| \le Ck^{-1}X_m^{\mu-\frac23}.
$
Next we estimate the integral on $[a,b]$. For $x\in[a,b]$, we have
$
\lambda_m-V(x)\le\lambda_m-V(a)\le  a_2X_m^{2\l},
$ then
$
g'(x)\ge \frac{C X_m^{\frac{(2\mu+2\l)(2\l-1)}{2\mu+2\l+1}} \cdot kX_m^\l}{X_m^{3\l}}\ge CkX_m^{-\frac{2\mu+4\l}{2\mu+2\l+1}}.
$
By Lemma \ref{oscint}, one has
\[
\Big|\int_{a}^{b}f(x)e^{{\rm i}(\zeta_m-\zeta_n+kx)}\Psi(x)dx\Big|
\le Ck^{-\frac12}X_m^{\frac{\mu+2\l}{2\mu+2\l+1}}\bigg(X_m^\mu\Big(|\Psi(b)| + \int_{a}^{b}|\Psi'(x) |dx\Big)
+\int_{a}^{b}x^{\mu-1}|\Psi(x)|dx\bigg).
\]
By similar procedures, we have
$
\Big|\int_{a}^{b}\cal F(x) dx\Big| \le Ck^{-\frac12}X_m^{\mu-\frac{\mu+1}{2\mu+2\l+1}}.
$
From a straightforward computation, we obtain
$
\Big|\int_{X_m^\frac23}^{a} \cal F(x) dx\Big| \le CX_m^{\mu-\frac{\mu+1}{2\mu+2\l+1}}.
$
It follows that
\[
\Big|\int_{X_m^\frac23}^{X_m-X_m^{\frac{1}{3}}} \cal F(x) dx\Big| \le C(k^{-1}\vee 1)(X_mX_n)^{\frac\mu2-\frac{\mu+1}{2(2\mu+2\l+1)}}.
\]
Since the other three integrals have better estimates, we finish the proof. \qed
\begin{figure}[H]
\begin{minipage}{0.49\textwidth}
\centering
\epsfig{file=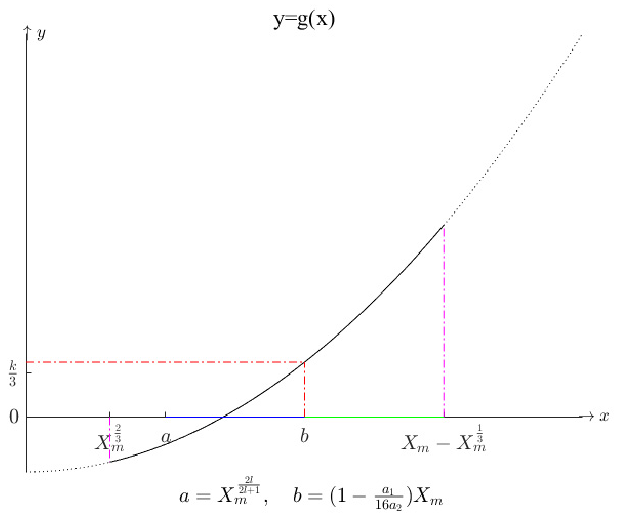,width=\textwidth}
\caption{Phase in Lemma \ref{xmcomplex4}}\label{case4}
\end{minipage}\hfill
\begin{minipage}{0.49\textwidth}
\centering
\epsfig{file=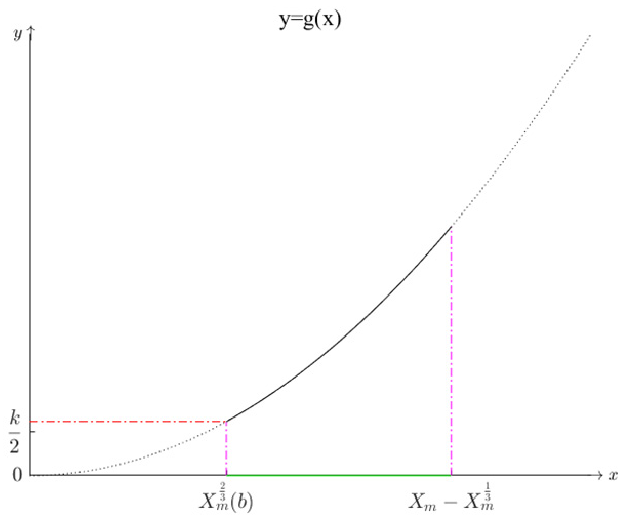,width=\textwidth}
\caption{Phase in Lemma \ref{xmcomplex5}}\label{case5}
\end{minipage}
\end{figure}
Proof of  Lemma \ref{xmcomplex5}. Since
$\sqrt{\lambda_n}+\sqrt{\lambda_m}\le2\sqrt{\lambda_n}\le2\sqrt{D_2}X_n^\l$, then for $x\in[X_m^\frac23,X_m-X_m^\frac13]$,
\[
g(x)=\frac{\lambda_n-\lambda_m}{\sqrt{\lambda_n-V(x)}+\sqrt{\lambda_m-V(x)}}-k
\ge\frac{3\sqrt{D_2}kX_n^\l}{\sqrt{\lambda_n}+\sqrt{\lambda_m}}-k
\ge\frac{3\sqrt{D_2}k X_n^\l}{2\sqrt{D_2}X_n^\l}-k=\frac k2,
\]
as the figure \ref{case5} above. By Lemma \ref{oscint}, one has
\begin{align*}
&\Big| \int_{X_m^\frac23}^{X_m-X_m^\frac{1}{3}} e^{{\rm i}(\zeta_m-\zeta_n+kx)}\Psi(x)dx\Big|\\
\le&~\frac{C}{k}\bigg(X_m^\mu\Big(|\Psi({\scs X_m-X_m^{\frac{1}{3}}}) | + \int_{X_m^\frac23}^{X_m-X_m^\frac{1}{3}} |\Psi'(x)|dx\Big)
+\int_{X_m^\frac23}^{X_m-X_m^\frac{1}{3}} |\Psi(x)|x^{\mu-1}dx\bigg).
\end{align*}
The remaining proof is similar  as Lemma \ref{xmcomplex1}.\qed

\begin{lemma}\label{potesim}
Assume $V(x)$ satisfies Assumption \ref{poteass} and $\theta\in(0,1)$, there exists a positive constant $\widetilde R\geq R_0$ such that $\theta^{C_1}V(x)\le V(\theta x)\le \theta V(x)$ for $\theta x\ge\widetilde R$.
\end{lemma}
\begin{proof}
By  Assumption \ref{poteass}, we have $\lim\limits_{x\to+\infty}V'(x)=+\infty$. Define
\begin{equation*}
\widetilde R:=\min\Big\{x\in[R_0,+\infty):V'(x)\ge\frac{V(R_0)}{R_0}\Big\},
\end{equation*}
then by \eqref{convex} in Assumption \ref{poteass}, one has
\[
V(x)=V(R_0)+\int_{R_0}^xV'(t)dt\le V(R_0)+V'(x)(x-R_0)\le V'(x)x,\quad\forall~x\ge \widetilde R.
\]
Together with \eqref{decay} in  Assumption \ref{poteass}, we have
$
\frac1x\le\frac{V'(x)}{V(x)}\le\frac{C_1}x$
for all $x\ge \widetilde R$.  Given $\theta\in(0,1)$, if $\theta x\ge \widetilde R$, then
$
\log\frac{V(\theta x)}{V(x)}=\int_x^{\theta x}d\log V(t)=\int_{x}^{\theta x}\frac{V'(t)}{V(t)}dt\ge\int_{x}^{\theta x}\frac{C_1}{t}dt=\log \theta^{C_1}.
$
Hence $V(\theta x)\ge \theta^{C_1}V(x)$. Similarly, we have $V(\theta x)\le\theta V(x)$.
\end{proof}

\begin{lemma}{\rm (\cite{T2})}\label{6.5}
For fixed $\lambda$, if $x>2X$, then
$\int_x^\infty|f(t)(\lambda-V(t))^\frac12|dt\le\frac{C}{x(V(x))^\frac12}$,
here $C$ is a constant independent of $x$ and $\lambda$.
\end{lemma}
\begin{lemma}{\rm (\cite{T2})}\label{2.2sub}
$\int_0^\infty |f(x)||\lambda-V(x)|^\frac12dx=O\bigg(\frac{1}{X\lambda^\frac12}\bigg), \quad \lambda\to\infty.$
\end{lemma}

The following lemma is from {\cite{LiangLuo2019}}.
\begin{lemma}\label{Bessel}
Bessel function of third kind $\belh(z)$ satisfies the following:
\begin{alignat*}{2}
\Bgs{\sbelhz{z}}&\le  1,&\quad& z\in(-\infty,-c_1),\\
\Bgs{\sbelhz{z}}&\le  \frac{20}{d_1} |z|^{\frac{1}{6}},&& z\in[-c_2,0),\\
\Bgs{\sbelhz{z}}&\le  \frac{Ce^{c_3}}{d_1} \max\{|z|^{\frac{1}{6}}, |z|^{\frac56}\},&& z\in (0,c_3]{\rm i},\\
\Bgs{ \sbelhz{z}}&\le  e^{-|z|},&& z\in(c_4,\infty){\rm i},
\end{alignat*}
where $c_1>0$, $c_2\in (0,\ 1]$, $c_3, c_4$ can be arbitrary positive numbers and $C$ is a positive constant.
\end{lemma}

The last lemma is  from \cite{Stein}.
\begin{lemma}\label{oscint}
Suppose $\phi$ is real-valued and smooth in $(a,b)$, $\psi$ is complex-valued, and that $|\phi^{(k)}(x)|\ge1$ for all $x\in(a,b)$. Then
\[
\Big|\int_a^b e^{\rm i\lambda\phi(x)}\psi(x)dx\Big|\le c_k\lambda^{-1/k}\Big(|\psi(b)|+\int_a^b|\psi'(x)|dx\Big).
\]
holds when:\\
(i) $k\ge 2$, or (ii) $k=1$ and $\phi'(x)$ is monotonic.\\
The bound $c_k$ is independent of $\phi$, $\psi$ and $\lambda$.
\end{lemma}

\end{document}